\pdfoutput=1
\documentclass[journal,twocolumn]{IEEEtran}

\usepackage{hyperref}
\usepackage{amsmath,amsfonts,amsthm,amssymb,mathtools,thmtools,amssymb} 
\usepackage{amsbsy}
\usepackage{bm}
\usepackage{cite}
\usepackage[percent]{overpic}
\usepackage{soul}
\usepackage{relsize}
 \usepackage{subfig}
 \usepackage[font=small]{caption}
 \usepackage{color} 
\usepackage{epstopdf}
\usepackage{cancel}
\usepackage{dsfont}
\usepackage{tikz}
\usepackage{epsf}
\usepackage[export]{adjustbox}
\usepackage{algorithm}
\usepackage{algpseudocode}
\usepackage{xpatch}

\algrenewcommand\alglinenumber[1]{{\sffamily\footnotesize#1}}
\algnewcommand{\Initialize}[1]{%
  \State \textbf{Initialize:}
  \Statex \hspace*{\algorithmicindent}\parbox[t]{.8\linewidth}{\raggedright #1}
}

\makeatletter
\xpatchcmd{\algorithmic}{\itemsep\z@}{\itemsep=0.15ex plus1pt}{}{}
\makeatother

\usetikzlibrary{calc,arrows}
\tikzstyle{line}=[draw]
\tikzstyle{arrow}=[draw, -latex]

\DeclarePairedDelimiter{\ceil}{\lceil}{\rceil}
\DeclareMathOperator*{\argmax}{arg\,max}
\newcommand*{\Scale}[2][4]{\scalebox{#1}{$#2$}}%

\newenvironment{brsm}{
  \bigl[ \begin{smallmatrix} }{%
  \end{smallmatrix} \bigr]}

\newtheorem{theorem}{Theorem}

\newtheorem{lemma}{Lemma}

\newtheorem{proposition}{Proposition}
\newtheorem{remark}{Remark}
\theoremstyle{definition}
\theoremstyle{remark}

\newcommand\Ccancel[2][black]{\renewcommand\CancelColor{\color{#1}}\xcancel{#2}}

\usepackage{geometry}
\begin{document}

\title{Polar Coding for Common Message Only Wiretap Broadcast Channel}


\author{Jaume~del~Olmo~Alos, and
  Javier~R.~Fonollosa,
  \thanks{This work is supported by ``Ministerio de Ciencia, Innovación y Universidades'' of the Spanish Government, TEC2015-69648-REDC and TEC2016-75067-C4-2-R AEI/FEDER, UE, and the Catalan Government, 2017
SGR 578 AGAUR.
  
  The authors are with the Department of Teoria del Senyal i Comunicacions, Universitat Politècnica de Catalunya, 08034, Barcelona, Spain
  (e-mail: jaume.del.olmo@upc.edu; javier.fonollosa@upc.edu)
  
    }}

\maketitle

\begin{abstract}
A polar coding scheme is proposed for the Wiretap Broadcast Channel with two legitimate receivers and one eavesdropper. We consider a model in which the transmitter wishes to send a private and a confidential message that must be reliably decoded by the receivers, and the confidential message must also be (strongly) secured from the eavesdropper. The coding scheme aims to use the optimal rate of randomness and does not make any assumption regarding the symmetry or degradedness of the channel. This paper extends previous work on polar codes for the wiretap channel by proposing a new chaining construction that allows to reliably and securely send the same confidential message to two different receivers. This construction introduces new dependencies between the random variables involved in the coding scheme that need to be considered in the secrecy analysis.
\end{abstract}




\IEEEpeerreviewmaketitle

\section{Introduction}\label{Section:Introduction}
Information-theoretic security over noisy channels was introduced by Wyner in \cite{6772207}, which characterized the secrecy-capacity of the degraded wiretap channel. Later, Csisz\'{a}r and K{\"o}rner in \cite{1055892} generalized Wyner's results to the general wiretap channel. In these settings, one transmitter wishes to reliably send one message to a legitimate receiver, while keeping it secret from an eavesdropper, where secrecy is defined based on a condition on some information-theoretic measure that is fully quantifiable. One of these measures is the \emph{information leakage}, defined as the mutual information $I(W;Z^n)$ between a uniformly distributed random message $W$ and the channel observations $Z^n$ at the eavesdropper, $n$ being the number of uses of the channel. Based on this measure, the most common secrecy conditions required to be satisfied by channel codes are the \emph{weak secrecy}, which requires $\lim_{n \rightarrow \infty} \frac{1}{n} I(W;Z^n) = 0$, and the \emph{strong secrecy}, requiring $\lim_{n \rightarrow \infty} I(W;Z^n) = 0$. Although the second notion of security is stronger, surprisingly both conditions result in the same secrecy-capacity \cite{maurer2000information}. 

In the last decade, information-theoretic security has been extended to a large variety of contexts, and polar codes have become increasingly popular in this area due to their easily provable secrecy capacity achieving property. Polar codes were originally proposed by Arikan in \cite{arikan2009channel} to achieve the capacity of binary-input, symmetric, point-to-point channels under Successive Cancellation (SC) decoding.  Secrecy capacity achieving polar codes for the binary symmetric degraded wiretap channel were introduced in \cite{6034749} and \cite{6620400}, satisfying the weak and the strong secrecy condition, respectively. Recently, polar coding has been extended to the general wiretap channel in \cite{renes2013efficient, 7346401,cihad2014achieving,7426806} and to different multiuser scenarios (for instance, see \cite{e20060467} and \cite{8438560}). Indeed, \cite{cihad2014achieving} and \cite{7426806} generalize their results providing polar codes for the broadcast channel with confidential messages.

This paper provides a polar coding scheme that allows to transmit \emph{strongly} confidential common information to two legitimate receivers over the Wiretap Broadcast Channel (WBC). Although \cite{chia2012three} provided an obvious lower-bound on the secrecy-capacity of this model, no constructive polar coding scheme has already been proposed so far. Our polar coding scheme is based mainly on the one introduced by \cite{7426806} for the broadcast channel with confidential messages. Therefore, the proposed polar coding scheme aims to use the optimal amount of randomness in the encoding. Moreover, in order to construct an explicit polar coding scheme that provides strong secrecy, the distribution induced by the encoder must be close in terms of statistical distance to the original one considered for the code construction, and transmitter and legitimate receivers need to share a secret key of negligible size in terms of rate. Nevertheless, the particularization for the model proposed in this paper is not straightforward. Specifically, we propose a new chaining construction \cite{6875073} (transmission will take place over several blocks) that is crucial to secretly transmit common information to different legitimate receivers. This construction introduces new dependencies between the random variables that are involved in the polar coding scheme that must be considered carefully in the secrecy analysis. These dependencies are bidirectional between random variables of adjacent blocks and, consequently, we need to introduce an additional privately-shared key of negligible length in terms of rate for the polar coding scheme to provide strong secrecy.

\subsection{Notation}
Through this paper, let the interval $[a,b]$, where $a,b \in \mathbb{Z}_{+}$ and $a \leq b$, denote the set of integers between $a$ and $b$ (both included). Let $u^n$ denote a row vector $(u(1), \dots, u(n))$. We write $u^{1:j}$ for $j \in [1,n]$ to denote the subvector $(u(1),\dots,u(j))$. For any set of indices $\mathcal{S} \subseteq [1,n]$, we write $u[\mathcal{S}]$ to denote the sequence $\{u(j)\}_{j\in \mathcal{S}}$, and we use $\mathcal{S}^{\text{C}}$ to denote the set complement in $[1,n]$, that is, $\mathcal{S}^{\text{C}} = [1,n] \setminus \mathcal{S}$. If $\mathcal{S}$ denotes an event, then $\mathcal{S}^{\text{C}}$ also denotes its complement. Consider some index $i \in [1,L]$, where $L  \in \mathbb{Z}^{+}$, and consider the vector $u_i^n$. We write $u_{1:L}^n$ to denote the set of vectors $\{ u_1^n, \dots, u_L^n \}$. We use $\ln$ to denote the natural logarithm, whereas $\log$ denotes the logarithm base 2. Let $X$ be a random variable taking values in $\mathcal{X}$, and let $q_x$ and $p_x$ be two different distributions with support $\mathcal{X}$, then $\mathbb{D}(q_x,p_x)$ and $\mathbb{V}(q_x,p_x)$ denote the Kullback-Leibler divergence and the total variation distance respectively. Finally, $h_2(p)$ denotes the binary entropy function, i.e., $h_2(p) = -p \log p - (1-p) \log (1-p)$, and we define the indicator function $\mathds{1} \{ u \}$ such that equals to 1 if the predicate $u$ is true and 0 otherwise.  

\subsection{Organization}
The remaining of this paper is organized as follows. Section~\ref{sec:SM} introduces the channel model formally. In Section~\ref{sec:rev}, the fundamental theorems of polar codes are revisited. Section~\ref{sec:PCS} describes the proposed polar coding scheme. In Section~\ref{sec:performance}, the performance of the coding scheme is analyzed. Finally, the concluding remarks are presented in Section~\ref{sec:conclusion}.

\section{Channel Model and Achievable Region}\label{sec:SM}
Formally, a WBC $(\mathcal{X}, p_{Y_{(1)}Y_{(2)} Z|X}, \mathcal{Y}_{(1)} \times \mathcal{Y}_{(2)} \times \mathcal{Z})$ with 2 legitimate receivers and an external eavesdropper is characterized by the probability transition function $p_{Y_{(1)}Y_{(2)}Z|X}$, where $X \in \mathcal{X}$ denotes the channel input, $Y_{(k)} \in \mathcal{Y}_{(k)}$ denotes the channel output corresponding to the legitimate receiver $k \in [1,2]$, and $Z \in \mathcal{Z}$ denotes the channel output corresponding to the eavesdropper. We consider a model, namely \emph{Common Information over the WBC} (CI-WBC), in which the transmitter wishes to send a private message $W$ and a confidential message $S$ to both legitimate receivers. A code $\big(\ceil{2^{nR_W}},\ceil{2^{nR_S}},\ceil{2^{nR_R}},n \big)$ for the CI-WBC consist of a private message set $\mathcal{W} \triangleq  \big[1, \ceil{2^{nR_W}} \big]$, a confidential message set $\mathcal{S} \triangleq  \big[1, \ceil{2^{nR_S}} \big]$, a randomization sequence set $\mathcal{R} \triangleq \big[1, \ceil{2^{nR_R}} \big]$ (typically referred as \emph{local randomness} and needed to confuse the eavesdropper about the confidential message $S$), an encoding function $f: \mathcal{W} \times \mathcal{S} \times \mathcal{R} \rightarrow \mathcal{X}^n$ that maps $(w,s,r)$ to a codeword $x^n$, and two decoding functions $g_{(1)}$ and $g_{(2)}$ such that $g_{(k)}:\mathcal{Y}_{(k)}^n \rightarrow \mathcal{W} \times \mathcal{S}$ ($k \in [1,2]$) maps the $k$-th legitimate receiver observations $y_{(k)}^n$ to the estimates $(\hat{w}_{(k)},\hat{s}_{(k)})$. The reliability condition to be satisfied by this code is measured in terms of the average probability of error and is given by
\begin{IEEEeqnarray}{c}
\lim_{n \rightarrow \infty} \mathbb{P}\left[  (W,S) \neq (\hat{W}_{(k)},\hat{S}_{(k)}) \right] = 0,  \quad k \in [1,2]. \IEEEeqnarraynumspace \label{eq:reliabilitycond}
\end{IEEEeqnarray}
The \emph{strong} secrecy condition is measured in terms of the information leakage and is given by
\begin{IEEEeqnarray}{c}
\lim_{n \rightarrow \infty} I \left( S ; Z^n \right) = 0.  \label{eq:secrecycond}
\end{IEEEeqnarray}
This model is graphically illustrated in Figure~\ref{fig:chmodel}. A triple of rates $(R_W,R_S,R_R) \in \mathbb{R}_{+}^3$ will be achievable for the CI-WBC if there exists a sequence of $(\ceil{2^{nR_W}},\ceil{2^{nR_S}},\ceil{2^{nR_R}},n )$ codes such that satisfy the reliability and secrecy conditions \eqref{eq:reliabilitycond} and \eqref{eq:secrecycond} respectively. 

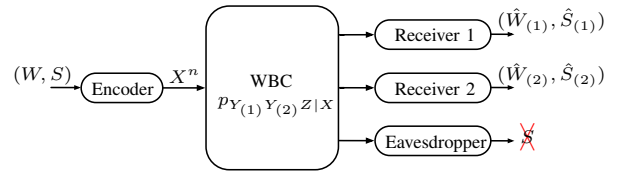
\begin{figure}[h!]
\hspace{0.45cm}
\begin{overpic}[width=0.77\linewidth]{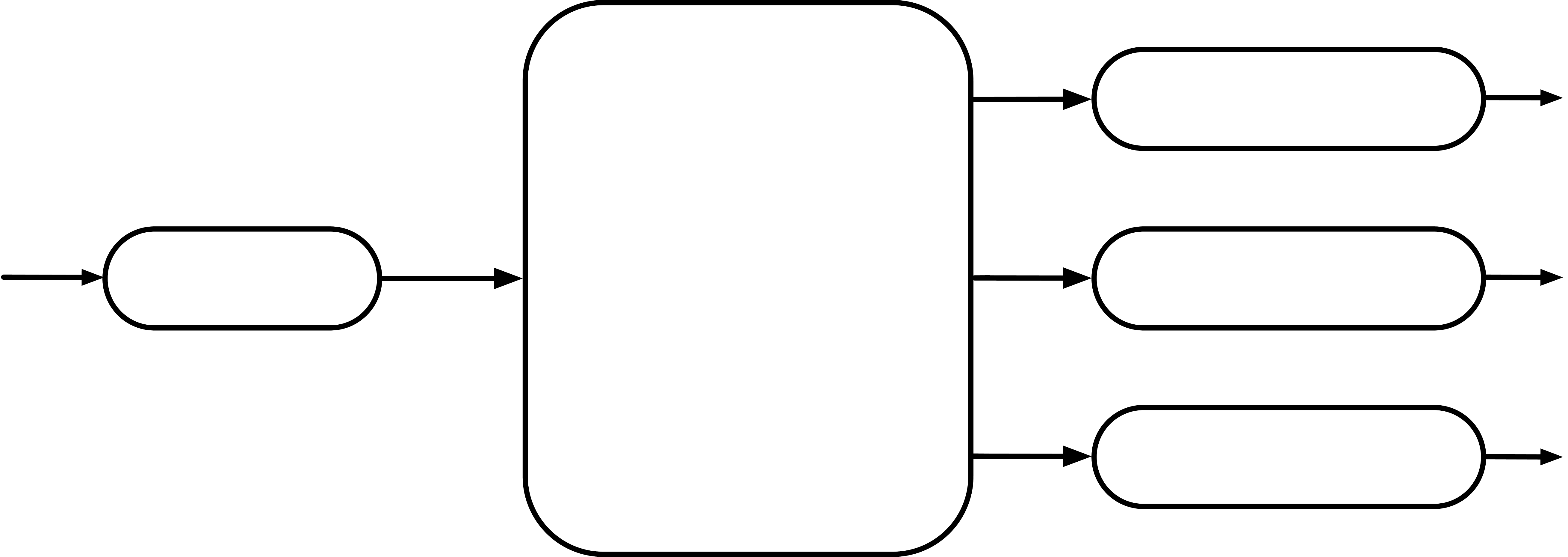}
\put (-8,19.4) {\scriptsize $(W,S)$}
\put (9,16.2) {\scriptsize Encoder}
\put (25.5,18.8) {\scriptsize $X^n$}
\put (74,27.8) {\scriptsize Receiver 1}
\put (74,16.2) {\scriptsize Receiver 2}
\put (72,4.8) {\scriptsize Eavesdropper}
\put (43,18) {\scriptsize WBC}
\put (36,14) {\scriptsize $p_{Y_{(1)}Y_{(2)}Z|X}$}
\put (96,31) {\scriptsize $(\hat{W}_{(1)},\hat{S}_{(1)})$}
\put (96,19.4) {\scriptsize $(\hat{W}_{(2)},\hat{S}_{(2)})$}
\put (101,5.8) {\scriptsize $\Ccancel[red]{{S}}$}
\end{overpic}
\caption{Channel model: CI-WBC.}\label{fig:chmodel}
\end{figure}

The achievable rate region is defined as the closure of the set of all achievable rate triples $(R_W,R_S,R_R)$. The following proposition defines an inner-bound on the achievable rate region.

\begin{proposition}[Adapted from \cite{chia2012three,watanabe2015optimal}]
\label{prop:MIB}
The region $\mathfrak{R}_{\text{\emph{CI-WBC}}}$ defined by the union over the tiples of rates $(R_W,R_S,R_R) \in \mathbb{R}_{+}^3$ satisfying
\begin{IEEEeqnarray*}{rCl}
R_W + R_S & \leq & \min \left\{  I(V ; Y_{(1)} ) ,  I(V ; Y_{(2)}) \right\}, \\ 
R_S & \leq & \min \left\{  I(V ; Y_{(1)} ) ,  I(V ; Y_{(2)}) \right\} - I(V ; Z), \\ 
R_W + R_R & \geq & I(X;Z), \\
R_R & \geq & I(X;Z|V),
\end{IEEEeqnarray*}
where the union is taken over all distributions $p_{VX}$ such that $V - X - (Y_{(1)},Y_{(2)},Z)$ forms a Markov chain, defines an inner-bound on the achievable region of the CI-WBC.
\end{proposition}


\section{Review of Polar Codes}\label{sec:rev}
Let $(\mathcal{X} \times \mathcal{Y}, p_{XY})$ be a Discrete Memoryless Source (DMS), where\footnote{Throughout this paper, we assume binary polarization. Nevertheless, an extension to $q$-ary alphabets is possible \cite{5513555,csasouglu2009polarization}.} ${X} \in \{0,1\}$ and $Y \in \mathcal{Y}$. The polar transform over the $n$-sequence $X^n$, $n$ being any power of $2$, is defined as $U^n \triangleq X^n G_n$, where $G_n \triangleq \begin{brsm}1 & 1\\1 & 0\end{brsm}^{\otimes n}$ is the source polarization matrix \cite{arikan2010source}. Since $G_n = G_n^{-1}$, then $X^n = U^n G_n$.

The polarization theorem for source coding with side information \cite[Theorem~1]{arikan2010source} states that the polar transform extracts the randomness of $X^n$ in the sense that, as $n \rightarrow \infty$, the set of indices $j \in [1,n]$ can be divided practically into two disjoint sets, namely $\smash{\mathcal{H}_{X|Y}^{(n)}}$ and $\smash{\mathcal{L}_{X|Y}^{(n)}}$, such that $U(j)$ for $j \in \mathcal{H}_{X|Y}^{(n)}$ is practically independent of $(U^{1:j-1},Y^n)$ and uniformly distributed, that is, $H ({U(j) | U^{1:j-1}, Y^n} ) \rightarrow 1$, and $U(j)$ for $j \in \smash{\mathcal{L}_{X|Y}^{(n)}}$ is almost determined by $(U^{1:j-1}, Y^n)$, which means that $H ( U(j) | U^{1:j-1}, Y^n ) \rightarrow 0$. Formally, let $\delta_n \triangleq 2^{-n^{\beta}}$, where $\beta \in (0, \frac{1}{2})$, and
\begin{IEEEeqnarray*}{rCl}
\mathcal{H}_{X|Y}^{(n)} & \triangleq & \left\{ j \in [1,n] \! : H \! \left( U(j) \left| U^{1:j-1},Y^n  \right. \! \right) \geq 1-\delta_n \right\}, \\
\mathcal{L}_{X|Y}^{(n)} & \triangleq & \left\{ j \in [1,n] \! : H \! \left( U(j) \left| U^{1:j-1},Y^n  \right. \! \right) \leq \delta_n \right\}. 
\end{IEEEeqnarray*}
Then, by \cite[Theorem~1]{arikan2010source} we have $\smash{\lim_{n \rightarrow \infty} \frac{1}{n} | \mathcal{H}_{X|Y}^{(n)} | }  = H(X|Y)$ and $\smash{\lim_{n \rightarrow \infty} \frac{1}{n} | \mathcal{L}_{X|Y}^{(n)} | } = 1 - H(X|Y)$. Consequently, the number of elements $U(j)$ that \emph{have not polarized} is asymptotically negligible in terms of rate, that is, $\smash{\lim_{n \rightarrow \infty} \frac{1}{n} | ( \mathcal{H}_{X|Y}^{(n)} )^{\text{C}} \setminus  \mathcal{L}_{X|Y}^{(n)}  | } = 0$. 

Furthermore, \cite[Theorem~2]{arikan2010source} states that given the part $U[(\mathcal{L}_{X|Y}^{(n)} )^{\text{C}}]$ and the channel output observations $Y^n$, the remaining part $\smash{U[\mathcal{L}_{X|Y}^{(n)}]}$ can be reconstructed by using SC decoding with error probability in $O(n \delta_n)$.

Similarly to $\mathcal{H}_{X|Y}^{(n)}$ and $\mathcal{L}_{X|Y}^{(n)}$, the sets $\mathcal{H}_{X}^{(n)}$ and $\mathcal{L}_{X}^{(n)}$ can be defined by considering that the observations $Y^n$ are absent. Since conditioning does not increase the entropy, we have $\mathcal{H}_{X}^{(n)} \supseteq \mathcal{H}_{X|Y}^{(n)}$ and $\mathcal{L}_{X}^{(n)} \subseteq \mathcal{L}_{X|Y}^{(n)}$.  A discrete memoryless channel $(\mathcal{X}, p_{Y|X}, \mathcal{Y})$ with some arbitrary $p_X$ can be seen as a DMS $(\mathcal{X} \times \mathcal{Y}, p_{X}p_{Y|X})$. In channel polar coding, first we define the sets of indices $\mathcal{H}_{X|Y}^{(n)}$, $\mathcal{L}_{X|Y}^{(n)}$, $\mathcal{H}_{X}^{(n)}$ and $\mathcal{L}_{X}^{(n)}$ from the target distribution $p_{X}p_{Y|X}$. Then, based on the previous sets, the encoder somehow constructs\footnote{Since the polar-based encoder will construct random variables that must approach the target distribution of the DMS, throughout this paper we use \emph{tilde} above the random variables to emphazise this purpose.} $\tilde{U}^n$ and applies the inverse polar transform $\tilde{X}^n = \tilde{U}^n G_n$. Afterwards, the transmitter sends $\tilde{X}^n$ over the channel, which induces $\tilde{Y}^n$. Let $(\tilde{X}^n,\tilde{Y}^n) \sim \tilde{q}_{X^n}\tilde{q}_{Y^n|X^n}$, if $\mathbb{V} (\tilde{q}_{X^nY^n}, p_{X^nY^n}) \rightarrow 0$ then the receiver can reliably reconstruct $\tilde{U}[\mathcal{L}_{X|Y}^{(n)}]$ from $\tilde{Y}^n$ and $\smash{\tilde{U}[(\mathcal{L}_{X|Y}^{(n)} )^{\text{C}}]}$ by performing SC decoding \cite{korada2010polar}.

\section{Polar Coding Scheme}\label{sec:PCS} 
Let $(\mathcal{V} \times \mathcal{X} \times \mathcal{Y}_{(1)} \times \mathcal{Y}_{(2)} \times \mathcal{Z} , p_{VXY_{(1)}Y_{(2)}Z})$ denote the DMS that represents the input $(V,X)$ and the output $(Y_{(1)},Y_{(2)},Z)$ random variables of the CI-WBC, where $|\mathcal{V} |=|\mathcal{X}|=2$. Without loss of generality, and to avoid the trivial case $R_S=0$ in Proposition~\ref{prop:MIB}, we assume that 
\begin{IEEEeqnarray}{c}
 \label{eq:assumpRate1}
H(V | Z) > H (V |  Y_{(1)})  \geq H (V |  Y_{(2)}). 
\end{IEEEeqnarray}
If $H (V |  Y_{(1)}) < H (V |  Y_{(2)})$, one can simply exchange the role of $Y_{(1)}$ and $Y_{(2)}$ in the encoding scheme of Section \ref{sec:PCS}. We propose a polar coding scheme that achieves the following rate triple,
\begin{IEEEeqnarray}{rCl}
  \IEEEeqnarraymulticol{3}{l}{%
(R_W,R_S,R_R)
}\nonumber\\*%
\quad &=&(I(V;Z), I(V;Y_{(1)}) - I(V;Z), I(X;Z|V)),  \IEEEeqnarraynumspace \label{eq:targetrate}
\end{IEEEeqnarray}
which corresponds to the one of the region in Proposition~\ref{prop:MIB} such that the confidential message rate is maximum and the amount of local randomness is minimum. 

For the input random variable $V$ of the DMS, we define the polar transform $A^n \triangleq V^n G_n$ and the sets
\begin{IEEEeqnarray}{rCl}
\mathcal{H}_V^{(n)} & \triangleq & \big\{j \in [1,n] \! : H  \big( A(j) \big| A^{1:j-1} \big) \geq 1 - \delta_n \big\},  \IEEEeqnarraynumspace \label{eq:HU} \\
\mathcal{H}_{V|Z}^{(n)} &  \triangleq & \big\{j \in [1,n]  \! :H \big( A(j) \big| A^{1:j-1}, Z^n \big)  \nonumber\\*
\IEEEeqnarraymulticol{3}{r}{
\geq 1 - \delta_n \big\},  \label{eq:HUZ}
} \IEEEeqnarraynumspace \\
\mathcal{L}_{V|Z}^{(n)} & \triangleq & \big\{j \in [1,n] \! : H \big( A(j) \big| A^{1:j-1}, Z^n \big) \leq \delta_n \big\}, \label{eq:LUZ} \IEEEeqnarraynumspace \\
\mathcal{L}_{V|Y_{(k)}}^{(n)} & \triangleq & \big\{j \in [1,n] \! : H  \big( A(j) \big| A^{1:j-1}, Y_{(k)}^n  \big) \leq  \delta_n \big\}, \label{eq:LUY_k}  \IEEEeqnarraynumspace
\end{IEEEeqnarray}
where $k \in [1,2]$. For the input random variable $X$, we define $T^n \triangleq X^n G_n$ and the associated sets
\begin{IEEEeqnarray}{rCl}
\mathcal{H}_{X|V}^{(n)} & \triangleq & \big\{ j \in [1,n]  \! : H \big( T(j) \big| T^{1:j-1}, V^n  \big) \nonumber\\*
\IEEEeqnarraymulticol{3}{r}{
\geq 1 - \delta_n \big\}. \label{eq:HV-UT}
} \IEEEeqnarraynumspace  \\
\mathcal{H}_{X|VZ}^{(n)} & \triangleq & \big\{ j \in [1,n]  \! : H \big( T(j) \big| T^{1:j-1}, V^n, Z^n  \big) \nonumber\\*
\IEEEeqnarraymulticol{3}{r}{
\geq 1 - \delta_n \big\}. \label{eq:HV-UTZ}
}  \IEEEeqnarraynumspace  
\end{IEEEeqnarray}
We have $p_{A^nT^n}(a^n, t^n) = p_{V^nX^n}(a^n G_n, t^n G_n)$ due to the invertibility of $G_n$ and, for convenience, we write
\begin{IEEEeqnarray}{rCl}
\IEEEeqnarraymulticol{3}{l}{%
p_{A^nT^n}(a^n, t^n)
}\nonumber\\*%
\quad
& = & p_{A^n}(a^n) p_{T^n | V^n}(t^n | a^n G_n) \nonumber\\* 
& = & \Big( \prod_{j=1}^n p_{A(j)|A^{1:j-1}}(a(j)|a^{1:j-1}) \Big)  \nonumber\\*
& & \cdot \Big( \prod_{j=1}^n p_{T(j)|T^{1:j-1} V^n}(t(j)|t^{1:j-1}, a^n G_n) \Big). \IEEEeqnarraynumspace \label{eq:distDMS}
\end{IEEEeqnarray}

The non-degraded nature of the broadcast channel means having to use a chaining construction \cite{6875073}. Hence, consider that the encoding takes place over $L$ blocks indexed by $i \in [1,L]$. At the $i$-th block, the encoder will construct $\tilde{A}_i^n$, which will carry the private and the confidential messages intended to both legitimate receivers. Additionally, the encoder will store into $\tilde{A}_i^n$ some elements from $\tilde{A}_{i-1}^n$ (if $i \in [2,L]$) and $\tilde{A}_{i+1}^n$ (if $i \in [1,L-1]$) so that both legitimate receivers are able to reliably reconstruct $\tilde{A}_{1:L}^n$. Then, given $\tilde{V}_i^n = \tilde{A}_i^n G_n$, the encoder will perform the polar-based channel prefixing to construct $\tilde{T}_i^n$. Finally, it will obtain $\tilde{X}_i^n = \tilde{T}_i^n G_n$, which will be transmitted over the WBC inducing the channel outputs $(\tilde{Y}_{(1),i}^n,\tilde{Y}_{(2),i}^n,\tilde{Z}_i^n)$.

\subsection{General polar-based encoding}\label{sec:encoding}

Consider the construction of $\tilde{A}^n_{1:L}$. Besides the sets defined in \eqref{eq:HU}--\eqref{eq:LUY_k}, we define the partition of $\mathcal{H}_{V}^{(n)}$:
\begin{IEEEeqnarray}{rCl}
\mathcal{G}^{(n)} & \triangleq & \mathcal{H}^{(n)}_{V|Z},  \label{eq:sG} \\
\mathcal{C}^{(n)} & \triangleq & \mathcal{H}^{(n)}_V \cap \big( \mathcal{H}^{(n)}_{V|Z} \big)^{\text{C}}. \label{eq:sC}
\end{IEEEeqnarray}
Moreover, we also define
\begin{IEEEeqnarray}{rCl}
\mathcal{G}^{(n)}_{0} & \triangleq & \mathcal{G}^{(n)} \cap \mathcal{L}^{(n)}_{V|Y_{(1)}}\cap   \mathcal{L}^{(n)}_{V|Y_{(2)}}, \label{eq:sG0} \\
\mathcal{G}^{(n)}_{1} & \triangleq  & \mathcal{G}^{(n)} \cap \big( \mathcal{L}^{(n)}_{V|Y_{(1)}} \big)^{\text{C}} \cap \mathcal{L}^{(n)}_{V|Y_{(2)}} , \label{eq:sG1} \\
\mathcal{G}^{(n)}_{2} & \triangleq & \mathcal{G}^{(n)} \cap \mathcal{L}^{(n)}_{V|Y_{(1)}} \cap  \big( \mathcal{L}^{(n)}_{V|Y_{(2)}} \big)^{\text{C}} , \label{eq:sG2}  \\ 
\mathcal{G}^{(n)}_{1,2} & \triangleq &  \mathcal{G}^{(n)} \cap \big( \mathcal{L}^{(n)}_{V|Y_{(1)}} \big)^{\text{C}}  \cap   \big( \mathcal{L}^{(n)}_{V|Y_{(2)}} \big)^{\text{C}}, \label{eq:sG12}
\end{IEEEeqnarray}
which form a partition of the set $\mathcal{G}^{(n)}$, and
\begin{IEEEeqnarray}{rCl}
\mathcal{C}^{(n)}_{0} & \triangleq & \mathcal{C}^{(n)} \cap \mathcal{L}^{(n)}_{V|Y_{(1)}} \cap   \mathcal{L}^{(n)}_{V|Y_{(2)}}, \label{eq:sC0} \\
\mathcal{C}^{(n)}_{1} & \triangleq  & \mathcal{C}^{(n)} \cap  \big( \mathcal{L}^{(n)}_{V|Y_{(1)}} \big)^{\text{C}} \cap \mathcal{L}^{(n)}_{V|Y_{(2)}}, \label{eq:sC1} \\
\mathcal{C}^{(n)}_{2} & \triangleq & \mathcal{C}^{(n)} \cap \mathcal{L}^{(n)}_{V|Y_{(1)}} \cap  \big( \mathcal{L}^{(n)}_{V|Y_{(2)}} \big)^{\text{C}} , \label{eq:sC2} \\
\mathcal{C}^{(n)}_{1,2} & \triangleq & \mathcal{C}^{(n)}  \cap  \big( \mathcal{L}^{(n)}_{V|Y_{(1)}} \big)^{\text{C}}  \cap   \big( \mathcal{L}^{(n)}_{V|Y_{(2)}} \big)^{\text{C}} ,\label{eq:sC12} 
\end{IEEEeqnarray}
which form a partition of $\mathcal{C}^{(n)}$. 
These sets are graphically represented in Figure~\ref{fig:partitionHv}. Roughly speaking, according to \eqref{eq:HU}, $A[\mathcal{H}^{(n)}_V ]$ is the \emph{nearly uniformly} distributed part of $A^n$. Thus, $\tilde{A}_i[\mathcal{H}^{(n)}_V ]$, $i \in [1,L]$, is suitable for storing uniformly distributed random sequences. According to \eqref{eq:HUZ}, $A[\mathcal{H}^{(n)}_{V|Z}]$ is \emph{almost} independent of $Z^n$. Thus, $\tilde{A}_i [\mathcal{G}^{(n)} ]$ is suitable for storing information to be secured from the eavesdropper, whereas $\tilde{A}_i [\mathcal{C}^{(n)} ]$ is not. Sets in \eqref{eq:sG0}--\eqref{eq:sC12} with subscript 1 (sets inside the red curve in Figure~\ref{fig:partitionHv}) form $\mathcal{H}^{(n)}_V \cap ( \mathcal{L}^{(n)}_{V|Y_{(1)}} )^{\text{C}}$, while those with subscript 2 (sets inside the blue curve) form $\smash{\mathcal{H}^{(n)}_V \cap ( \mathcal{L}^{(n)}_{V|Y_{(2)}} )^{\text{C}}}$. From Section~\ref{sec:rev}, $\smash{\tilde{A}_i [ \mathcal{H}^{(n)}_V \cap ( \mathcal{L}^{(n)}_{V|Y_{(k)}} )^{\text{C}} ]}$ is the nearly uniformly distributed part of the sequence $\tilde{A}_i^n$ required by legitimate receiver $k$ to reliably reconstruct the entire sequence by performing SC decoding. 

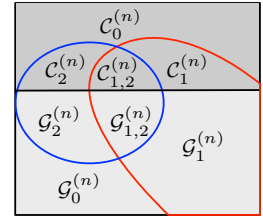
\begin{figure}[h!]
\centering
\begin{overpic}[width=0.2\textwidth]{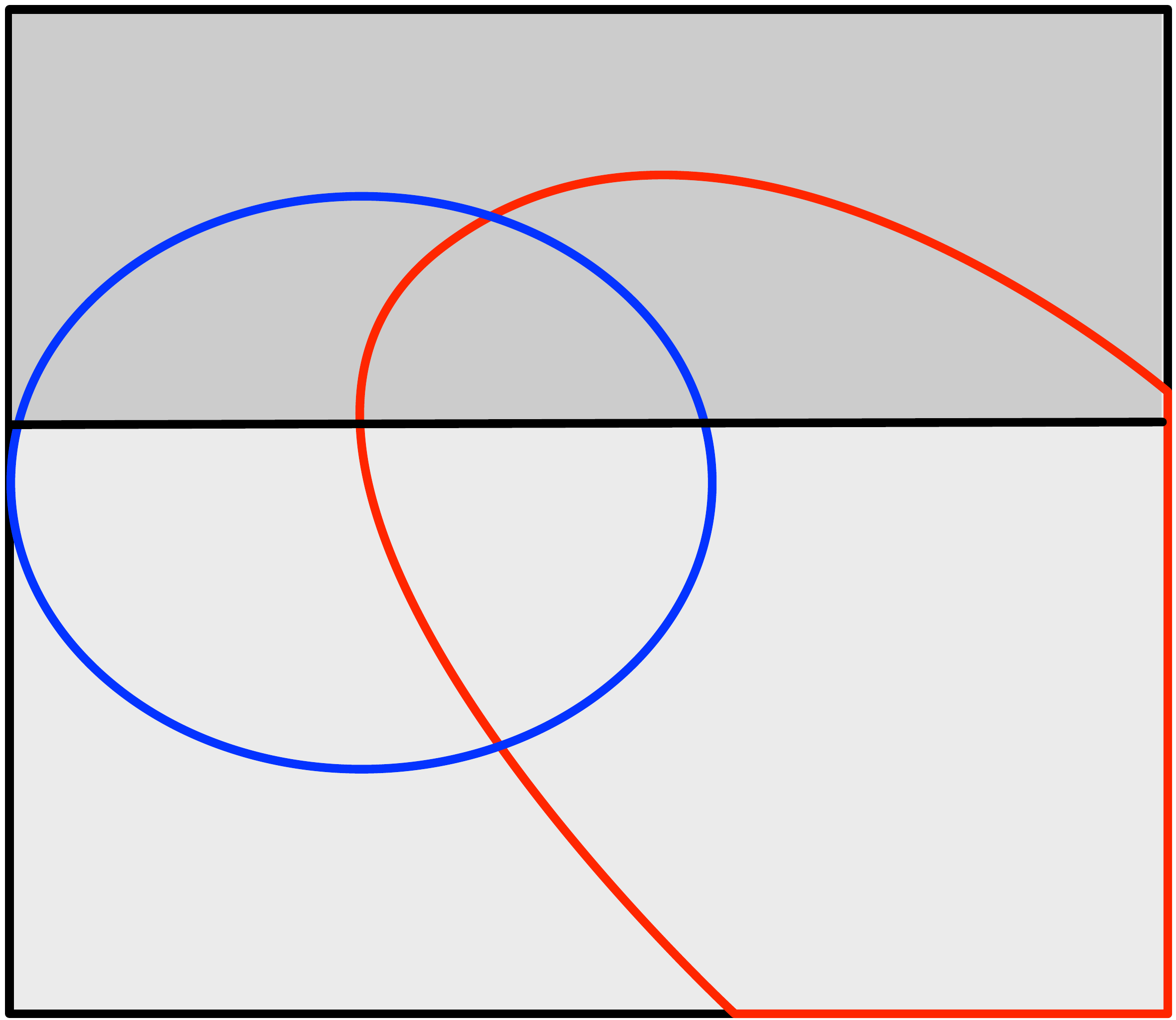}
\put (35,74) {\footnotesize $\mathcal{C}^{(n)}_0$}
\put (13,56) {\footnotesize $\mathcal{C}^{(n)}_2$}
\put (62,56) {\footnotesize $\mathcal{C}^{(n)}_1$}
\put (34,55) {\footnotesize $\mathcal{C}^{(n)}_{1,2}$}
\put (18,8) {\footnotesize $\mathcal{G}^{(n)}_0$}
\put (10,35) {\footnotesize $\mathcal{G}^{(n)}_2$}
\put (68,25) {\footnotesize $\mathcal{G}^{(n)}_1$}
\put (39,35) {\footnotesize $\mathcal{G}^{(n)}_{1,2}$}
\end{overpic}
\caption{Graphical representation of the sets in \eqref{eq:sG}--\eqref{eq:sC12}. The indices inside the soft and dark gray area form $\mathcal{G}^{(n)}$ and $\mathcal{C}^{(n)}$ respectively. The indices that form $\smash{\mathcal{H}^{(n)}_V \cap ( \mathcal{L}^{(n)}_{V|Y_{(1)}} )^{\text{C}}}$ are those inside the red curve, while those inside the blue curve form $\smash{\mathcal{H}^{(n)}_V \cap ( \mathcal{L}^{(n)}_{V|Y_{(2)}} )^{\text{C}}}$.}\label{fig:partitionHv}
\end{figure}

For sufficiently large $n$, assumption \eqref{eq:assumpRate1} imposes the following restriction on the size of the previous sets:
\begin{IEEEeqnarray}{c}
\big| \mathcal{G}^{(n)}_1 \big|  - \big| \mathcal{C}^{(n)}_2 \big| \geq \big| \mathcal{G}^{(n)}_2 \big|  - \big| \mathcal{C}^{(n)}_1 \big| > \big| \mathcal{C}^{(n)}_{1,2} \big| - \big| \mathcal{G}^{(n)}_0 \big|. \nonumber\\*
\label{eq:assumpRate1Impl2}
\end{IEEEeqnarray}
The left-hand inequality in \eqref{eq:assumpRate1Impl2} holds from the fact that 
\begin{IEEEeqnarray}{rCl}
  \IEEEeqnarraymulticol{3}{l}{%
\big| \mathcal{C}^{(n)}_1  \cup \mathcal{G}^{(n)}_1 \big| - \big| \mathcal{C}^{(n)}_2  \cup \mathcal{G}^{(n)}_2 \big|
}\nonumber\\*%
\quad 
& = & \big| \mathcal{H}^{(n)}_V \cap \big( \mathcal{L}^{(n)}_{V|Y_{(1)}} \big)^{\text{C}} \setminus  \mathcal{H}^{(n)}_V \cap \big( \mathcal{L}^{(n)}_{V|Y_{(2)}} \big)^{\text{C}} \big|  \nonumber \\*
& & -  \big| \mathcal{H}^{(n)}_V \cap \big( \mathcal{L}^{(n)}_{V|Y_{(2)}} \big)^{\text{C}} \setminus  \mathcal{H}^{(n)}_V \cap \big( \mathcal{L}^{(n)}_{V|Y_{(1)}} \big)^{\text{C}} \big|  \nonumber \\*
& = & \big| \mathcal{H}^{(n)}_V \cap \big( \mathcal{L}^{(n)}_{V|Y_{(1)}} \big)^{\text{C}} \big| -  \big| \mathcal{H}^{(n)}_V \cap \big( \mathcal{L}^{(n)}_{V|Y_{(2)}} \big)^{\text{C}} \big| \nonumber \\*
& \geq &  0, \nonumber
\end{IEEEeqnarray}
where the positivity holds by \cite[Theorem~1]{arikan2010source} because
\begin{IEEEeqnarray}{rCl}
  \IEEEeqnarraymulticol{3}{l}{%
\lim_{n \rightarrow \infty} \frac{1}{n} \big| \mathcal{H}^{(n)}_V \cap \big( \mathcal{L}^{(n)}_{V|Y_{(k)}} \big)^{\text{C}} \big| 
}\nonumber\\*%
\quad 
& = & \lim_{n \rightarrow \infty}  \frac{1}{n} \big| \mathcal{H}^{(n)}_{V|Y_{(k)}} \big| \nonumber \\*
& & + \lim_{n \rightarrow \infty}  \frac{1}{n} \big| \mathcal{H}^{(n)}_V \cap \big( \mathcal{L}^{(n)}_{V|Y_{(k)}} \big)^{\text{C}} \setminus  \mathcal{H}^{(n)}_{V|Y_{(k)}}  \big| \nonumber \\*
& = & H(V|Y_{(k)}), \nonumber
\end{IEEEeqnarray}
for any $k \in [1,2]$. Similarly, the right-hand inequality in \eqref{eq:assumpRate1Impl2} follows from \cite[Theorem~1]{arikan2010source} and the fact that
\begin{IEEEeqnarray}{rCl}
  \IEEEeqnarraymulticol{3}{l}{%
\big| \mathcal{G}^{(n)}_0  \cup \mathcal{G}^{(n)}_2 \big| - \big| \mathcal{C}^{(n)}_{1}  \cup \mathcal{C}^{(n)}_{1,2} \big|
}\nonumber\\*%
\quad 
& = & \big| \mathcal{H}^{(n)}_{V|Z} \setminus \mathcal{H}^{(n)}_V \cap \big( \mathcal{L}^{(n)}_{V|Y_{(1)}} \big)^{\text{C}} \big|  \nonumber \\*
& & - \big| \mathcal{H}^{(n)}_V  \cap \big( \mathcal{L}^{(n)}_{V|Y_{(1)}} \big)^{\text{C}} \setminus  \mathcal{H}^{(n)}_{V|Z}  \big|  \nonumber \\
 &  = &  \big| \mathcal{H}^{(n)}_{V|Z} \big| -  \big| \mathcal{H}^{(n)}_V \cap \big( \mathcal{L}^{(n)}_{V|Y_{(1)}} \big)^{\text{C}} \big|. \nonumber
\end{IEEEeqnarray}
Thus, according to \eqref{eq:assumpRate1Impl2}, we must consider four cases:
\begin{itemize}
\item[A.] $| \mathcal{G}^{(n)}_1 | > | \mathcal{C}^{(n)}_2 | $, ${| \mathcal{G}^{(n)}_2 |  > | \mathcal{C}^{(n)}_1 |}$ and ${| \mathcal{G}^{(n)}_0 |  \geq | \mathcal{C}^{(n)}_{1,2} |}$;
\item[B.] $| \mathcal{G}^{(n)}_1 | > | \mathcal{C}^{(n)}_2 | $, ${| \mathcal{G}^{(n)}_2 |  > | \mathcal{C}^{(n)}_1 |}$ and ${| \mathcal{G}^{(n)}_0 |  < | \mathcal{C}^{(n)}_{1,2} |}$;
\item[C.] $| \mathcal{G}^{(n)}_1 | \geq | \mathcal{C}^{(n)}_2 |$, $| \mathcal{G}^{(n)}_2 |  \leq | \mathcal{C}^{(n)}_1 |$ and $| \mathcal{G}^{(n)}_0 |  > | \mathcal{C}^{(n)}_{1,2} |$;
\item[D.] $| \mathcal{G}^{(n)}_1 | < | \mathcal{C}^{(n)}_2 | $, $| \mathcal{G}^{(n)}_2 |  < | \mathcal{C}^{(n)}_1 |$ and $| \mathcal{G}^{(n)}_0 |  > | \mathcal{C}^{(n)}_{1,2} |$.
\end{itemize}

The generic encoding process for all cases is summarized in Algorithm~\ref{alg:genericenc}. For $i \in [1,L]$, let $W_i$ be a uniformly distributed vector of length $| \mathcal{C}^{(n)} |$ that represents the private message. The encoder forms $\tilde{A}_{i}[\mathcal{C}^{(n)}]$ by simply storing $W_i$. Indeed, if $i \in [1,L-1]$, notice that the encoder forms $\tilde{A}_{i+1}[\mathcal{C}^{(n)}]$ before constructing the entire sequence $\tilde{A}_{i}^n$. 

From $\tilde{A}_{i}[\mathcal{C}^{(n)}]$, $i \in [1,L]$, we define the sequences
\begin{IEEEeqnarray}{rCl}
\Psi_{i}^{(V)} & \triangleq & \tilde{A}_{i}[\mathcal{C}_2^{(n)}],  \\
\Gamma_{i}^{(V)}  & \triangleq & \tilde{A}_{i}[\mathcal{C}_{1,2}^{(n)}],  \\
\Theta_{i}^{(V)} & \triangleq & \tilde{A}_{i}[\mathcal{C}_1^{(n)}]. 
\end{IEEEeqnarray}
Notice that $[\Psi_i^{(V)},\Gamma_i^{(V)}] = \tilde{A}_i[\mathcal{C}^{(n)}_2 \cup \mathcal{C}^{(n)}_{1,2}]$ is required by legitimate receiver 2 to reliably estimate $\tilde{A}_i^n$ entirely and, thus, the encoder repeats $[\Psi_i^{(V)},\Gamma_i^{(V)}]$, if $i \in [1,L-1]$, conveniently in $\tilde{A}_{i+1}[\mathcal{G}^{(n)}]$ (the function \texttt{form\_A$_{\text{\texttt{G}}}$} is responsible of the chaining construction and is described later). On the other hand, $[\Theta_i^{(V)},\Gamma_i^{(V)}]= \tilde{A}_i[\mathcal{C}^{(n)}_1 \cup \mathcal{C}^{(n)}_{1,2}]$ is required by legitimate receiver 1. Nevertheless, in order to satisfy the strong secrecy condition in~\eqref{eq:secrecycond}, $[\Theta_i^{(V)},\Gamma_i^{(V)}]$, $i \in [2,L]$, is not repeated directly into $\tilde{A}_{i-1}[\mathcal{G}^{(n)}]$, but the encoder copies instead $\bar{\Theta}_i^{(V)}$ and $\bar{\Gamma}_i^{(V)}$ obtained as follows. Let $\kappa_{\Theta}^{(V)}$ and $\kappa_{\Gamma}^{(V)}$ be uniformly distributed keys with length $\smash{|\mathcal{C}_{1}^{(n)}|}$ and $\smash{|\mathcal{C}_{1,2}^{(n)} |}$ respectively, which are privately shared between transmitter and both legitimate receivers. For any $i \in [2,L]$, we define the sequences
\begin{IEEEeqnarray}{rCl}
\bar{\Theta}^{(V)}_{i} \triangleq \tilde{A}_{i}[\mathcal{C}_{1}^{(n)}] \oplus \kappa_{\Theta}^{(V)},  \\
\bar{\Gamma}_i^{(V)} \triangleq \tilde{A}_{i}[\mathcal{C}_{1,2}^{(n)}] \oplus \kappa_{\Gamma}^{(V)}.
\end{IEEEeqnarray}
Since these secret keys are reused in all blocks, their size becomes negligible in terms of rate for $L$ large enough.

\algrenewcommand\algorithmicindent{0.5em}
\begin{algorithm}[t]
\caption{Generic PC encoding scheme}\label{alg:genericenc}
\begin{small}
\begin{algorithmic}[1]
\Require Private and confidential messages $W_{1:L}$ and $S_{1:L}$; randomization sequences $R_{1:L}$; random sequence $\Lambda^{(X)}_{0}$; and secret keys $\kappa_{\Theta}^{(V)}$, $\kappa_{\Gamma}^{(V)}$, $\kappa_{{\Upsilon \Phi}_{(1)}}^{(V)}$ and $\kappa_{{\Upsilon \Phi}_{(2)}}^{(V)}$.
\State $\Psi^{(V)}_{0}$, $\Gamma^{(V)}_{0}$, $\Pi^{(V)}_{0}$, $\Lambda^{(V)}_{0}$, $\bar{\Theta}^{(V)}_{L+1}$, $\bar{\Gamma}^{(V)}_{L+1} \leftarrow \varnothing$
\State $\tilde{A}_{1}[\mathcal{C}^{(n)}] \leftarrow W_1$
\State $\Psi^{(V)}_{1}, \Gamma^{(V)}_{1}  \leftarrow \tilde{A}_{1}[\mathcal{C}^{(n)}] $
\State $\bar{\Theta}^{(V)}_{1}, \bar{\Gamma}^{(V)}_{1}   \leftarrow \varnothing$ 
\Comment For notation purposes
\For{$i = 1$ \text{to} $L$}
\If{$i \neq L$}  
\State $  \tilde{A}_{i+1}[\mathcal{C}^{(n)}] \leftarrow W_{i+1}$ 
\State $ \Psi^{(V)}_{i+1}, \Gamma^{(V)}_{i+1}, \bar{\Theta}^{(V)}_{i+1}, \bar{\Gamma}^{(V)}_{i+1}    \leftarrow \!  \big(  \tilde{A}_{i+1}[\mathcal{C}^{(n)}] , \kappa_{\Theta}^{(V)}  , \kappa_{\Gamma}^{(V)}   \big)$
\EndIf
\State $\tilde{A}_{i}^n$, $\Pi^{(V)}_{i}$, $\Lambda^{(V)}_{i} \leftarrow$ \texttt{form\_A$_{\text{\texttt{G}}}$}$\big(i, S_i,\bar{\Theta}^{(V)}_{i+1}, \bar{\Gamma}^{(V)}_{i+1}, \dots $ \par
$\qquad \qquad \qquad \qquad  \qquad \quad \, \,   \Psi^{(V)}_{i-1}, \Gamma^{(V)}_{i-1}, \Pi^{(V)}_{i-1}, \Lambda^{(V)}_{i-1}\big)$ 
\State \textbf{if} $i=1$ \textbf{then} $\Upsilon_{(1)}^{(V)} \leftarrow \tilde{A}_1 \big[ \mathcal{H}_V^{(n)} \cap \big( \mathcal{L}_{V|Y_{(1)}}^{(n)} \big)^{\text{C}}\big]$
\State \textbf{if} $i=L$ \textbf{then} $\Upsilon_{(2)}^{(V)} \leftarrow \tilde{A}_L \big[ \mathcal{H}_V^{(n)} \cap \big( \mathcal{L}_{V|Y_{(2)}}^{(n)} \big)^{\text{C}}\big]$
\For{$j \in \big( \mathcal{H}^{(n)}_V \big)^{\text{C}}$}
\If{$j \in \big( \mathcal{H}^{(n)}_V \big)^{\text{C}}  \setminus \mathcal{L}^{(n)}_V$}
\State $\tilde{A}(j) \leftarrow p_{A(j)|A^{1:j-1}}  \big( \tilde{a}_i(j) \big| \tilde{a}_i^{1:j-1}  \big)$
\ElsIf{$j \in \mathcal{L}^{(n)}_V$}
\State $\tilde{A}(j) \leftarrow \argmax_{a \in \mathcal{V}} p_{A(j)|A^{1:j-1}} \big( {a} | {a}^{1:j-1} \big)$
\EndIf
\EndFor
\State $\Phi_{(k),i}^{(V)} \leftarrow \tilde{A}_i \big[ \big( \mathcal{H}_V^{(n)}  \big)^{\text{C}} \cap \big( \mathcal{L}_{V|Y_{(k)}}^{(n)} \big)^{\text{C}} \big], \quad k \in [1,2]$
\State $\tilde{X}_i^n, \Lambda_{i}^{(X)} \leftarrow \text{\texttt{pb\_ch\_pref}}\big( \tilde{A}_i^n G_n, R_i, \Lambda_{i-1}^{(X)} \big)$
\EndFor
\State Send $\big( \Phi_{(k),i}^{(V)},\Upsilon_{(k)}^{(V)} \big) \oplus \kappa_{{\Upsilon \Phi}_{(k)}}^{(V)}$ to the receiver $k \in [1,2]$
\State \textbf{return} $\tilde{X}_{1:L}^n$
\end{algorithmic}
\end{small}
\end{algorithm}

The function \texttt{form\_A$_{\text{\texttt{G}}}$} in Algorithm~\ref{alg:genericenc} constructs sequences $\tilde{A}_{1:L}[\mathcal{G}^{(n)}]$ differently depending on which case, among cases A, B, C or D described before, characterizes the given CI-WBC. This part of the encoding is described in detail in Section~\ref{sec:formAG} and Algorithm~\ref{alg:formA}. 

Then, given $\tilde{A}_i [\mathcal{C}^{(n)} \cup \mathcal{G}^{(n)}]$, the encoder forms the remaining entries of $\tilde{A}_i^n$, i.e., $\tilde{A}_i [(\mathcal{H}_V^{(n)})^{\text{C}}]$, as follows. If $j \in \mathcal{L}_{V}^{(n)}$, it constructs $\tilde{A}_{i}(j)$ deterministically by using SC encoding as in \cite{7447169}. Thus, we define the SC encoding function $\xi^{(V)}_{(j)}: \{0,1\}^{j-1} \rightarrow \{ 0,1 \}$ in Algorithm~\ref{alg:genericenc} as
\begin{IEEEeqnarray}{c}
\xi^{(V)}_{(j)} \big( {a}^{1:j-1} \big)  \triangleq \argmax_{a \in \mathcal{V}} p_{A(j)|A^{1:j-1}} \left( {a} \left| {a}^{1:j-1} \right. \right), \IEEEeqnarraynumspace \label{eq:scdeterministic}
\end{IEEEeqnarray} 
$p_{A(j)|A^{1:j-1}}$ corresponding to the distribution of the original DMS --see \eqref{eq:distDMS}--. Therefore, notice that only the part $\tilde{A}_i [ ( \mathcal{H}^{(n)}_V )^{\text{C}} \setminus \mathcal{L}^{(n)}_V]$ of $\tilde{A}_i^n$ is constructed randomly.

Finally, for $i \in [1,L]$, given $\tilde{V}_i^n = \tilde{A}_i^n G_n$, a randomization sequence $R_i$ and a uniformly distributed random sequence $\Lambda_0^{(V)}$, the encoder performs the polar-based channel prefixing (function \texttt{pb\_ch\_pref} in Algorithm~\ref{alg:genericenc}) to obtain $\tilde{X}_i^n$, which is transmitted over the WBC inducing outputs $\tilde{Y}_{(1),i}^n$, $\tilde{Y}_{(2),i}^n$ and $\tilde{Z}_i^n$. This part of the encoding is described in detail in Section~\ref{sec:PCSCP}. 

Furthermore, the encoder obtains
\begin{IEEEeqnarray}{c}
\Phi_{(k),i}^{(V)} \triangleq \tilde{A}_i \big[ ( \mathcal{H}_V^{(n)}  )^{\text{C}} \cap ( \mathcal{L}_{V|Y_{(k)}}^{(n)} )^{\text{C}} \big]
\end{IEEEeqnarray} 
for any $k \in [1,2]$ and $i \in [1,L]$, which is required by legitimate receiver $k$ to reliably estimate $\tilde{A}_i^n$ entirely. Since $\Phi_{(k),i}^{(V)}$ is not \emph{nearly uniform}, the encoder cannot make it available to the legitimate receiver~$k$ by means of the chaining structure. Also, the encoder obtains 
\begin{IEEEeqnarray}{rCl}
\Upsilon_{(1)}^{(V)} & \triangleq & \tilde{A}_1 \big[ \mathcal{H}_V^{(n)} \cap ( \mathcal{L}_{V|Y_{(1)}}^{(n)} )^{\text{C}} \big],  \\
\Upsilon_{(2)}^{(V)} & \triangleq & \tilde{A}_L \big[ \mathcal{H}_V^{(n)} \cap ( \mathcal{L}_{V|Y_{(2)}}^{(n)} )^{\text{C}} \big]. 
\end{IEEEeqnarray} 
The sequence $\smash{\Upsilon_{(1)}^{(V)}}$ is required by legitimate receiver~1 to initialize the decoding process, while the sequence $\Upsilon_{(2)}^{(V)}$ is required by legitimate receiver~2. Therefore, the transmitter additionally sends $\smash{( \Upsilon_{(k)}^{(V)}, \Phi_{(k),i}^{(V)}) \oplus \kappa_{{\Upsilon \Phi}_{(k)}}^{(V)}}$ to legitimate receiver~$k$, $\kappa_{{\Upsilon \Phi}_{(k)}}^{(V)}$ being a uniformly distributed key with size
\begin{IEEEeqnarray}{c}
L  \Big| \big( \mathcal{H}_V^{(n)}  \big)^{\text{C}} \cap \big( \mathcal{L}_{V|Y_{(k)}}^{(n)} \big)^{\text{C}} \Big| +  \Big| \mathcal{H}_V^{(n)} \cap \big( \mathcal{L}_{V|Y_{(k)}}^{(n)} \big)^{\text{C}} \Big| \nonumber 
\end{IEEEeqnarray} 
that is privately shared between transmitter and the corresponding receiver. We show in Section~\ref{sec:performance_rates} that the length of $\kappa_{{\Upsilon \Phi}_{(1)}}^{(V)}$ and $\kappa_{{\Upsilon \Phi}_{(2)}}^{(V)}$ is asymptotically negligible in terms of rate.

\subsection{Function \texttt{form\_A$_{\text{\texttt{G}}}$}.}\label{sec:formAG}
Function \texttt{form\_A$_{\text{\texttt{G}}}$} encodes the confidential messages $S_{1:L}$ and builds the chaining construction. 

Based on the sets in \eqref{eq:sG}--\eqref{eq:sC12}, let $\mathcal{R}^{(n)}_{1} \subseteq \mathcal{G}^{(n)}_{0} \cup \mathcal{G}^{(n)}_{2}$, $\mathcal{R}^{\prime (n)}_{1} \subseteq \mathcal{G}^{(n)}_{2}$,  $\mathcal{R}^{(n)}_{2} \subseteq \mathcal{G}^{(n)}_{1}$, $\mathcal{R}^{\prime (n)}_{2} \subseteq \mathcal{G}^{(n)}_{1}$, $\mathcal{R}^{(n)}_{1,2} \subseteq \mathcal{G}^{(n)}_{0}$, $\mathcal{R}^{\prime (n)}_{1,2} \subseteq \mathcal{G}^{(n)}_{0}$, $\mathcal{I}^{(n)} \subseteq \mathcal{G}^{(n)}_{0} \cup \mathcal{G}^{(n)}_{2}$, $\mathcal{R}^{(n)}_{\text{S}} \subseteq \mathcal{G}^{(n)}_{1}$ and $\mathcal{R}^{(n)}_{\Lambda} \subseteq \mathcal{G}^{(n)}_{1}$ form an additional partition of $\mathcal{G}^{(n)}$. The definition of $\mathcal{R}^{(n)}_{1}$, $\mathcal{R}^{\prime (n)}_{1}$, $\mathcal{R}^{(n)}_{2}$, $\mathcal{R}^{\prime (n)}_{2}$, $\mathcal{R}^{(n)}_{1,2}$ and $\smash{\mathcal{R}^{\prime (n)}_{1,2}}$ will depend on the particular case (among A to D), while
\begin{IEEEeqnarray}{rCl}
\mathcal{I}^{(n)} & \triangleq & \big( \mathcal{G}^{(n)}_{0} \cup \mathcal{G}^{(n)}_{2} \big) \nonumber \\*
\IEEEeqnarraymulticol{3}{r}{
\setminus   \big(  \mathcal{R}^{(n)}_{\text{1}} \cup  \mathcal{R}^{\prime (n)}_{1} \cup \mathcal{R}^{(n)}_{1,2} \cup \mathcal{R}^{\prime (n)}_{1,2} \big),
} \IEEEeqnarraynumspace  \label{eq:ASetI} \\
\mathcal{R}^{(n)}_{\text{S}} & \triangleq & \text{any subset of } \mathcal{G}^{(n)}_{1} \setminus \big( \mathcal{R}^{(n)}_{\text{2}} \cup  \mathcal{R}^{\prime (n)}_{2} \big) \nonumber \\*
  \IEEEeqnarraymulticol{3}{r}{
\text{ with size } \big| \mathcal{I}^{(n)} \cap \mathcal{G}^{(n)}_{2} \big|, 
} \IEEEeqnarraynumspace \label{eq:ASetRs} \\
\mathcal{R}_{\Lambda}^{(n)} & \triangleq & \mathcal{G}_{1,2}^{(n)} \cup \big( \mathcal{G}^{(n)}_{1} \setminus \big( \mathcal{R}^{(n)}_{2} \cup \mathcal{R}^{\prime (n)}_{2} \cup \mathcal{R}^{(n)}_{\text{S}} \big) \big). \IEEEeqnarraynumspace \label{eq:ASetLambda}
\end{IEEEeqnarray}

For $i \in [1,L]$, let $S_i$ denote a uniformly distributed vector that represents the confidential message. The message $S_1$ has size $| \mathcal{I}^{(n)} \cup \mathcal{G}^{(n)}_{1} \cup \mathcal{G}^{(n)}_{1,2} |$; for $i \in [2,L-1]$, $S_i$ has size $| \mathcal{I}^{(n)}|$; and $S_L$ has size $| \mathcal{I}^{(n)} \cup \mathcal{G}_2^{(n)} |$.

For $i \in [1,L]$, we write $\Psi_i^{(V)} \triangleq \big[ \Psi_{1,i}^{(V)}, \Psi_{2,i}^{(V)}  \big]$, $\Gamma_i^{(V)} \triangleq \big[ \Gamma_{1,i}^{(V)}, \Gamma_{2,i}^{(V)}  \big]$, $\bar{\Theta}_i^{(V)}  \triangleq  \big[ \bar{\Theta}_{1,i}^{(V)}, \bar{\Theta}_{2,i}^{(V)}  \big]$ and $\bar{\Gamma}_i^{(V)}  \triangleq  \big[ \bar{\Gamma}_{1,i}^{(V)}, \bar{\Gamma}_{2,i}^{(V)}  \big]$; and we define $\Psi_{p,i}$, $\Gamma_{p,i}$, $\bar{\Theta}_{p,i}$ and $\bar{\Gamma}_{p,i}$, for any $p \in [1,2]$, accordingly in each case.

\algrenewcommand\algorithmicindent{0.5em}
\begin{algorithm}[t]
\caption{Function \texttt{form\_A$_{\text{\texttt{G}}}$}}\label{alg:formA}
\begin{small}
\begin{algorithmic}[1]
\Require $i$, $S_i$, $\bar{\Theta}^{(V)}_{i+1}$, $\bar{\Gamma}^{(V)}_{i+1}$, $\Psi^{(V)}_{i-1}$, $\Gamma^{(V)}_{i-1}$, $\Pi^{(V)}_{i-1}$, $\Lambda^{(V)}_{i-1}$
\State Define a partition of $\mathcal{G}^{(n)}$ according to each case: $\mathcal{R}^{(n)}_{1}$, $\mathcal{R}^{\prime (n)}_{1}$,  $\mathcal{R}^{(n)}_{2}$, $\mathcal{R}^{\prime (n)}_{2}$, $\mathcal{R}^{(n)}_{1,2}$, $\mathcal{R}^{\prime (n)}_{1,2}$, $\mathcal{I}^{(n)}$, $\mathcal{R}^{(n)}_{\text{S}}$ and $\mathcal{R}^{(n)}_{\Lambda}$
\State \textbf{if} $i=1$ \textbf{then} $\tilde{A}_{1}[\mathcal{I}^{(n)} \cup \mathcal{G}^{(n)}_{1} \cup \mathcal{G}^{(n)}_{1,2}] \leftarrow S_1$
\State \textbf{if} $i \in [2,L-1]$ \textbf{then} $\tilde{A}_{i}[\mathcal{I}^{(n)}] \leftarrow S_i$
\State \textbf{if} $i = L$ \textbf{then} $\tilde{A}_{L}[\mathcal{I}^{(n)} \cup \mathcal{G}^{(n)}_{2}] \leftarrow S_L$
\State $\Psi^{(V)}_{1,i-1}$, $\Psi^{(V)}_{2,i-1} \leftarrow \Psi^{(V)}_{i-1}$
\State $\Gamma^{(V)}_{1,i-1}$, $\Gamma^{(V)}_{2,i-1} \leftarrow \Gamma^{(V)}_{i-1}$
\State $\bar{\Theta}^{(V)}_{1,i+1}$, $\bar{\Theta}^{(V)}_{2,i+1} \leftarrow \bar{\Theta}^{(V)}_{i+1}$
\State $\bar{\Gamma}^{(V)}_{1,i+1}$, $\bar{\Gamma}^{(V)}_{2,i+1} \leftarrow \bar{\Gamma}^{(V)}_{i+1}$
\State $\tilde{A}_{i}[\mathcal{R}_{1,2}^{(n)}] \leftarrow \Gamma^{(V)}_{1,i-1} \oplus \bar{\Gamma}^{(V)}_{1,i+1}$
\State $\tilde{A}_{i}[\mathcal{R}_{1,2}^{\prime (n)}] \leftarrow \Psi^{(V)}_{2,i-1} \oplus \bar{\Theta}^{(V)}_{2,i+1}$
\If{$i \in [1,L-1]$}
\State $\tilde{A}_{1}[\mathcal{R}_{1}^{(n)}] \leftarrow \bar{\Theta}^{(V)}_{1,i+1}$
\State $\tilde{A}_{1}[\mathcal{R}_{1}^{\prime (n)}] \leftarrow \bar{\Gamma}^{(V)}_{2,i+1}$
\EndIf
\If{$i \in [2,L]$}
\State $\tilde{A}_{i}[\mathcal{R}_{2}^{(n)}] \leftarrow \Psi^{(V)}_{1,i-1}$
\State $\tilde{A}_{i}[\mathcal{R}_{2}^{\prime (n)}] \leftarrow \Gamma^{(V)}_{2,i-1}$
\State $\tilde{A}_{i}[\mathcal{R}_{\text{S}}^{(n)}] \leftarrow \Pi^{(V)}_{i-1}$
\State $\tilde{A}_{i}[\mathcal{R}_{\Lambda}^{(n)}] \leftarrow \Lambda^{(V)}_{i-1}$
\EndIf
\State $\Pi_i^{(V)} \leftarrow \tilde{A}_i [\mathcal{I}^{(n)} \cap \mathcal{G}^{(n)}_{2}]$
\State $\Lambda_i^{(V)} \leftarrow \tilde{A}_i [\mathcal{R}^{(n)}_{\Lambda}]$
\State \textbf{return} the sequences $\tilde{A}_{i}^n$, $\Pi_i^{(V)}$ and $\Lambda_i^{(V)}$
\end{algorithmic}
\end{small}
\end{algorithm}

\subsubsection{Case A}
In this case, recall that $| \mathcal{G}^{(n)}_1 | > | \mathcal{C}^{(n)}_2 |$, ${| \mathcal{G}^{(n)}_2 |  > | \mathcal{C}^{(n)}_1 |}$ and ${| \mathcal{G}^{(n)}_0 |  \geq | \mathcal{C}^{(n)}_{1,2} |}$. We define
\begin{IEEEeqnarray}{rCl}
\mathcal{R}^{(n)}_{1} & \triangleq & \text{any subset of } \mathcal{G}^{(n)}_2 \text{ with size } \big| \mathcal{C}^{(n)}_1 \big|, \label{eq:ASetR1} \\
\mathcal{R}^{(n)}_{2} & \triangleq & \text{any subset of } \mathcal{G}^{(n)}_{1} \text{ with size } \big| \mathcal{C}^{(n)}_2 \big|, \label{eq:ASetR2} \\
\mathcal{R}^{(n)}_{1,2} & \triangleq & \text{any subset of } \mathcal{G}^{(n)}_0 \text{ with size } \big| \mathcal{C}^{(n)}_{1,2} \big|,  \label{eq:ASetR12}
\end{IEEEeqnarray}
and $\mathcal{R}^{\prime (n)}_{1} = \mathcal{R}^{\prime (n)}_{2} = \mathcal{R}^{\prime (n)}_{1,2} \triangleq \emptyset$. By the assumption of Case A, it is clear that $\mathcal{R}^{(n)}_{1}$, $\mathcal{R}^{(n)}_{2}$ and $\mathcal{R}^{(n)}_{1,2}$ exist. Also, by \eqref{eq:assumpRate1Impl2}, the set $\mathcal{I}^{(n)}$ exists, and so will $\mathcal{R}^{(n)}_{\text{S}}$ because
\begin{IEEEeqnarray}{rCl}
\IEEEeqnarraymulticol{3}{l}{%
\big| \mathcal{G}^{(n)}_{1} \setminus \big( \mathcal{R}^{(n)}_{\text{2}} \cup  \mathcal{R}^{\prime (n)}_{2} \big) \big| - \big| \mathcal{I}^{(n)} \cap \mathcal{G}^{(n)}_{2}\big| 
} \nonumber \\*
& = & \big| \mathcal{G}^{(n)}_{1} \setminus \big( \mathcal{R}^{(n)}_{\text{2}} \cup  \mathcal{R}^{\prime (n)}_{2} \big) \big| - \big| \big(  \mathcal{G}^{(n)}_{2} \setminus \mathcal{R}^{(n)}_{\text{1}} \cup  \mathcal{R}^{\prime (n)}_{1} \big) \big| \nonumber \\*
& = & \big| \mathcal{G}^{(n)}_{1} \big| - \big| \mathcal{C}^{(n)}_{2} \big|  -   \big| \mathcal{G}^{(n)}_{2} \big| - \big| \mathcal{C}^{(n)}_{1} \big|   \nonumber \\
& \geq &  0. \nonumber
\end{IEEEeqnarray}

For $i \in [1,L]$, we define $\Psi_{1,i}^{(V)} \triangleq \Psi_{i}^{(V)}$, $\Gamma_{1,i}^{(V)} \triangleq \Gamma_{i}^{(V)}$, $\bar{\Theta}_{1,i}^{(V)} \triangleq \bar{\Theta}_{i}^{(V)}$, $\bar{\Gamma}_{1,i}^{(V)} \triangleq \bar{\Gamma}_{i}^{(V)}$ and, therefore, we have $\Psi_{2,i}^{(V)} = \Gamma_{2,i}^{(V)} = \bar{\Theta}_{2,i}^{(V)} = \bar{\Gamma}_{2,i}^{(V)} \triangleq \varnothing$.

From \eqref{eq:sC2}, we have $\mathcal{C}_2^{(n)} \subseteq \mathcal{L}_{V|Y_{(1)}}^{(n)} \setminus \mathcal{L}_{V|Y_{(2)}}^{(n)}$. Thus, $\smash{\Psi_{i-1}^{(V)}=\tilde{A}_{i-1}[\mathcal{C}_2^{(n)}]}$ is needed by receiver 2 to reliably reconstruct $\tilde{A}_{i-1}^n$, but can be reliably inferred by receiver 1 given $\smash{\tilde{A}_{i-1}[(\mathcal{L}_{V|Y_{(1)}}^{(n)})^{\text{C}}]}$. Thus, according to Algorithm~\ref{alg:formA}, the encoder repeats the entire sequence $\Psi_{i-1}^{(V)}$ in $\tilde{A}_{i}[\mathcal{R}_{2}^{(n)}] \subseteq \tilde{A}_{i}[\mathcal{L}_{V|Y_{(2)}}^{(n)} \setminus \mathcal{L}_{V|Y_{(1)}}^{(n)}]$. 

Similarly, from \eqref{eq:sC1}, we have $\mathcal{C}_1^{(n)} \subseteq \mathcal{L}_{V|Y_{(1)}}^{(n)} \setminus \mathcal{L}_{V|Y_{(2)}}^{(n)}$. Thus, $\smash{\Theta_{i+1}^{(V)}=\tilde{A}_{i+1}[\mathcal{C}_1^{(n)}]}$ is needed by receiver 1 to reliably reconstruct $\smash{\tilde{A}_{i+1}^n}$ but given $\smash{\tilde{A}_{i}[(\mathcal{L}_{V|Y_{(2)}}^{(n)})^{\text{C}}]}$ can be inferred by receiver~2. Thus, the encoder repeats the sequence $\bar{\Theta}_{i+1}^{(V)}$ in $\tilde{A}_{i}[\mathcal{R}_{1}^{(n)}] \subseteq \tilde{A}_{i}[\mathcal{L}_{V|Y_{(2)}}^{(n)} \setminus \mathcal{L}_{V|Y_{(1)}}^{(n)}]$.

Finally, from \eqref{eq:sC12}, $\mathcal{C}_{1,2}^{(n)} \subseteq (\mathcal{L}_{V|Y_{(2)}}^{(n)})^{\text{C}} \cap (\mathcal{L}_{V|Y_{(1)}}^{(n)})^{\text{C}}$. Thus, $\smash{\Gamma_{i-1}^{(V)}}$ and $\smash{\Gamma_{i+1}^{(V)}}$ are needed by both receivers to reliably reconstruct the sequences $\smash{\tilde{A}_{i-1}^n}$ and $\smash{\tilde{A}_{i+1}^n}$ respectively. Consequently, the encoder repeats $\Gamma_{i-1}^{(V)}$ and $\bar{\Gamma}_{i+1}^{(V)}$ in $\smash{\tilde{A}_{i}[\mathcal{R}_{1,2}^{(n)}] \subseteq \tilde{A}_{i}[\mathcal{L}_{V|Y_{(1)}}^{(n)} \cap \mathcal{L}_{V|Y_{(2)}}^{(n)}]}$. Indeed, both sequences are repeated in the same entries of $\tilde{A}_{i}[\mathcal{G}_{0}^{(n)}]$ by performing $\Gamma_{i-1}^{(V)} \oplus \bar{\Gamma}_{i+1}^{(V)}$. Since $\Gamma_0^{(V)} = \bar{\Gamma}_{L+1}^{(V)} = \varnothing$, only $\bar{\Gamma}_{2}^{(V)}$ is repeated at block 1 and $\Gamma_{L-1}^{(V)}$ at block $L$.

Moreover, part of secret message $S_i$, $i \in [1,L]$, is stored into some entries of $\tilde{A}_{i}^n$ whose indices belong to $\mathcal{G}_{2}^{(n)}$. Thus, in any block $i \in [2,L]$, the encoder repeats 
\begin{IEEEeqnarray}{c}
\Pi_{i-1}^{(V)} \triangleq \tilde{A}_{i-1}[\mathcal{I}^{(n)} \cap \mathcal{G}_{2}^{(n)}]
\end{IEEEeqnarray}
in $\tilde{A}_{i}[\mathcal{R}_{\text{S}}^{(n)}]  \subseteq \tilde{A}_{i}[\mathcal{L}_{V|Y_{(2)}}^{(n)} \setminus \mathcal{L}_{V|Y_{(1)}}^{(n)}]$. Also, it repeats 
\begin{IEEEeqnarray}{c}
\Lambda_{i-1}^{(V)} \triangleq \tilde{A}_{i-1}[\mathcal{R}_{\Lambda}^{(n)}] 
\end{IEEEeqnarray}
in $\tilde{A}_{i}[\mathcal{R}_{\Lambda}^{(n)}]$. Hence, $\Lambda_{1}^{(V)}$ is repeated in all blocks. 


These sets that form the partition of $\mathcal{G}^{(n)}$ in Case A can be seen in Figure~\ref{fig:EncCasA}, which also displays the encoding process that aims to construct $\tilde{A}_{1:L}[\mathcal{C}^{(n)} \cup \mathcal{G}^{(n)}]$.

\begin{figure}[h!]
\centering
\begin{overpic}[width=0.8\linewidth]{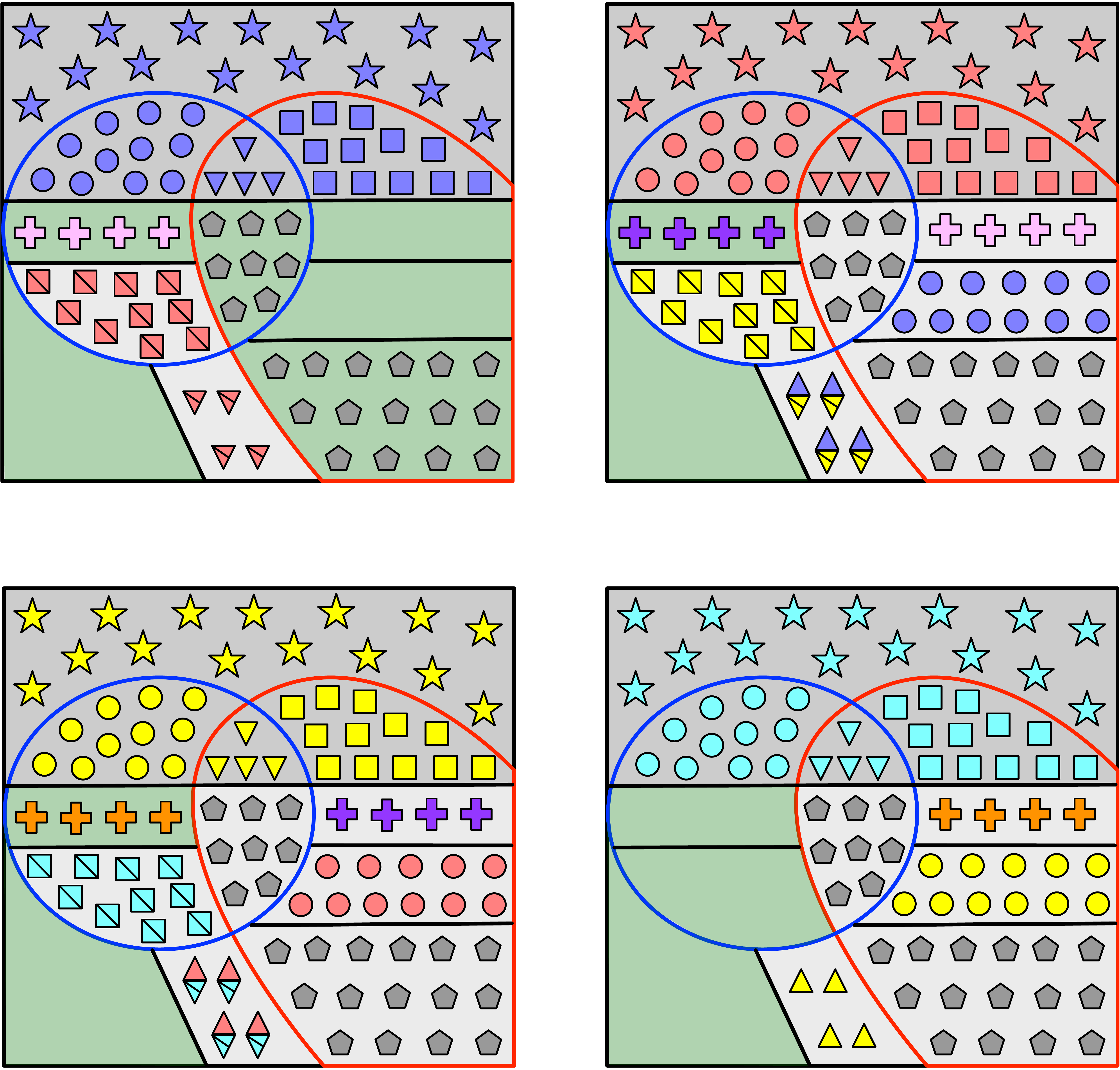}
\put (16,47) {\small Block 1}
\put (69,47) {\small Block 2}
\put (16,-5) {\small Block 3}
\put (69,-5) {\small Block 4}
\end{overpic}
\vspace{0.7cm}
\caption{For Case A, graphical representation of the encoding that leads to the construction of $\tilde{A}_{1:L}[ \mathcal{H}_V^{(n)} ]$ when $L=4$. Consider the block 2, $\mathcal{R}^{(n)}_{1}$, $\mathcal{R}^{(n)}_{2}$, $\mathcal{R}^{(n)}_{1,2}$, $\mathcal{R}^{(n)}_{\text{S}}$ and $\mathcal{R}^{(n)}_{\Lambda}$ are those areas filled with yellow squares, blue circles, blue and yellow diamonds, pink crosses, and gray pentagons, respectively; the set $\mathcal{I}^{(n)}$ is the green filled area. At block $i \in [1,L]$, $W_i$ is represented by symbols of the same color (e.g., red symbols at block 2), and $\Theta_i^{(V)}$, $\Psi_i^{(V)}$ and $\Gamma_i^{(V)}$ are represented by squares, circles and triangles respectively. Also, $\bar{\Theta}_i^{(V)}$ and $\bar{\Gamma}_i^{(V)}$ are denoted by squares and triangles, respectively, with a line through them. At block $i \in [2,L-1]$, the diamonds denote $\Gamma_{i-1}^{(V)} \oplus \bar{\Gamma}_{i+1}^{(V)}$. In block $i \in [1,L]$, $S_i$ is stored into those entries whose indices belong to the green area. For $i \in [1,L-1]$, $\Pi_i^{(V)}$ is denoted by crosses (e.g., purple crosses at block 2), and is repeated into $\tilde{A}_{i+1}\big[ \mathcal{R}^{(n)}_{\text{S}} \big]$. The sequence $\Lambda_1^{(V)}$ from $S_1$ is represented by gray pentagons and is repeated in all blocks. The sequences $\Upsilon_{(1)}^{(V)}$ and $\Upsilon_{(2)}^{(V)}$ are those entries inside the red and blue curve at block 1 and $L$, respectively.}\label{fig:EncCasA} 
\end{figure}

\subsubsection{Case B}
In this case, recall that $| \mathcal{G}^{(n)}_1 | > | \mathcal{C}^{(n)}_2 |$, ${| \mathcal{G}^{(n)}_2 |  > | \mathcal{C}^{(n)}_1 |}$ and ${| \mathcal{G}^{(n)}_0 |  < | \mathcal{C}^{(n)}_{1,2} |}$. We define $\mathcal{R}^{(n)}_{1}$ and $\mathcal{R}^{(n)}_{2}$ as in \eqref{eq:ASetR1} and \eqref{eq:ASetR2} respectively, and $\mathcal{R}_{1,2}^{\prime (n)} \triangleq \emptyset$. Now, since $| \mathcal{G}^{(n)}_0 |  < | \mathcal{C}^{(n)}_{1,2} |$, for any $i \in [1,L]$ only a part of $\Gamma_{i-1}^{(V)} \oplus \bar{\Gamma}_{i+1}^{(V)}$ can be repeated entirely in $\tilde{A}_i[\mathcal{G}_0^{(n)}]$. Thus, we define $\mathcal{R}^{(n)}_{1,2} \triangleq \mathcal{G}^{(n)}_{0}$ and 
\begin{IEEEeqnarray}{rCl}
\mathcal{R}^{\prime (n)}_{1} & \triangleq & \text{any subset of } \mathcal{G}^{(n)}_2 \setminus \mathcal{R}^{(n)}_{1} \nonumber \\*
& &  \qquad \qquad \, \,
 \text{ with size } \big| \mathcal{C}^{(n)}_{1,2} \big| - \big| \mathcal{G}^{(n)}_{0} \big|, 
 \IEEEeqnarraynumspace \label{eq:BSetR1p} \\
\mathcal{R}^{\prime (n)}_{2} & \triangleq & \text{any subset of } \mathcal{G}^{(n)}_1 \setminus \mathcal{R}^{(n)}_{2} \nonumber \\*
& & \qquad \qquad \, \,
 \text{ with size } \big| \mathcal{C}^{(n)}_{1,2} \big| - \big| \mathcal{G}^{(n)}_{0} \big|.
 \IEEEeqnarraynumspace \label{eq:BSetR2p}
\end{IEEEeqnarray}
Obviously, $\smash{\mathcal{R}^{(n)}_{1,2}}$ exists and, by the assumption of Case B, so do $\smash{\mathcal{R}^{(n)}_{1}}$ and $\smash{\mathcal{R}^{(n)}_{2}}$. By \eqref{eq:assumpRate1Impl2}, $\smash{\mathcal{R}^{\prime (n)}_{1}}$ exists and so does $\smash{\mathcal{I}^{(n)}}$. Indeed, since $\smash{\mathcal{G}^{(n)}_{0} \setminus\mathcal{R}^{(n)}_{1,2} = \emptyset}$, then $\smash{\mathcal{I}^{(n)} \subseteq \mathcal{G}^{(n)}_{2}}$. Again by \eqref{eq:assumpRate1Impl2}, $\smash{\mathcal{R}^{\prime (n)}_{2}}$ exists and so does $\smash{\mathcal{R}^{(n)}_{\text{S}}}$ because
\begin{IEEEeqnarray}{rCl}
\IEEEeqnarraymulticol{3}{l}{%
\big| \mathcal{G}^{(n)}_{1} \setminus \big( \mathcal{R}^{(n)}_{\text{2}} \cup  \mathcal{R}^{\prime (n)}_{2} \big) \big| - \big| \big(  \mathcal{G}^{(n)}_{2} \setminus \mathcal{R}^{(n)}_{\text{1}} \cup  \mathcal{R}^{\prime (n)}_{1} \big) \big|
} \nonumber \\* 
& = & \big| \mathcal{G}^{(n)}_{1} \big| - \big| \mathcal{C}^{(n)}_{2} \big| - \big( \big| \mathcal{C}^{(n)}_{1,2} \big| - \big| \mathcal{G}^{(n)}_{0} \big| \big) \nonumber \\
&  &  - \Big(  \big| \mathcal{G}^{(n)}_{2} \big| - \big| \mathcal{C}^{(n)}_{1} \big|  - \big( \big| \mathcal{C}^{(n)}_{1,2} \big| - \big| \mathcal{G}^{(n)}_{0} \big| \big) \Big) \nonumber \\
& = & \big| \mathcal{G}^{(n)}_{1} \big| - \big| \mathcal{C}^{(n)}_{2} \big|  -  \big| \mathcal{G}^{(n)}_{2} \big| + \big| \mathcal{C}^{(n)}_{1} \big|  \nonumber \\
& \geq &  0. \nonumber
\end{IEEEeqnarray}
Indeed, since $\smash{\mathcal{I}^{(n)} \subseteq \mathcal{G}^{(n)}_{2}}$, notice that $|\smash{\mathcal{R}_{\text{S}}^{(n)}} | = | \smash{\mathcal{I}^{(n)}}|$. These sets that form the partition of $\mathcal{G}^{(n)}$ in Case B can be seen in Figure~\ref{fig:EncCasB}, which also displays the encoding process that aims to construct $\tilde{A}_{1:L}[\mathcal{C}^{(n)} \cup \mathcal{G}^{(n)}]$. 

\begin{figure}[h!]
\centering
\begin{overpic}[width=0.8\linewidth]{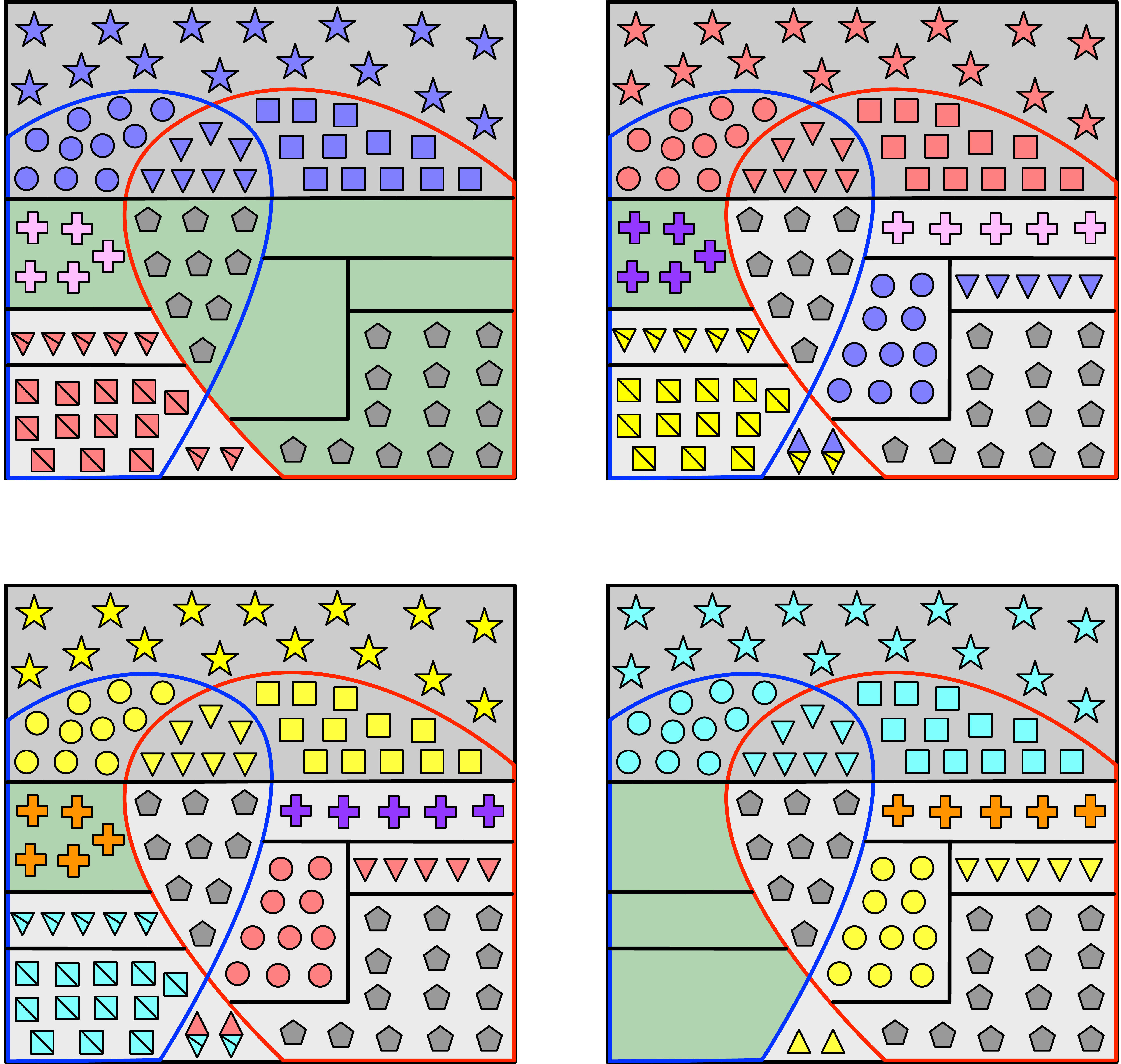}
\put (16,47) {\small Block 1}
\put (69,47) {\small Block 2}
\put (16,-5) {\small Block 3}
\put (69,-5) {\small Block 4}
\end{overpic}
\vspace{0.7cm}
\caption{For Case B, graphically representation of the encoding that leads to the construction of $\tilde{A}_{1:L}[ \mathcal{H}_V^{(n)} ]$ when $L=4$. Consider the block 2, the sets $\mathcal{R}^{(n)}_{1}$, $\mathcal{R}^{\prime(n)}_{1}$, $\mathcal{R}^{(n)}_{2}$, $\mathcal{R}^{\prime (n)}_{2}$, $\mathcal{R}^{(n)}_{1,2}$, $\mathcal{R}^{(n)}_{\text{S}}$ and $\mathcal{R}^{(n)}_{\Lambda}$ are those areas filled with yellow squares, yellow triangles, blue circles, blue triangles, blue and yellow diamonds, pink crosses, and gray pentagons, respectively; and $\mathcal{I}^{(n)}$ is the green filled area with purple crosses. At block $i \in [1,L]$, $W_i$ is represented by symbols of the same color (e.g., red symbols at block 2), and $\Theta_i^{(V)}$, $\Psi_i^{(V)}$ and $\Gamma_i^{(V)}$ are represented by squares, circles and triangles respectively. Also, $\bar{\Theta}_i^{(V)}$ and $\bar{\Gamma}_i^{(V)}$ are denoted by squares and triangles, respectively, with a line through them. At block $i \in [2,3]$, the diamonds denote $\Gamma_{i-1,1}^{(V)} \oplus \bar{\Gamma}_{1,i+1}^{(V)}$. In block $i \in [1,L]$, $S_i$ is stored into those entries whose indices belong to the green area. For $i \in [2,L-1]$, $\Pi_i^{(V)} = S_i$ and, therefore, $S_{i}$ is repeated entirely into $\tilde{A}_{i+1}[ \mathcal{R}^{(n)}_{\text{S}} ]$. The sequence $\Lambda_1^{(V)}$ from $S_1$ is represented by gray pentagons and is repeated in all blocks. The sequences $\Upsilon_{(1)}^{(V)}$ and $\Upsilon_{(2)}^{(V)}$ are the entries inside the red and blue curve at block 1 and $L$, respectively.
}\label{fig:EncCasB} 
\end{figure}

In this case, for any $i \in [1,L]$, $\Psi_{1,i}^{(V)} \triangleq \Psi_{i}^{(V)}$, $\bar{\Theta}_{1,i}^{(V)} \triangleq \bar{\Theta}_{i}^{(V)}$ and $\Psi_{2,i}^{(V)} = \bar{\Theta}_{2,i}^{(V)} \triangleq  \varnothing$; and we define $\Gamma_{1,i}^{(V)}$ and $\bar{\Gamma}_{1,i}^{(V)}$ as any part of $\Gamma_{i}^{(V)}$ and $\bar{\Gamma}_{i}^{(V)}$, respectively, with size $|\mathcal{R}^{(n)}_{1,2}|$, and $\Gamma_{2,i}^{(V)}$ and $\bar{\Gamma}_{2,i}^{(V)}$ as the remaining parts with size $|\mathcal{C}^{(n)}_{1,2}| - | \mathcal{R}^{(n)}_{1,2} |$. Now, the encoder copies $\Gamma_{1,i-1}^{(V)} \oplus \bar{\Gamma}_{1,i+1}^{(V)}$ into $\tilde{A}_i[\mathcal{R}_{1,2}^{(n)}]$, and $\Gamma_{2,i-1}^{(V)}$ and $\bar{\Gamma}_{2,i+1}^{(V)}$ into $\tilde{A}_i[\mathcal{R}_{2}^{\prime (n)}]$ and $\tilde{A}_i[\mathcal{R}_{1}^{\prime (n)}]$ respectively. Moreover, since $\mathcal{I}^{(n)} \subseteq \mathcal{G}_2^{(n)}$, notice that $\Pi_{i}^{(V)} = S_i$ for any $i \in [2,L-1]$.

\subsubsection{Case C}
In this case, recall that $| \mathcal{G}^{(n)}_1 | \geq | \mathcal{C}^{(n)}_2 |$, $| \mathcal{G}^{(n)}_2 |  \leq | \mathcal{C}^{(n)}_1 |$ and $| \mathcal{G}^{(n)}_0 |  > | \mathcal{C}^{(n)}_{1,2} |$. Hence, we define $\mathcal{R}^{(n)}_2$ and $\mathcal{R}^{(n)}_{1,2}$ as in \eqref{eq:ASetR2} and \eqref{eq:ASetR12} respectively, and $\mathcal{R}^{\prime (n)}_1 = \mathcal{R}^{\prime (n)}_2 = \mathcal{R}^{\prime (n)}_{1,2} \triangleq \emptyset$. On the other hand, since $| \mathcal{G}^{(n)}_2 |  \leq | \mathcal{C}^{(n)}_1 |$, now for $i \in [1,L-1]$ only a part of $\bar{\Theta}_{i+1}^{(V)}$ can be repeated entirely in $\tilde{A}_i[\mathcal{G}_2^{(n)}]$, and we define
\begin{IEEEeqnarray}{rCl}
\mathcal{R}^{(n)}_{1} & \triangleq & \text{the union of } \mathcal{G}^{(n)}_{2} \text{ with any subset of } \nonumber \\
& & \quad \mathcal{G}^{(n)}_{0} \setminus  \mathcal{R}_{1,2}^{(n)} \text{ with size } \big| \mathcal{C}^{(n)}_{1} \big| - \big| \mathcal{G}^{(n)}_{2} \big|. \IEEEeqnarraynumspace \label{eq:CSetR1}
\end{IEEEeqnarray}
It is clear that $\mathcal{R}^{(n)}_{2}$ and $\mathcal{R}^{(n)}_{1,2}$ exist. By \eqref{eq:assumpRate1Impl2}, $\mathcal{R}^{(n)}_{1}$ also exists and, hence, so does $\mathcal{I}^{(n)}$. Since $\mathcal{R}^{(n)}_{1} \supseteq \mathcal{G}^{(n)}_{2}$, $\mathcal{I}^{(n)} \cap \mathcal{G}^{(n)}_2 = \emptyset$ and $\mathcal{R}^{(n)}_{\text{S}} = \emptyset$. These sets that form $\mathcal{G}^{(n)}$ are represented in Figure~\ref{fig:EncCasC}, which also displays the part of the encoding that aims to construct $\tilde{A}_{1:L}[\mathcal{C}^{(n)} \cup \mathcal{G}^{(n)}]$.

\begin{figure}[h!]
\centering
\begin{overpic}[width=0.8\linewidth]{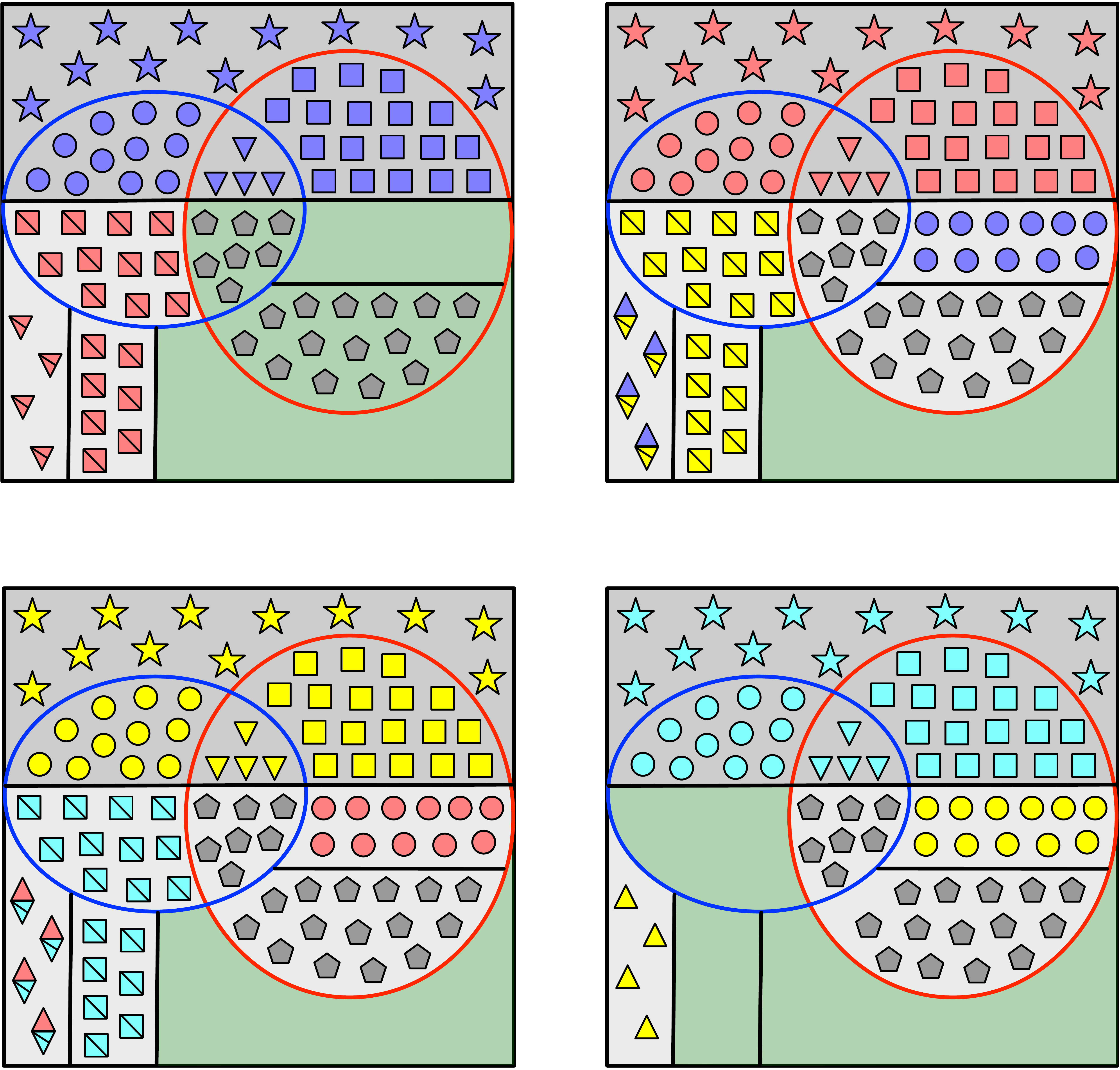}
\put (16,47) {\small Block 1}
\put (69,47) {\small Block 2}
\put (16,-5) {\small Block 3}
\put (69,-5) {\small Block 4}
\end{overpic}
\vspace{0.7cm}
\caption{For Case C, graphically representation of the encoding that leads to the construction of $\tilde{A}_{1:L}[ \mathcal{H}_V^{(n)} ]$ when $L=4$. Consider the block 2, $\mathcal{R}^{(n)}_{1}$, $\mathcal{R}^{(n)}_{2}$, $\mathcal{R}^{(n)}_{1,2}$ and $\mathcal{R}^{(n)}_{\Lambda}$ are those areas filled with yellow squares, blue circles, blue and yellow diamonds, and gray pentagons, respectively; and $\mathcal{I}^{(n)}$ is the green filled area. At block $i \in [1,L]$, $W_i$ is represented by symbols of the same color (e.g., red symbols at block 2), and $\Theta_i^{(V)}$, $\Psi_i^{(V)}$ and $\Gamma_i^{(V)}$ are represented by squares, circles and triangles respectively. Also, $\bar{\Theta}_i^{(V)}$ and $\bar{\Gamma}_i^{(V)}$ are denoted by squares and triangles, respectively, with a line through them. At block $i \in [2,3]$, the diamonds denote $\Gamma_{i-1,1}^{(V)} \oplus \bar{\Gamma}_{1,i+1}^{(V)}$. For $i \in [1,L]$, $S_i$ is stored into those entries belonging to the green area. $\Lambda_1^{(V)}$ is represented by gray pentagons and is repeated in all blocks. The sequences $\Upsilon_{(1)}^{(V)}$ and $\Upsilon_{(2)}^{(V)}$ are the entries inside the red and blue circumference at block 1 and $L$, respectively.}\label{fig:EncCasC} 
\end{figure}

In this case, for $i \in [1,L]$, we define $\Psi_{1,i}^{(V)} \triangleq \Psi_{i}^{(V)}$, $\Gamma_{1,i}^{(V)} \triangleq \Gamma_{i}^{(V)}$, $\bar{\Theta}_{1,i}^{(V)} \triangleq \bar{\Theta}_{i}^{(V)}$, $\bar{\Gamma}_{1,i}^{(V)} \triangleq \bar{\Gamma}_{i}^{(V)}$, and $\Psi_{2,i}^{(V)} = \Gamma_{2,i}^{(V)} = \bar{\Theta}_{2,i}^{(V)} = \bar{\Gamma}_{2,i}^{(V)} \triangleq \varnothing$. Moreover, notice that  $\Pi_{i}^{(V)} = \varnothing$ because $\mathcal{I}^{(n)} \cap \mathcal{G}^{(n)}_2 = \emptyset$.

\subsubsection{Case D}
In this case, recall that $| \mathcal{G}^{(n)}_1 | < | \mathcal{C}^{(n)}_2 |$, $| \mathcal{G}^{(n)}_2 |  \leq | \mathcal{C}^{(n)}_1 |$ and $| \mathcal{G}^{(n)}_0 |  > | \mathcal{C}^{(n)}_{1,2} |$. As in Case A and Case C, since $| \mathcal{G}^{(n)}_0 |  > | \mathcal{C}^{(n)}_{1,2} |$ we define $\mathcal{R}^{(n)}_{1,2}$ as in \eqref{eq:ASetR12} and $\mathcal{R}^{\prime (n)}_1 = \mathcal{R}^{\prime (n)}_2 \triangleq \emptyset$. Thus, for $i \in [1,L]$, we set $\Gamma_{1,i}^{(V)} \triangleq \Gamma_{i}^{(V)}$, $\bar{\Gamma}_{1,i}^{(V)} \triangleq \bar{\Gamma}_{i}^{(V)}$ and $\bar{\Gamma}_{2,i}^{(V)} = \Gamma_{2,i}^{(V)} \triangleq \varnothing$. On the other hand, since $| \mathcal{G}^{(n)}_1 | < | \mathcal{C}^{(n)}_2 |$, now for $i \in [2,L]$ only a part of $\bar{\Psi}_{i-1}^{(V)}$ can be repeated entirely in $\tilde{A}_i[\mathcal{G}_1^{(n)}]$, and we define $\mathcal{R}^{(n)}_{2} \triangleq \mathcal{G}_{1}^{(n)}$ and
\begin{IEEEeqnarray}{rCl}
\mathcal{R}_{1,2}^{\prime (n)} & \triangleq & \text{any subset of } \mathcal{G}_0^{(n)} \setminus \mathcal{R}_{1,2}^{(n)} \nonumber \\
& & \qquad \qquad \, \, \text{ with size } \big| \mathcal{C}_{2}^{(n)} \big| - \big| \mathcal{G}_{1}^{(n)} \big|. \IEEEeqnarraynumspace \label{eq:DSetR12p}
\end{IEEEeqnarray}
By \eqref{eq:assumpRate1Impl2}, it is clear that $\mathcal{R}_{1,2}^{\prime (n)}$ exists. Also, for $i\in [1,L]$, we define $\Psi_{1,i}^{(V)}$ as any part $\Psi_{i}^{(V)}$ with size $| \mathcal{G}_{1}^{(n)} |$, and $\Psi_{2,i}^{(V)}$ as the remaining part with size $|\mathcal{C}_{2}^{(n)} \big| - \big| \mathcal{G}_{1}^{(n)}|$. 

Despite $\big| \mathcal{G}^{(n)}_2 \big|  < \big| \mathcal{C}^{(n)}_1 \big|$ as in Case C, the set $\mathcal{R}_1^{(n)}$ is not defined as in \eqref{eq:CSetR1}, but
\begin{IEEEeqnarray}{rCl}
\mathcal{R}^{(n)}_{1} & \triangleq & \text{the union of } \mathcal{G}^{(n)}_{2} \text{ with any subset }  \nonumber \\
& & \quad  \text{of } \mathcal{G}^{(n)}_{0} \setminus  \big( \mathcal{R}_{1,2}^{(n)} \cup \mathcal{R}_{1,2}^{\prime (n)}  \big) \text{ with size } \nonumber \\
& & \qquad  \big| \mathcal{C}^{(n)}_{1} \big| - \big| \mathcal{G}^{(n)}_{2} \big| -  \big( \big| \mathcal{C}^{(n)}_{2} \big|  - \big| \mathcal{G}^{(n)}_{1} \big|\big). \IEEEeqnarraynumspace \label{eq:DSetR1}
\end{IEEEeqnarray}
By the assumption in \eqref{eq:assumpRate1Impl2}, the set $\mathcal{R}_{1}^{(n)}$ exists because 
\begin{IEEEeqnarray}{rCl}
\IEEEeqnarraymulticol{3}{l}{%
\big|\mathcal{G}^{(n)}_{0} \setminus \big( \mathcal{R}_{1,2}^{(n)} \cup \mathcal{R}_{1,2}^{\prime (n)}  \big) \big| - \big| \mathcal{R}^{(n)}_{1} \big|
} \nonumber \\*
& = & \big| \mathcal{G}^{(n)}_{0} \big| - \big| \mathcal{C}^{(n)}_{1,2} \big| - \big| \mathcal{C}^{(n)}_{2} \big| + \big| \mathcal{G}^{(n)}_{1} \big| \nonumber \\
&  &  - \Big( \big| \mathcal{C}^{(n)}_{1} \big| - \big| \mathcal{G}^{(n)}_{2} \big| -  \big| \mathcal{C}^{(n)}_{2} \big|  + \big| \mathcal{G}^{(n)}_{1} \big|  \Big) \nonumber \\
& = & \big| \mathcal{G}^{(n)}_{0} \big| - \big| \mathcal{C}^{(n)}_{1,2} \big|  -  \big| \mathcal{C}^{(n)}_{1} \big| + \big| \mathcal{G}^{(n)}_{2} \big|   \nonumber \\
& \geq &  0. \nonumber
\end{IEEEeqnarray}
Also, for $i\in [1,L]$, we define $\bar{\Theta}_{1,i}^{(V)}$ as any part $\bar{\Theta}_{i}^{(V)}$ with size $| \mathcal{C}^{(n)}_{1} | -  ( | \mathcal{C}^{(n)}_{2} |  - | \mathcal{G}^{(n)}_{1} | )$, and $\bar{\Theta}_{2,i}^{(V)}$ as the remaining part with size $|\mathcal{C}_{2}^{(n)} \big| - \big| \mathcal{G}_{1}^{(n)}|$. 

Thus, according to Algorithm~\ref{alg:formA},  instead of repeating $\Psi_{2,i-1}^{(V)}$, that is, the part of $\Psi_{i-1}^{(V)}$ that does not fit in $\tilde{A}_{i}^n[\mathcal{G}_{1}^{(n)}]$, in a specific part of $\tilde{A}_{i}^n[\mathcal{G}_{0}^{(n)}]$, the encoder stores $\Psi_{2,i-1}^{(V)} \oplus \bar{\Theta}_{2,i+1}^{(V)}$ into $\tilde{A}_{i}[\mathcal{R}_{1,2}^{\prime (n)}] \subseteq \tilde{A}_{i}[\mathcal{G}_{0}^{(n)}]$, where $\bar{\Theta}_{2,i+1}^{(V)}$ denotes part of those elements of $\bar{\Theta}_{i+1}^{(V)}$ that do not fit in $\tilde{A}_{i}^n[\mathcal{G}_{2}^{(n)}]$. Furthermore, as in Case C, since $\mathcal{I}^{(n)} \cap \mathcal{G}^{(n)}_2 = \emptyset$, we have $\Pi_{i}^{(V)} = \varnothing$.


%
%
%
%


The sets that form the partition of $\mathcal{G}^{(n)}$ in Case D can be seen in Figure~\ref{fig:EncCasD}, which also displays the encoding process that aims to construct of $\tilde{A}_{1:L}[\mathcal{C}^{(n)} \cup \mathcal{G}^{(n)}]$.

\begin{figure}[t]
\centering
\begin{overpic}[width=0.8\linewidth]{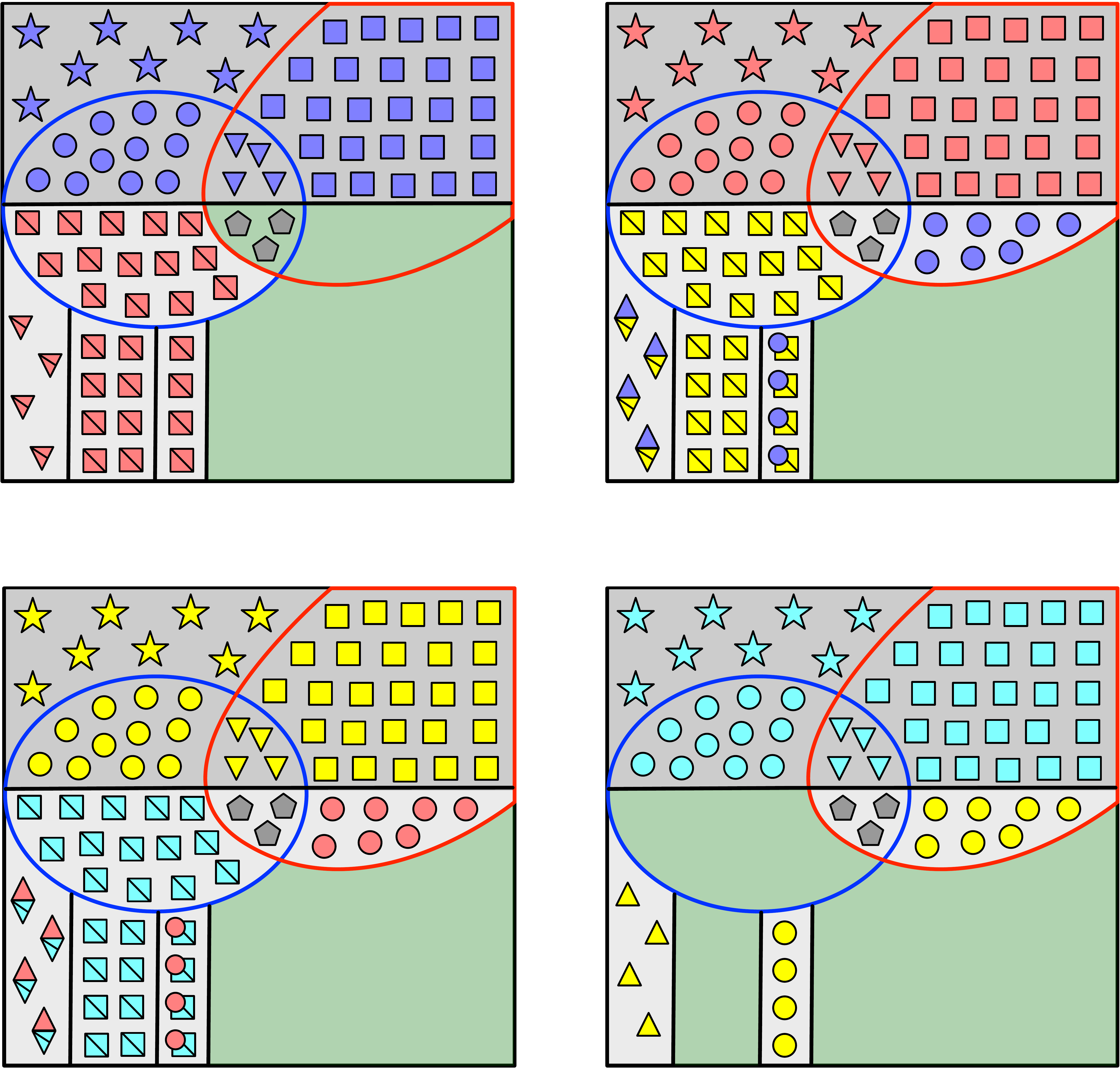}
\put (16,47) {\small Block 1}
\put (69,47) {\small Block 2}
\put (16,-5) {\small Block 3}
\put (69,-5) {\small Block 4}
\end{overpic}
\vspace{0.7cm}
\caption{For Case D, graphically representation of the encoding that leads to the construction of $\tilde{A}_{1:L}\big[ \mathcal{H}_V^{(n)} \big]$ when $L=4$. Consider the block 2, $\mathcal{R}^{(n)}_{1}$, $\mathcal{R}^{(n)}_{2}$, $\mathcal{R}^{(n)}_{1,2}$, $\mathcal{R}^{\prime (n)}_{1,2}$ and $\mathcal{R}^{(n)}_{\Lambda}$ are those areas filled with yellow squares, blue circles, blue and yellow diamonds, yellow squares overlapped by blue circles, and gray pentagons, respectively; the set $\mathcal{I}^{(n)}$ is the green filled area. At block $i \in [1,L]$, $W_i$ is represented by symbols of the same color (e.g., red symbols at block 2), and $\Theta_i^{(V)}$, $\Psi_i^{(V)}$ and $\Gamma_i^{(V)}$ are represented by squares, circles and triangles respectively. Also, $\bar{\Theta}_i^{(V)}$ and $\bar{\Gamma}_i^{(V)}$ are denoted by squares and triangles, respectively, with a line through them. At block $i \in [2,3]$, $\Gamma_{i-1,1}^{(V)} \oplus \bar{\Gamma}_{1,i+1}^{(V)}$ are represented by diamonds, and the squares overlapped by circles denote $\Psi_{2,i-1}^{(V)} \oplus \bar{\Theta}_{2,i+1}^{(V)}$. At block $i \in [1,L]$, $S_i$ is stored into those entries that belong to the green area. $\Lambda_1^{(V)}$ is denoted by gray pentagons and is repeated in all blocks. The sequences $\Upsilon_{(1)}^{(V)}$ and $\Upsilon_{(2)}^{(V)}$ are the entries inside the red and blue curve at block 1 and $L$, respectively.}\label{fig:EncCasD} 
\end{figure}
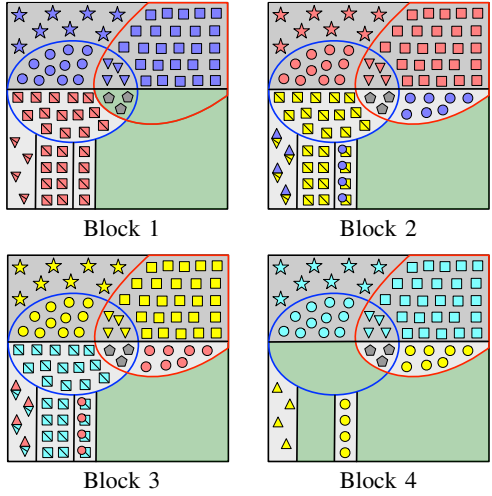

\subsection{Channel prefixing}\label{sec:PCSCP}
For $i \in [1,L]$, let $R_i$ be a uniformly distributed vector of length $| \mathcal{H}^{(n)}_{X|V} \setminus \mathcal{H}^{(n)}_{X|VZ} |$ that represents the randomization sequence (local randomness). Moreover, let $\Lambda^{(X)}_0$ be a uniformly distributed random sequence of size $| \mathcal{H}^{(n)}_{X|VZ} |$. The channel prefixing aims to construct $\tilde{X}_i^n = \tilde{T}_i^n G_n$ and is summarized in Algorithm~\ref{alg:polarchannelpref}. 

\begin{algorithm}[h!]
\caption{Function \texttt{pb\_ch\_pref}}\label{alg:polarchannelpref}
\begin{small}
\begin{algorithmic}[1]
\Require $\tilde{V}_i^n$, $R_i$, $\Lambda^{(X)}_{i-1}$
\State $\tilde{T}_i[\mathcal{H}^{(n)}_{X|VZ}] \leftarrow \Lambda^{(X)}_{i-1}$
\State $\tilde{T}_i [\mathcal{H}^{(n)}_{X|V} \cap (\mathcal{H}^{(n)}_{X|VZ})^{\text{C}}] \leftarrow R_{i}$
\For{$j \in \big( \mathcal{H}^{(n)}_{X|V} \big)^{\text{C}}$}
\If{$j \in \big( \mathcal{H}^{(n)}_{X|V} \big)^{\text{C}}  \setminus \mathcal{L}^{(n)}_{X|V}$}
\State $\tilde{T}(j) \leftarrow \! p_{T(j)|T^{1:j-1}V^n} \! \big( \tilde{t}_i(j) \big| \tilde{t}_i^{1:j-1}, \tilde{v}_i^n  \big)$
\ElsIf{$j \in \mathcal{L}^{(n)}_{X|V}$}
\State $\tilde{T}(j) \leftarrow \xi^{(X)}_{(j)} \big( \tilde{t}_i^{1:j-1}, \tilde{v}_i^n \big)$ 
\EndIf
\EndFor 
\State $\tilde{X}_{i}^n \leftarrow \tilde{T}_i^n G_n$
\State $\Lambda^{(X)}_{i} \leftarrow \tilde{T}_i [\mathcal{H}^{(n)}_{X|V} \setminus \mathcal{H}^{(n)}_{X|VZ}]$
\State \textbf{return} $\tilde{X}_{i}^n$ and $\Lambda^{(X)}_{i}$
\end{algorithmic}
\end{small}
\end{algorithm}

Notice that the sequence $\Lambda^{(X)}_0$ is copied in $\tilde{T}_i[\mathcal{H}^{(n)}_{X|VZ}]$ at any block $i \in [1,L]$, while $R_i$ is stored into $\tilde{T}_i[\mathcal{H}^{(n)}_{X|V} \cap (\mathcal{H}^{(n)}_{X|VZ})^{\text{C}}]$. After obtaining $\tilde{T}_i[\mathcal{H}^{(n)}_{X|V}]$, and given the sequence $\tilde{V}_i^n \triangleq \tilde{A}_i^n G_n$, the encoder forms the remaining entries of $\tilde{T}_i^n$, that is, $\tilde{T}_i [(\mathcal{H}_{X|V}^{(n)})^{\text{C}}]$ as follows. If $j \in (\mathcal{H}_{X|V}^{(n)})^{\text{C}} \setminus \mathcal{L}_{X|V}^{(n)}$, the encoder randomly draws $\tilde{T}_{i}(j)$ from $p_{T(j)|T^{1:j-1}|V^n}$ corresponding to the joint distribution of the original DMS --see \eqref{eq:distDMS}--. Otherwise, if $j \in \mathcal{L}_{X|V}^{(n)}$, it constructs $\tilde{T}_{i}(j)$ deterministically by using SC encoding as in \cite{7447169}. Thus, we define the SC encoding function $\xi^{(X)}_{(j)}: \{0,1\}^{(j-1)\cdot n} \rightarrow \{ 0,1 \}$ as 
\begin{IEEEeqnarray}{rCl}
\IEEEeqnarraymulticol{3}{l}{%
\xi_{(j)}^{(X)} \big( {t}^{1:j-1} , {v}^n \big)
} \nonumber \\* 
& \triangleq & \argmax_{t \in \mathcal{X}} p_{T(j)|T^{1:j-1}V^n} \left( {t} \left| {t}^{1:j-1}, v^n \right. \right), \label{eq:scdeterministicX}
\end{IEEEeqnarray}
Hence, besides the required local randomness, only $\Lambda^{(X)}_0$ and $\tilde{T}_{1:L}[ ( \mathcal{H}^{(n)}_V )^{\text{C}}\setminus \mathcal{L}^{(n)}_{X|V} ]$ are constructed randomly.

\subsection{Decoding}\label{sec:decoding}
Consider that $(\Upsilon_{(k)}^{(V)}, \Phi_{(k),1:L}^{(V)})$, $k \in [1,L]$, is available to the $k$-th legitimate receiver. In the decoding process, both legitimate receivers form the estimates $\hat{A}^n_{1:L}$ of $\tilde{A}^n_{1:L}$ and then output the messages $(\hat{W}_{1:L},\hat{S}_{1:L})$.

\subsubsection{Legitimate receiver 1}\label{sec:PCSD1}
This receiver forms the estimates $\hat{A}^n_{1:L}$ by going forward, i.e., from $\hat{A}^n_{1}$ to $\hat{A}^n_{L}$, and this process is summarized in Algorithm~\ref{alg:decoding1}.

\begin{algorithm}
\caption{Decoding at legitimate receiver 1}\label{alg:decoding1}
\begin{small}
\begin{algorithmic}[1]
\Require $\Upsilon_{(1)}^{(V)}$, $\Phi_{(1),1:L}^{(V)}$, $\kappa_{\Theta}^{(V)}$ and  $\kappa_{\Gamma}^{(V)}$, and $\tilde{Y}_{(1),1:L}^n$.
\State $\hat{\Lambda}_{1:L-1}^{(V)} \leftarrow \Upsilon_{(1)}^{(V)}$  
\State $\hat{A}_1^n \leftarrow \big( \Upsilon_{(1)}^{(V)}, \Phi_{(1),1}^{(V)}, \tilde{Y}_{(1),1}^n \big)$
\For{$i = 1$ \text{to} $L-1$} 
\State $\hat{\Psi}_{i}^{(V)} \leftarrow \hat{A}_i [\mathcal{C}_{2}^{(n)}]$
\State $\hat{\Gamma}_{i}^{(V)} \leftarrow \hat{A}_i[\mathcal{C}_{1,2}^{(n)}]$
\State $\hat{\bar{\Theta}}_{i+1}^{(V)} \leftarrow \big( \hat{A}_i[\mathcal{R}_{1}^{(n)}], \hat{A}_i[\mathcal{R}_{1,2}^{\prime (n)}] \oplus \hat{\Psi}_{2,i-1}^{(V)} \big)$
\State $\hat{\Theta}_{i+1}^{(V)} \leftarrow \hat{\bar{\Theta}}_{i+1}^{(V)} \oplus \kappa_{\Theta}^{(V)}$
\State $\hat{\bar{\Gamma}}_{i+1}^{(V)} \leftarrow \big( \hat{A}_i[\mathcal{R}_{1,2}^{(n)}]  \oplus \hat{\Gamma}_{1,i-1}^{(V)}, \hat{A}_i[\mathcal{R}_{1}^{\prime (n)}] \big)$
\State $\hat{\Gamma}_{i+1}^{(V)} \leftarrow \hat{\bar{\Gamma}}_{i+1}^{(V)} \oplus \kappa_{\Gamma}^{(V)}$
\State $\hat{\Pi}_{i}^{(V)} \leftarrow \hat{A}_{i}[ \mathcal{I}^{(n)} \cap \mathcal{G}_{2}^{(n)}]$
\State $\hat{\Upsilon}_{(1),i+1}^{\prime (V)} \leftarrow  \big(\hat{\Psi}_{1,i}^{(V)} , \hat{\Gamma}_{2,i}^{(V)} ,\hat{\Theta}_{i+1}^{(V)} ,\hat{\Gamma}_{i+1}^{(V)} , \hat{\Pi}_{i}^{(V)} , \hat{\Lambda}_{i}^{(V)} \big)$
\State $\hat{A}_{i+1}^n \leftarrow \big( \hat{\Upsilon}_{(1),i+1}^{\prime (V)}, \Phi_{(1),i+1}^{(V)}, \tilde{Y}_{(1),i+1}^n \big)$
\EndFor
\end{algorithmic}
\end{small}
\end{algorithm}

In all cases (among Case A to Case D), receiver~1 constructs $\hat{A}^n_{1}$ as follows. Given $\smash{\Upsilon_{(1)}^{(V)}}$ (all the elements inside the red curve at block 1 in Figures~\ref{fig:EncCasA}--\ref{fig:EncCasD}) and $\Phi_{(1),1}^{(V)}$, notice that receiver~1 knows $\smash{\tilde{A}_{1}[ (\mathcal{L}_{V|Y_{(1)}}^{(n)} )^{\text{C}}]}$. Therefore, from $(\Upsilon_{(1)}^{(V)}$, $\Phi_{(1),1}^{(V)})$ and channel observations $\tilde{Y}_{(1),1}^n$, receiver~1 performs SC decoding for source coding with side information \cite{arikan2010source} to form $\hat{A}^n_{1}$. Moreover, since $\Lambda_1^{(V)} \subseteq \Upsilon_{(1)}^{(V)}$ has been replicated in all blocks, legitimate receiver~1 gets $\smash{\hat{\Lambda}_{1:L-1}^{(V)}}$ (gray pentagons at all blocks).

For $i \in [1,L-1]$, consider the construction of $\hat{A}_{i+1}^n$. First, from $\hat{A}_{i}^n$ that has already been estimated, receiver~1 gets $\hat{\Psi}_{i}^{(V)} = \hat{A}_i[\mathcal{C}_2^{(n)}]$ (e.g., blue or red circles at block 1 or 2 respectively in Figures~\ref{fig:EncCasA}--\ref{fig:EncCasD}) and $\hat{\Gamma}_{i}^{(V)}= \hat{A}_i[\mathcal{C}_{1,2}^{(n)}]$ (blue or red triangles at block 1 or 2 respectively). 

Also, from $\hat{A}_i^n$, receiver~1 obtains $\hat{\Theta}_{i+1}^{(V)}$ as follows. At block~1, in all cases it gets $\smash{\hat{\bar{\Theta}}_{2}^{(V)}} = \tilde{A}_1[\mathcal{R}_1^{(n)} \cup \mathcal{R}_{1,2}^{\prime (n)}]$ (all the red squares with a line through them at block 1 in Figures~\ref{fig:EncCasA}--\ref{fig:EncCasD}). At block $i \in [2,L-1]$, we distinguish two situations. In Case D, receiver~1 gets $\smash{\hat{\bar{\Theta}}_{1,i+1}^{(V)} = \hat{A}_i[\mathcal{R}^{(V)}_1]}$ (e.g., yellow squares with a line through them at block~2 in Figure~\ref{fig:EncCasD}) and $\smash{\hat{\Psi}_{2,i-1}^{(V)} \oplus \hat{\bar{\Theta}}_{2,i+1}^{(V)}}$ (yellow squares with a line through them overlapped by blue circles). Since $\hat{\Psi}_{2,i-1}^{(V)} \subset \hat{A}_{i-1}^n$ (blue circles) has already been estimated, receiver~1 obtains $\smash{\hat{\bar{\Theta}}_{2,i+1}^{(V)}}$ (yellow squares with a line through them). Otherwise, in other cases, receiver~1 obtains $\hat{\bar{\Theta}}_{i+1}^{(V)} = \hat{A}_i[\mathcal{R}_{1}^{(n)}]$ directly (yellow squares with a line through them at block 2 in Figures~{\ref{fig:EncCasA}--\ref{fig:EncCasC}}). 
Then, given $\hat{\bar{\Theta}}_{i+1}^{(V)} = [\hat{\bar{\Theta}}_{1,i+1}^{(V)}, \hat{\bar{\Theta}}_{2,i+1}^{(V)}]$,
in all cases receiver~1 recovers the sequence $\hat{\Theta}_{i+1}^{(V)}=\hat{\bar{\Theta}}_{i+1}^{(V)} \oplus \kappa_{\Theta}^{(V)}$.

From $\hat{A}_i^n$, receiver~1 also obtains $\hat{\Gamma}_{i+1}^{(V)}$ as follows.
At block~$1$, in all cases it gets $\hat{\bar{\Gamma}}_{2}^{(V)} = \tilde{A}_1 [\mathcal{R}^{(n)}_{1,2} \cup \mathcal{R}^{\prime (n)}_{1}]$ directly (e.g., all red triangles with a line through them at block~1 in Figures~\ref{fig:EncCasA}--\ref{fig:EncCasD}).
At block $i \in [2,L-1]$, in all cases receiver~1 obtains $\hat{\Gamma}_{1,i-1}^{(V)} \oplus \hat{\bar{\Gamma}}_{1,i+1}^{(V)} = \tilde{A}_i[\mathcal{R}^{(n)}_{1,2}]$ (e.g., blue and yellow diamonds with a line through them at block~2). Since $\hat{\Gamma}_{1,i-1}^{(V)} \subset \tilde{A}_{i-1}^n$ (blue triangles) has already been estimated, receiver~1 obtains $\hat{\bar{\Gamma}}_{1,i+1}^{(V)} = \tilde{A}_i[\mathcal{R}^{(n)}_{1,2}] \oplus \hat{\Gamma}_{1,i-1}^{(V)}$ (yellow triangles with a line through them). Also, only in Case B, receiver~1 obtains $\hat{\bar{\Gamma}}_{2,i+1}^{(V)} = \tilde{A}_{i}[\mathcal{R}^{\prime (n)}_{1}]$ (remaining yellow triangles with a line through them at block~3 in Figure~\ref{fig:EncCasB}).
Then, given $\hat{\bar{\Gamma}}_{i+1}^{(V)} = [\hat{\bar{\Gamma}}_{1,i+1}^{(V)}, \hat{\bar{\Gamma}}_{2,i+1}^{(V)}]$, 
in all cases receiver~1 recovers the sequence $\hat{\Gamma}_{i+1}^{(V)}=\hat{\bar{\Gamma}}_{i+1}^{(V)} \oplus \kappa_{\Gamma}^{(V)}$.

Only in cases A and B, it gets $\hat{\Pi}_i^{(V)}=\hat{A}_i[\mathcal{I}^{(n)} \cap \mathcal{G}_2^{(n)}]$ (e.g. purple crosses at block 2 in Figures~\ref{fig:EncCasA}--\ref{fig:EncCasB}). Let 
\begin{IEEEeqnarray}{rCl}
\Upsilon^{\prime (V)}_{(1),i+1} & \triangleq & \big[ \hat{\Psi}_{1,i}^{(V)}, \hat{\Gamma}_{2,i}^{(V)}, \hat{\Theta}_{i+1}^{(V)}, \hat{\Gamma}_{i+1}^{(V)}, \hat{\Pi}_{i}^{(V)} , \hat{\Lambda}_{i}^{(V)} \big]. \IEEEeqnarraynumspace 
\end{IEEEeqnarray}
Notice that $\Upsilon^{\prime (V)}_{(1),i+1} = \tilde{A}_{i+1}[\mathcal{H}_V^{(n)} \cap (\mathcal{L}_{V|Y_{(1)}}^{(n)})^{\text{C}}]$ (all the elements inside red curve at block~$i+1$ in Figures~\ref{fig:EncCasA}--\ref{fig:EncCasD}). Thus, receiver~1 performs SC decoding to form $\hat{A}^n_{i+1}$ by using $\Upsilon^{\prime (V)}_{(1),i+1}$, $\Phi_{(1),i+1}^{(V)}$ and the observations $\tilde{Y}_{(1),i+1}^n$.

\subsubsection{Legitimate receiver 2}\label{sec:PCSD2}
This receiver forms the estimates $\hat{A}^n_{1:L}$ by going backward, i.e., from $\hat{A}^n_{L}$ to $\hat{A}^n_{1}$, and this process is summarized in Algorithm~\ref{alg:decoding2}.

\begin{algorithm}[b]
\caption{Decoding at legitimate receiver 2}\label{alg:decoding2}
\begin{small}
\begin{algorithmic}[1]
\Require $\Upsilon_{(2)}^{(V)}$, $\Phi_{(2),1:L}^{(V)}$, $\kappa_{\Theta}^{(V)}$ and  $\kappa_{\Gamma}^{(V)}$, and $\tilde{Y}_{(2),1:L}^n$.
\State $\hat{A}_L^n \leftarrow \big( \Upsilon_{(2)}^{(V)}, \Phi_{(2),L}^{(V)}, \tilde{Y}_{(2),L}^n \big)$ 
\State $\hat{\Lambda}_{1:L-1}^{(V)} \leftarrow \hat{A}_L^n$
\For{$i = L$ \text{to} $2$} 
\State $\hat{\bar{\Theta}}_{i}^{(V)} \leftarrow \hat{A}_i [\mathcal{C}_{1}^{(n)}] \oplus \kappa_{\Theta}^{(V)}$
\State $\hat{\bar{\Gamma}}_{i}^{(V)} \leftarrow \hat{A}_i [\mathcal{C}_{1,2}^{(n)}] \oplus \kappa_{\Gamma}^{(V)}$
\State $\hat{\Psi}_{i-1}^{(V)} \leftarrow \big( \hat{A}_i[\mathcal{R}_{2}^{(n)}], \hat{A}_i[\mathcal{R}_{1,2}^{\prime (n)}] \oplus \hat{\bar{\Theta}}_{2,i+1}^{(V)} \big)$
\State $\hat{\Gamma}_{i-1}^{(V)} \leftarrow \big( \hat{A}_i[\mathcal{R}_{1,2}^{(n)}] \oplus \hat{\bar{\Gamma}}_{1,i+1}^{(V)}, \hat{A}_i[\mathcal{R}_{2}^{\prime (n)}] \big)$
\State $\hat{\Pi}_{i-1}^{(V)} \leftarrow \hat{A}_{i}[ \mathcal{R}_{\text{S}}^{(n)}]$
\State $\Upsilon^{\prime(V)}_{(2),i-1} \leftarrow \big( \hat{\bar{\Theta}}_{1,i}^{(V)}, \hat{\bar{\Gamma}}_{2,i}^{(V)}, \hat{\Psi}_{i-1}^{(V)}, \hat{\Gamma}_{i-1}^{(V)}, \hat{\Pi}_{i-1}^{(V)} , \hat{\Lambda}_{i-1}^{(V)} \big)$
\State $\hat{A}_{i-1}^n \leftarrow \big( \Upsilon^{\prime(V)}_{(2),i-1} , \Phi_{(2),i-1}^{(V)}, \tilde{Y}_{(2),i-1}^n \big)$
\EndFor
\end{algorithmic}
\end{small}
\end{algorithm}

In all cases (among Case A to Case D), receiver~2 constructs $\hat{A}^n_{L}$ as follows. Given $\smash{\Upsilon_{(2)}^{(V)}}$ (all the elements inside blue curve at block 4 in Figures~\ref{fig:EncCasA}--\ref{fig:EncCasD}) and $\Phi_{(2),L}^{(V)}$, notice that receiver~2 knows $\smash{\tilde{A}_{L}[ (\mathcal{L}_{V|Y_{(2)}}^{(n)} )^{\text{C}}]}$. Hence, from $(\Upsilon_{(2)}^{(V)}$, $\Phi_{(2),L}^{(V)})$ and channel observations $\tilde{Y}_{(2),L}^n$, receiver~2 performs SC decoding for source coding with side information to form $\hat{A}^n_{L}$. Since $\Lambda_1^{(V)}$ has been replicated in all blocks, from $\tilde{A}_{L}^n$ receiver~2 obtains $\smash{\hat{\Lambda}_{1:L-1}^{(V)}} =\tilde{A}_{L}[\mathcal{R}_{\Lambda}^{(n)}]$ (gray pentagons at all blocks).

For $i \in [2,L]$, consider the construction of $\hat{A}_{i-1}^n$. First, from $\hat{A}_{i}^n$ that has already been estimated, legitimate receiver~2 obtains the sequence $\hat{\Theta}_{i}^{(V)}=\tilde{A}_i[\mathcal{C}_1^{(n)}]$ (e.g., cyan or yellow squares at block 4 or 3  respectively in Figures~\ref{fig:EncCasA}--\ref{fig:EncCasD}). Given $\hat{\Theta}_{i}^{(V)}$, 
the encoder computes $\hat{\bar{\Theta}}_{i}^{(V)} = \hat{\Theta}_{i}^{(V)} \oplus \kappa_{\Theta}^{(V)}$ (corresponding previous squares with a line through them). Also, receiver~2 obtains $\hat{\Gamma}_{i}^{(V)}= \tilde{A}_i[\mathcal{C}_{1,2}^{(n)}]$ (cyan or yellow triangles at block 4 or 3 respectively in Figures~\ref{fig:EncCasA}--\ref{fig:EncCasD}). Given this sequence, 
receiver~2 computes $\hat{\bar{\Gamma}}_{i}^{(V)} = \hat{\Gamma}_{i}^{(V)} \oplus \kappa_{\Gamma}^{(V)}$ (corresponding previous triangles with a line through them).

Also, from $\hat{A}_i^n$, receiver~2 obtains $\hat{\Theta}_{i+1}^{(V)}$ as follows. At block $L$, in all cases it gets $\hat{\Psi}_{L-1}^{(V)}= \tilde{A}_i [\mathcal{R}^{(n)}_{2} \cup \mathcal{R}^{\prime (n)}_{1,2}]$ directly (all yellow circles at block $L$ in Figures~\ref{fig:EncCasA}--\ref{fig:EncCasD}). At block $i \in [2,L-1]$, we distinguish two situations. In Case D, it obtains $\hat{\Psi}_{1,i-1}^{(V)} = \tilde{A}_{i}[\mathcal{R}^{(n)}_{2}]$ (e.g., red circles at block~3 in Figure~\ref{fig:EncCasD}) and $\hat{\Psi}_{2,i-1}^{(V)} \oplus \hat{\bar{\Theta}}_{2,i+1}^{(V)} = \tilde{A}_i[\mathcal{R}^{\prime (n)}_{1,2}]$ (cyan squares with a line through them overlapped by red circles). Since $\smash{\hat{\bar{\Theta}}_{2,i+1}^{(V)}}$ (cyan squares with a line through them) has already been estimated, receiver~2 obtains $\hat{\Psi}_{2,i-1}^{(V)} = \tilde{A}_i[\mathcal{R}^{\prime (n)}_{1,2}] \oplus \hat{\bar{\Theta}}_{2,i+1}^{(V)}$ (red circles). Otherwise, in other cases, receiver~2 obtains directly $\hat{\Psi}_{i-1}^{(V)} = \tilde{A}_{i}[\mathcal{R}_{2}^{(n)}]$ (e.g., red circles at block 3 in Figures~{\ref{fig:EncCasA}--\ref{fig:EncCasC}}).

From $\hat{A}_i^n$, receiver~2 also obtains $\hat{\Gamma}_{i-1}^{(V)}$ as follows.
At block~$L$, in all cases it gets $\hat{\bar{\Gamma}}_{L-1}^{(V)} = \tilde{A}_L [\mathcal{R}^{(n)}_{1,2} \cup \mathcal{R}^{\prime (n)}_{2}]$ (e.g., all yellow triangles at block~$L$ in Figures~\ref{fig:EncCasA}--\ref{fig:EncCasD}).
At block $i \in [2,L-1]$, in all cases receiver~2 obtains $\hat{\Gamma}_{1,i-1}^{(V)} \oplus \hat{\bar{\Gamma}}_{1,i+1}^{(V)} = \tilde{A}_i[\mathcal{R}^{(n)}_{1,2}]$ (e.g., red and cyan diamonds with a line through them at block~3). Since $\hat{\bar{\Gamma}}_{1,i+1}^{(V)}$ (cyan triangles with a line through them) has already been estimated, receiver~2 obtains $\hat{\Gamma}_{1,i-1}^{(V)} = \tilde{A}_i[\mathcal{R}^{(n)}_{1,2}] \oplus \hat{\bar{\Gamma}}_{1,i+1}^{(V)}$ (red triangles). Also, only in Case B, receiver~2 obtains the sequence $\hat{\Gamma}_{2,i-1}^{(V)} = \tilde{A}_{i}[\mathcal{R}^{\prime (n)}_{2}]$ (remaining red triangles at block~3 in Figure~\ref{fig:EncCasB}).

Finally, only in Case A and Case B, receiver~2 obtains the sequence $\smash{\hat{\Pi}_{i-1}^{(V)} = \tilde{A}_i[\mathcal{R}_{\text{S}}]}$ (e.g., purple crosses at block 3 in Figures~\ref{fig:EncCasA}--\ref{fig:EncCasB}). Let 
\begin{IEEEeqnarray}{rCl}
\Upsilon^{\prime(V)}_{(2),i-1} & \triangleq & \big[ \hat{\bar{\Theta}}_{1,i}^{(V)}, \hat{\bar{\Gamma}}_{2,i}^{(V)}, \hat{\Psi}_{i-1}^{(V)}, \hat{\Gamma}_{i-1}^{(V)}, \hat{\Pi}_{i-1}^{(V)} , \hat{\Lambda}_{i-1}^{(V)} \big]. \IEEEeqnarraynumspace 
\end{IEEEeqnarray}
Notice that\footnote{We have $\smash{\Upsilon^{\prime (V)}_{(2),i-1} \supseteq \tilde{A}_{i-1}[\mathcal{H}_V^{(n)} \cap (\mathcal{L}_{V|Y_{(2)}}^{(n)})^{\text{C}}]}$ because part of $\hat{\bar{\Theta}}_{1,i}^{(V)}$ in cases C and D could be copied in some entries of $\tilde{A}_{i-1}[\mathcal{G}_0^{(n)}]$.} $\Upsilon^{\prime (V)}_{(2),i-1} \supseteq \tilde{A}_{i-1}[\mathcal{H}_V^{(n)} \cap (\mathcal{L}_{V|Y_{(2)}}^{(n)})^{\text{C}}]$ (all the elements inside blue curve at block~$i-1$ in Figures~\ref{fig:EncCasA}--\ref{fig:EncCasD}). Thus, receiver~2 performs SC decoding to form $\hat{A}^n_{i-1}$ by using $\smash{\Upsilon^{\prime (V)}_{(2),i-1}}$, $\smash{\Phi_{(2),i-1}^{(V)}}$ and $\smash{\tilde{Y}_{(2),i-1}^n}$.


\section{Performance of the Polar Coding Scheme}\label{sec:performance}
The analysis of the polar coding scheme of Section~\ref{sec:PCS} leads to the following theorem.

\begin{theorem}\label{th:th1}
Let $(\mathcal{X}, p_{Y_{(1)}Y_{(2)} Z|X}, \mathcal{Y}_{(1)} \times \mathcal{Y}_{(2)} \times \mathcal{Z})$  be an arbitrary WBC such that $\mathcal{X} \in \{0,1\}$. The polar coding scheme described in Section~\ref{sec:PCS} achieves the corner point in \eqref{eq:targetrate} of the region $\mathfrak{R}_{\text{\emph{CI-WBC}}}$ defined in Proposition~\ref{prop:MIB}.
\end{theorem}

The proof of Theorem~\ref{th:th1} follows in four steps and is provided in the following subsections. In Section~\ref{sec:performance_rates} we show that the polar coding scheme approaches \eqref{eq:targetrate}. Then, in Section~\ref{sec:performance_rates} we prove that the joint distribution of $(\tilde{V}_i^n, \tilde{X}_i^n,\tilde{Y}_{(1),i}^n,\tilde{Y}_{(2),i}^n,\tilde{Z}_{i}^n)$ is asymptotically indistinguishable of the one of the original DMS. Finally, in Section~\ref{sec:performance_reliability} and Section~\ref{sec:performance_secrecy} we show that the polar coding scheme satisfies the reliability and the secrecy conditions \eqref{eq:reliabilitycond} and \eqref{eq:secrecycond} respectively.

\subsection{Transmission Rates}\label{sec:performance_rates}
We prove that the polar coding scheme described in Section~\ref{sec:PCS} approaches the rate tuple \eqref{eq:targetrate}. Also, we show that the overall length of the secret keys $\kappa_{\Theta}^{(V)}$, $\kappa_{\Gamma}^{(V)}$, $\kappa_{\Upsilon\Phi_{(1)}}^{(V)}$ and $\kappa_{\Upsilon\Phi_{(2)}}^{(V)}$, and the additional randomness used in the encoding (besides the randomization sequences) are asymptotically negligible in terms of rate.

\subsubsection{Private message rate}\label{sec:rate}
For $i\in[1,L]$, we have $W_{i} = \tilde{A}_{i}[ \mathcal{C}^{(n)}]$. According to the definition of $\mathcal{C}^{(n)}$ in \eqref{eq:sC}, and since $\mathcal{H}^{(n)}_{V|Z} \subseteq \mathcal{H}^{(n)}_{V}$, the rate of $W_{1:L}$ is
\begin{IEEEeqnarray}{rCl}
\frac{1}{nL}\sum_{i=1}^L |W_i| & = & \frac{1}{n} \Big| \mathcal{H}^{(n)}_{V} \cap \big( \mathcal{H}^{(n)}_{V|Z} \big)^{\text{C}} \Big| \nonumber \\
& = & \frac{1}{n} \big| \mathcal{H}^{(n)}_{V} \big| - \frac{1}{n} \big| \mathcal{H}^{(n)}_{V|Z} \big| \nonumber \\
& \IEEEeqnarraymulticol{2}{l}{%
\, \, \xrightarrow{n \rightarrow \infty}  H(V) - H(V|Z) }\nonumber \\
& = & I(V;Z), \nonumber
\end{IEEEeqnarray}
where the limit holds by \cite[Theorem~1]{arikan2010source}. Therefore, the private message rate achieved by the polar coding scheme is $R_{W}$ in \eqref{eq:targetrate}.

\subsubsection{Confidential message rate}
From Section~\ref{sec:formAG}, in all cases $|S_1| = \tilde{A}_1[\mathcal{I}^{(n)}  \mathcal{G}^{(n)}_{1} \cup \mathcal{G}^{(n)}_{1,2}]$; for $i\in[2,L-1]$, $S_{i} = \tilde{A}_{i}[ \mathcal{I}^{(n)}]$; and $S_L=\tilde{A}_L[\mathcal{I}^{(n)}  \mathcal{G}^{(n)}_2]$. Thus, we have
\begin{IEEEeqnarray}{rCl}
\IEEEeqnarraymulticol{3}{l}{
\frac{1}{nL} \sum_{i=1}^L |S_i|
}\nonumber \\
& = & \frac{( L \! - \! 2  ) \big| \mathcal{I}^{(n)} \big| \! + \! \big| \mathcal{I}^{(n)}  \mathcal{G}^{(n)}_{1} \cup \mathcal{G}^{(n)}_{1,2} \big| \! + \!   \big| \mathcal{I}^{(n)} \cup \mathcal{G}^{(n)}_{2} \big|}{nL} \nonumber \\
& = & \frac{1}{n} \big| \mathcal{I}^{(n)} \big| + \frac{1}{nL} \Big( \big| \mathcal{G}^{(n)}_{1} \big| + \big| \mathcal{G}^{(n)}_{2} \big| + \big| \mathcal{G}^{(n)}_{1,2} \big| \Big) \nonumber \\
& = & \frac{1}{n} \big| \mathcal{I}^{(n)} \big| + \frac{1}{nL}\big| \mathcal{G}^{(n)} \setminus \mathcal{G}^{(n)}_0  \big|  \nonumber \\
& \stackrel{(a)}{=} & \frac{\big| \mathcal{G}^{(n)}_0 \big| \!  + \!  \big| \mathcal{G}^{(n)}_{2} \big| \!  - \!  \big| \mathcal{R}^{(n)}_{1,2} \big| \! - \!  \big| \mathcal{R}^{\prime (n)}_{1,2} \big| \!  - \! \big| \mathcal{R}^{(n)}_{1} \big| \! - \! \big| \mathcal{R}^{\prime (n)}_{1} \big|}{n} \nonumber \\
& & +  \frac{1}{nL}\big| \mathcal{G}^{(n)} \setminus \mathcal{G}^{(n)}_0  \big| \nonumber \\
& \stackrel{(b)}{=} & \frac{\big| \mathcal{G}^{(n)}_0 \big| \!  + \!  \big| \mathcal{G}^{(n)}_{2} \big| \!  - \!  \big| \mathcal{C}^{(n)}_{1} \big| \! - \!  \big| \mathcal{C}^{(n)}_{1,2} \big| \!}{n} +  \frac{\big| \mathcal{G}^{(n)} \setminus \mathcal{G}^{(n)}_0  \big|}{nL} \nonumber \\
& \stackrel{(c)}{=} & \frac{\Big| \mathcal{H}^{(n)}_{V|Z} \cap \mathcal{L}^{(n)}_{V|Y_{(1)}} \Big| - \Big| \big( \mathcal{H}^{(n)}_{V|Z} \big)^{\text{C}} \cap \big( \mathcal{L}^{(n)}_{V|Y_{(1)}} \big)^{\text{C}} \Big|}{n} \nonumber \\
& & + \frac{\Big| \mathcal{H}^{(n)}_{V|Z} \cap \Big( \mathcal{L}^{(n)}_{V|Y_{(1)}}\cap   \mathcal{L}^{(n)}_{V|Y_{(2)}}\Big)^{\text{C}}\Big|}{nL} \nonumber \\
& \geq & \frac{\Big| \mathcal{H}^{(n)}_{V|Z} \cap \mathcal{L}^{(n)}_{V|Y_{(1)}} \Big| - \Big| \big( \mathcal{H}^{(n)}_{V|Z} \big)^{\text{C}} \cap \big( \mathcal{L}^{(n)}_{V|Y_{(1)}} \big)^{\text{C}} \Big|}{n} \nonumber \\
& & + \frac{\Big| \mathcal{H}^{(n)}_{V|Z} \Big| -  \Big| \big( \mathcal{L}^{(n)}_{V|Y_{(1)}} \big)^{\text{C}} \Big|}{nL} \nonumber \\
& = & \frac{\Big| \mathcal{H}^{(n)}_{V|Z} \Big| - \Big|  \big( \mathcal{L}^{(n)}_{V|Y_{(1)}} \big)^{\text{C}} \Big|}{n} + \frac{\Big| \mathcal{H}^{(n)}_{V|Z} \Big| -  \Big| \big( \mathcal{L}^{(n)}_{V|Y_{(1)}} \big)^{\text{C}} \Big|}{nL} \nonumber \\
& \IEEEeqnarraymulticol{2}{l}{%
\, \, \xrightarrow{n \rightarrow \infty}  H(V|Z) \! - \! H(V|Y_{(1)}) \! + \! \frac{H(V|Z) \! - \! H(V|Y_{(1)})}{L}
} \nonumber \\
& \IEEEeqnarraymulticol{2}{l}{%
\, \, \xrightarrow{L \rightarrow \infty}  I(V;Y_{(1)}) - I(V;Z)
}\nonumber 
\end{IEEEeqnarray}
where $(a)$ holds by the definition of $\mathcal{I}^{(n)}$ in \eqref{eq:ASetI}; $(b)$ holds because, in all cases, $|\mathcal{R}^{(n)}_{1,2}| + | \mathcal{R}^{\prime (n)}_{1} | = |\mathcal{C}_{1,2}^{(n)}|$ and $|\mathcal{R}^{(n)}_{1}| + | \mathcal{R}^{\prime (n)}_{1,2} | = |\mathcal{C}_{1}^{(n)}|$, $(c)$ follows from the partition of $\mathcal{H}_{V}^{(n)}$ defined in \eqref{eq:sG0}--\eqref{eq:sC12}; and the limit when $n$ goes to infinity holds by \cite[Theorem~1]{arikan2010source}. Hence, the coding scheme achieves the confidential message rate $R_{S}$ in \eqref{eq:targetrate}.

\subsubsection{Randomization sequence rate.}
For $i\in[1,L]$, we have $R_{i} = \tilde{T}_{i}[ \mathcal{H}^{(n)}_{X|V} \cap ( \mathcal{H}^{(n)}_{X|VZ} )^{\text{C}}]$. Therefore, we have
\begin{IEEEeqnarray}{rCl}
\frac{1}{nL}\sum_{i=1}^L |R_i| & = & \frac{1}{n} \Big|\mathcal{H}^{(n)}_{X|V} \cap \big( \mathcal{H}^{(n)}_{X|VZ} \big)^{\text{C}}\Big| \nonumber \\
& \stackrel{(a)}{=} & \frac{1}{n} \big| \mathcal{H}^{(n)}_{X|V} \big| - \frac{1}{n} \big| \mathcal{H}^{(n)}_{X|VZ} \big| \nonumber \\
& \IEEEeqnarraymulticol{2}{l}{%
\, \, \xrightarrow{n \rightarrow \infty}  H(X|Z) - H(X|VZ) }\nonumber \\
& = & I(X;Z|V), \nonumber
\end{IEEEeqnarray}
where $(a)$ holds because $\mathcal{H}^{(n)}_{X|VZ} \supseteq \mathcal{H}^{(n)}_{X|V}$, and the limit by \cite[Theorem~1]{arikan2010source}. Thus, the randomization sequence rate used by the polar coding scheme is $R_{R}$ in \eqref{eq:targetrate}.

\subsubsection{Privacy shared sequence rate}
The transmitter and the $k$-th legitimate receiver must privately share the keys $\kappa_{\Theta}^{(V)}$, $\kappa_{\Gamma}^{(V)}$ and $\kappa_{\Upsilon\Phi_{(k)}}^{(V)}$. Hence, the overall rate is
\begin{IEEEeqnarray}{rCl}
\IEEEeqnarraymulticol{3}{l}{
\frac{\big|\kappa_{\Theta}^{(V)}\big| + \big|\kappa_{\Gamma}^{(V)}\big| + \sum_{k=1}^2 \big| \kappa_{\Upsilon\Phi_{(k)}}^{(V)}\big|}{nL}
} \nonumber \\
& = & \frac{ \big|  \mathcal{C}^{(n)}_{1} \big|  + \big|\mathcal{C}^{(n)}_{1,2} \big|}{nL}  \nonumber \\
& & \! \! \! + \! \sum_{k=1}^2 \! \frac{L  \Big|  \big( \mathcal{H}_V^{(n)}  \big)^{\text{C}} \!  \cap \big( \mathcal{L}_{V|Y_{(k)}}^{(n)} \big)^{\text{C}} \Big| \! \!  + \! \Big| \mathcal{H}_V^{(n)} \cap \big( \mathcal{L}_{V|Y_{(k)}}^{(n)} \big)^{\text{C}} \Big|}{nL} \nonumber \\
& \stackrel{(a)}{=} & \frac{ \Big| \mathcal{H}_V^{(n)} \cap \big(  \mathcal{H}_{V|Z}^{(n)} \big)^{\text{C}} \cap \big( \mathcal{L}_{V|Y_{(1)}}^{(n)} \big)^{\text{C}} \Big|}{nL} \nonumber \\
& & \! \! \! + \! \sum_{k=1}^2 \! \frac{L  \Big|  \big( \mathcal{H}_V^{(n)}  \big)^{\text{C}} \!  \cap \big( \mathcal{L}_{V|Y_{(k)}}^{(n)} \big)^{\text{C}} \Big| \! \!  + \! \Big| \mathcal{H}_V^{(n)} \cap \big( \mathcal{L}_{V|Y_{(k)}}^{(n)} \big)^{\text{C}} \Big|}{nL} \nonumber \\ 
& \stackrel{(b)}{\leq} & \frac{ \Big|  \big( \mathcal{L}_{V|Y_{(1)}}^{(n)} \big)^{\text{C}} \Big|}{nL} \!  \nonumber \\
& &  + \sum_{k=1}^2  \frac{ L \Big| \big( \mathcal{H}_{V|Y_{(k)}}^{(n)} \big)^{\text{C}}  \setminus \mathcal{L}_{V|Y_{(k)}}^{(n)} \Big|   +  \Big| \big(  \mathcal{L}_{V|Y_{(k)}}^{(n)} \big)^{\text{C}} \Big|}{nL} \nonumber \\
 & \IEEEeqnarraymulticol{2}{l}{%
\, \, \xrightarrow{n \rightarrow \infty}  \frac{2 H(V|Y_{(1)}) + H(V|Y_{(2)}) }{L}
}\nonumber \\
& \IEEEeqnarraymulticol{2}{l}{%
\, \, \xrightarrow{L \rightarrow \infty} 0,
}\nonumber 
\end{IEEEeqnarray}
where $(a)$ follows from the definition of $\mathcal{C}^{(n)}_{1}$ and $\mathcal{C}^{(n)}_{1,2}$ in \eqref{eq:sC1} and \eqref{eq:sC12} respectively; for the second term in $(b)$ we have used $(\mathcal{H}_V^{(n)})^{\text{C}} \subseteq (\mathcal{H}_{V|Y_{(k)}})^{\text{C}}$; and the limit when $n$ goes to infinity holds by \cite[Theorem~1]{arikan2010source}.

\subsubsection{Rate of the additional randomness}
Besides the randomization sequences $R_{1:L}$, the encoder uses the random sequence $\Lambda_0^{(X)}$, with size $|\mathcal{H}_{X|V}^{(n)}|$, for the polar-based channel prefixing. Moreover, for $i \in [1,L]$, the encoder randomly draws those elements $\tilde{A}_{i}(j)$ such that $j \in (\mathcal{H}_V^{(n)})^{\text{C}}\setminus \mathcal{L}_V^{(n)}$, and those elements $\tilde{T}_{i}(j)$ such that $j \in (\mathcal{H}_{X|V}^{(n)})^{\text{C}}\setminus \mathcal{L}_{X|V}^{(n)}$. Nevertheless, we have  
\begin{IEEEeqnarray}{rCl}
\IEEEeqnarraymulticol{3}{l}{
\frac{\Big| \mathcal{H}_{X|V}^{(n)} \Big| + L \Big| \big(\mathcal{H}_{V}^{(n) }\big)^{\text{C}}\setminus \mathcal{L}_{V}^{(n)}\Big| + L \Big| \big(\mathcal{H}_{X|V}^{(n) }\big)^{\text{C}}\setminus \mathcal{L}_{X|V}^{(n)}\Big|}{nL}
} \nonumber \\
& \IEEEeqnarraymulticol{2}{l}{%
\, \, \xrightarrow{n \rightarrow \infty}  \frac{H(X|V)}{L}
}\nonumber \\
& \IEEEeqnarraymulticol{2}{l}{%
\, \, \xrightarrow{L \rightarrow \infty} 0,
}\nonumber 
\end{IEEEeqnarray}
where the limit when $n$ approaches to infinity follows from applying \cite[Theorem~1]{arikan2010source}.

\subsection{Distribution of the DMS after the polar encoding}\label{sec:performance_statistics}
For $i \in [1,L]$, let $\tilde{q}_{A_i^n T_i^n}$ denote the distribution of $(\tilde{A}_i^n,\tilde{T}_i^n)$ after the encoding. The following lemma proves 
that $\tilde{q}_{A_i^n T_i^n}$ and $p_{A^n T^n}$, the later being a marginal of the joint distribution of the original DMS in \eqref{eq:distDMS}, are nearly statistically indistinguishable for sufficiently large $n$ and, consequently, so are $\tilde{q}_{V_i^n X_i^n Y_{(1),i}^n Y_{(2),i}^n Z_i^n }$ and $p_{V^n X^nY_{(1)}^n Y_{(2)}^n Z^n }$. This result is crucial for the reliability and secrecy performance of the polar coding scheme.

\begin{lemma}
\label{lemma:distDMS}
For $i \in  [1,L]$, we have 
\begin{IEEEeqnarray}{rCl}
\mathbb{V}(\tilde{q}_{A_i^n T_i^n},p_{A^n T^n}) & \leq & \delta^{(*)}_n, \nonumber
\end{IEEEeqnarray}
which implies
\begin{IEEEeqnarray}{c}
\mathbb{V}(\tilde{q}_{V_i^n X_i^n Y_{(1),i}^n Y_{(2),i}^n Z_i^n },p_{V^n X^nY_{(1)}^n Y_{(2)}^n Z^n }) \leq \delta^{(*)}_n, \nonumber 
\end{IEEEeqnarray}
where
\begin{IEEEeqnarray}{rCl}
\delta^{(*)}_n & \triangleq & 2 n \sqrt{ 4 \sqrt{ \! n \delta_n \ln 2} \big( 2n \!  - \! \log \big(2 \sqrt{n \delta_n \ln 2} \big) \big) \! + \! \delta_n} \nonumber \\
& & + 2 \sqrt{n \delta_n \ln 2}. \nonumber
\end{IEEEeqnarray}
\end{lemma}

\begin{proof}
See Appendix~\ref{app:distributionDMS}
\end{proof}

\subsection{Reliability analysis}\label{sec:performance_reliability}
We prove that both legitimate receivers can reliably reconstruct the private and the confidential messages $(W_{1:L},S_{1:L})$ with arbitrary small error probability. 

For $i \in [1,L]$ and $k \in [1,2]$, let $\tilde{q}_{V_i^n Y_{(k),i}^n}$ and $p_{V^n Y_{(k)}^n}$ be marginals of $\tilde{q}_{V_i^n X_i^n Y_{(1),i}^n Y_{(2),i}^n Z_i^n }$ and $p_{V^n X^n Y_{(1)}^n Y_{(2)}^n Z^n }$ respectively, and define an optimal coupling \cite[Proposition~4.7]{levin2009markov} between $\tilde{q}_{V_i^n Y_{(k),i}^n}$ and $p_{V^n Y_{(k)}^n}$ such that 
\begin{IEEEeqnarray}{c}
\mathbb{P} [ \mathcal{E}_{V_i^n Y_{(k),i}^n} ] = \mathbb{V} ( \tilde{q}_{V_i^n Y_{(k),i}^n} , p_{V^n Y_{(k)}^n} ), \nonumber
\end{IEEEeqnarray}
where $\mathcal{E}_{V_i^n Y_{(k),i}^n} \triangleq \big\{ \big( \tilde{V}_i^n, \tilde{Y}_{(k),i}^n \big) \neq \big( V^n, Y_{(k)}^n \big) \big\}$. Additionally, define the error event
\begin{IEEEeqnarray}{c}
\mathcal{E}_{(k),i} \triangleq \Big\{ \tilde{A}_i \big[ \big( \mathcal{L}_{V|Y_{(k)}}^{(n)} \big)^{\text{C}} \big]  \neq \hat{A}_i \big[ \big( \mathcal{L}_{V|Y_{(k)}}^{(n)} \big)^{\text{C}} \big] \Big\}. \nonumber 
\end{IEEEeqnarray}
Recall that $(\Upsilon_{(k)}^{(V)},\Phi_{(k),1:L}^{(V)})$ is available to the $k$-th legitimate receiver. Thus, notice that we have
\begin{IEEEeqnarray}{c}
\mathbb{P} [ \mathcal{E}_{(1),1} ] = \mathbb{P} [ \mathcal{E}_{(2),L} ]  = 0 \label{eq:e1L}  
\end{IEEEeqnarray}
because given $\Upsilon_{(1)}^{(V)}$ and $\Phi_{(1),1}^{(V)}$ legitimate receiver 1 knows $\tilde{A}_{1} [( \mathcal{L}_{V|Y_{(1)}}^{(n)} )^{\text{C}}  ]$, and given $\Upsilon_{(2)}^{(V)}$ and $\Phi_{(2),L}^{(V)}$ legitimate receiver 2 knows $\tilde{A}_{L} [( \mathcal{L}_{V|Y_{(2)}}^{(n)} )^{\text{C}}  ]$. Moreover, due to the chaining structure, in Section~\ref{sec:decoding} we have seen that $\hat{A}_i [ \mathcal{H}^{(n)}_V \cap ( \mathcal{L}_{V|Y_{(1)}}^{(n)})^{\text{C}} ] \subset \hat{A}_{i-1}^n$ for  $i \in [2,L]$. Therefore, at legitimate receiver 1, for $i \in [2,L]$ we obtain
\begin{IEEEeqnarray}{c}
\mathbb{P} [ \mathcal{E}_{(1),i} ] \leq \mathbb{P} [ \tilde{A}_{i-1}^n \neq \hat{A}_{i-1}^n ].  \label{eq:e1i} 
\end{IEEEeqnarray}
Similarly, due to the chaining construction, we have seen that $\hat{A}_i [ \mathcal{H}^{(n)}_V \cap ( \mathcal{L}_{V|Y_{(2)}}^{(n)} )^{\text{C}} ] \subset \hat{A}_{i+1}^n$ for $i \in [1,L-1]$. Thus, at legitimate receiver 2, for $i \in [1,L-1]$ we obtain
\begin{IEEEeqnarray}{c}
\mathbb{P} [ \mathcal{E}_{(2),i} ] \leq \mathbb{P} [ \tilde{A}_{i+1}^n \neq \hat{A}_{i+1}^n ], \quad i \in [1,L-1].  \IEEEeqnarraynumspace \label{eq:e2i} 
\end{IEEEeqnarray}

Therefore, the probability of incorrectly decoding $(W_{i},S_{i})$ at the $k$-th receiver can be bounded as
\begin{IEEEeqnarray}{rCl}
\IEEEeqnarraymulticol{3}{l}{
\mathbb{P} \big[ (W_i, S_i) \neq (\hat{W}_i , \hat{S}_i ) \big]
} \nonumber \\
& \leq & \mathbb{P} \big[ \tilde{A}_i^n \neq \hat{A}_i^n \big] \nonumber \\
& = & \mathbb{P} \Big[ \tilde{A}_i^n \neq \hat{A}_i^n \Big| \mathcal{E}_{V_i^n Y_{(k),i}^n}^{\text{C}} \! \cap \mathcal{E}_{(k),i}^{\text{C}} \Big] \mathbb{P} \Big[ \mathcal{E}_{V_i^n Y_{(k),i}^n}^{\text{C}} \! \cap  \mathcal{E}_{(k),i}^{\text{C}}  \Big] \nonumber \\
&  &  \! + \mathbb{P} \Big[ \tilde{A}_i^n \neq \hat{A}_i^n \Big| \mathcal{E}_{V_i^n Y_{(k),i}^n} \!\! \cup  \! \mathcal{E}_{(k),i} \Big] \mathbb{P} \Big[ \mathcal{E}_{V_i^n Y_{(k),i}^n} \!  \cup  \mathcal{E}_{(k),i} \Big]  \nonumber \\
& \leq & \mathbb{P} \Big[ \tilde{A}_i^n \neq \hat{A}_i^n \Big| \mathcal{E}_{V_i^n Y_{(k),i}^n}^{\text{C}} \! \cap \mathcal{E}_{(k),i}^{\text{C}} \Big] \! \nonumber \\
& & \! + \mathbb{P} \Big[ \mathcal{E}_{V_i^n Y_{(k),i}^n} \! \cup  \mathcal{E}_{(k),i} \Big] \nonumber \\
& \stackrel{(a)}{\leq} & n \delta_n + \mathbb{P} \big[ \mathcal{E}_{V_i^n Y_{(k),i}^n} \! \cup  \mathcal{E}_{(k),i} \big]  \nonumber \\
& \stackrel{(b)}{\leq} & n \delta_n + \mathbb{P} \big[ \mathcal{E}_{V_i^n Y_{(k),i}^n} \big]  + \mathbb{P} \big[ \mathcal{E}_{(k),i} \big] \nonumber \\
& \stackrel{(c)}{\leq} & n \delta_n + 2\delta_{n}^{(*)}  + \mathbb{P} \big[ \mathcal{E}_{(k),i} \big],  \nonumber 
\end{IEEEeqnarray}
where $(a)$ holds by \cite[Theorem~2]{arikan2010source}; $(b)$ holds by the union bound; and $(c)$ follows from the optimal coupling and Lemma~\ref{lemma:distDMS}. Therefore, we have
\begin{IEEEeqnarray}{rCl}
\IEEEeqnarraymulticol{3}{l}{
\mathbb{P} \big[ (W_{1:L}, S_{1:L}) \neq (\hat{W}_{1:L} , \hat{S}_{1:L}) \big]
} \nonumber \\
& \stackrel{(a)}{\leq} & \sum_{i=1}^L \mathbb{P} \big[ \tilde{A}_i^n \neq \hat{A}_i^n \big] \nonumber \\
& \stackrel{(b)}{\leq} & \frac{L (L+1)}{2} \Big( n \delta_n + 2 \delta_{n}^{(*)} \Big), \nonumber 
\end{IEEEeqnarray}
where $(a)$ follows from applying the union bound, and $(b)$ holds by induction and \eqref{eq:e1L}--\eqref{eq:e2i}.

\subsection{Secrecy analysis}\label{sec:performance_secrecy}
Since encoding in Section~\ref{sec:PCS} takes place over $L$ blocks of size $n$, we need to prove that
\begin{align}
\lim_{n,L \rightarrow \infty} I(S_{1:L},\tilde{Z}_{1:L}^n) = 0. \nonumber
\end{align}
For clarity and with slight abuse of notation, for any block~$i \in [1,L]$ let 
\begin{IEEEeqnarray}{rCl}
\Xi_i^{(V)} & \triangleq & [\Pi_i^{(V)},\Lambda_i^{(V)},\Psi_i^{(V)},\Gamma_i^{(V)}],  
\end{IEEEeqnarray}
which denotes the entire random sequence depending on $\tilde{A}_i^n$ at block~$i$ that is repeated at block~$i+1$. Also, let
\begin{IEEEeqnarray}{c}
\bar{\Omega}_i^{(V)} \triangleq  [\bar{\Theta}^{(V)}_i ,\bar{\Gamma}^{(V)}_i],
\end{IEEEeqnarray}
which represents the sequence at block~$i$ that is repeated at block~$i-1$. Also, we define $\kappa_{\Omega}^{(V)} \triangleq [\kappa_{\Theta}^{(V)}, \kappa_{\Gamma}^{(V)}]$.

A Bayesian graph describing the dependencies between all the variables involved in the coding scheme of Section~\ref{sec:PCS} is given in Figure~\ref{fig:dependencies}. 
Despite $\Gamma_{i}^{(V)} \subseteq \Xi_i^{(V)}$ and $\bar{\Gamma}_{i}^{(V)} = \Gamma_{i}^{(V)} \oplus \kappa_{\Gamma}^{(V)} \subseteq \bar{\Omega}_i^{(V)}$,
we represent $\Xi_i^{(V)}$ and $\bar{\Omega}_i^{(V)}$ as two independent nodes in the Bayesian graph because, by \emph{crypto lemma}~\cite{crypto1241}, $\Gamma_{i}^{(V)}$ and $\bar{\Gamma}_{i}^{(V)}$ are statistically independent.
Furthermore, for convenience, we have considered that dependencies only take place forward (from block~$i$ to block~$i+1$), which is possible by reformulating the encoding as follows. According to Section~\ref{sec:encoding}, for any $i \in [1,L]$ we have $\tilde{A}_i [\mathcal{C}_i^{(n)}] = W_i$. Consequently, we can write $W_i \triangleq [W_{1,i}, W_{2,i}]$, where $W_{1,i} \triangleq \tilde{A}_i[\mathcal{C}_1^{(n)} \cup \mathcal{C}_{1,2}^{(n)}]$ and $W_{2,i} \triangleq \tilde{A}_i[\mathcal{C}_2^{(n)} \cup \mathcal{C}_{0}^{(n)}]$. Thus, we regard $\bar{\Omega}_{i}^{(V)}$ as an independent random sequence generated at block~$i-1$, which is stored properly into some part of $\tilde{A}_{i-1}[\mathcal{G}^{(n)}]$. Then, we consider that the encoder obtains $W_{1,i} \triangleq  \bar{\Omega}_{i}^{(V)} \oplus \kappa_{\Omega}^{(V)}$, which is stored into $\tilde{A}_i[\mathcal{C}_1^{(n)} \cup \mathcal{C}_{1,2}^{(n)}]$. On the other hand, the remaining part $W_{2,i}$ is independently generated at block~$i$.


\begin{figure}[h]
\vspace{0.15cm}
\centering
\hspace{0.6cm}
\begin{small}
\begin{overpic}[width=0.85\linewidth]{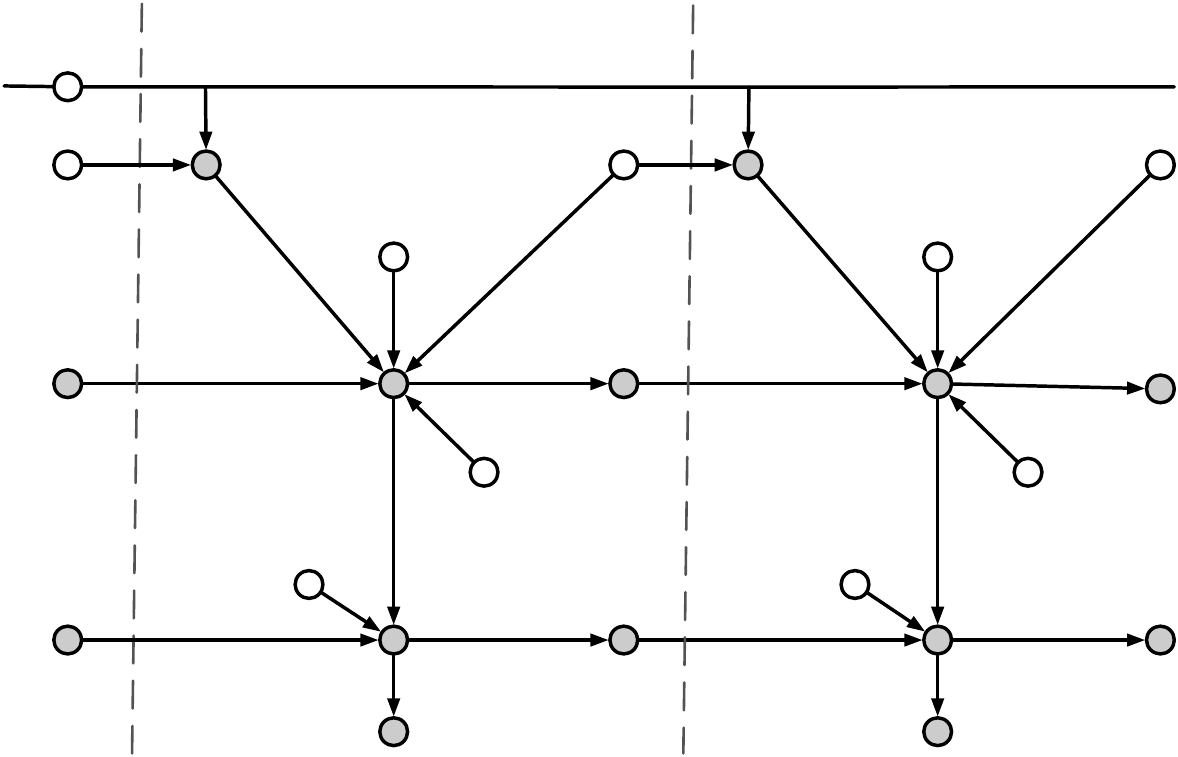}
\put (-15,65) {Block $i \! - \! 2$}
\put (22,65) {Block $i \! - \! 1$}
\put (70,65) {Block $i$}
\put (65,51) {$\Scale[0.95]{W_{1,i}}$}
\put (19,51) {$\Scale[0.95]{W_{1,i-1}}$}
\put (-1,59) {$\Scale[0.95]{\kappa_{\Omega}^{(V)}}$}
\put (46,51) {$\Scale[0.95]{\bar{\Omega}_{i}^{(V)}}$}
\put (-5,49) {$\Scale[0.95]{\bar{\Omega}_{i-1}^{(V)}}$}
\put (87.5,51) {$\Scale[0.95]{\bar{\Omega}_{i+1}^{(V)}}$}
\put (-3,34.6) {$\Scale[0.95]{\Xi_{i-2}^{(V)}}$}
\put (47.8,34.6) {$\Scale[0.95]{\Xi_{i-1}^{(V)}}$}
\put (96.5,34.3) {$\Scale[0.95]{\Xi_{i}^{(V)}}$}
\put (-3,12.8) {$\Scale[0.95]{\Lambda_{i-2}^{(X)}}$}
\put (46.6,12.8) {$\Scale[0.95]{\Lambda_{i-1}^{(X)}}$}
\put (95,12.8) {$\Scale[0.95]{\Lambda_{i}^{(X)}}$}
\put (79,44) {$\Scale[0.95]{W_{2,i}}$}
\put (68,17) {$\Scale[0.95]{R_{i}}$}
\put (32,44) {$\Scale[0.95]{W_{2,i-1}}$}
\put (20,17) {$\Scale[0.95]{R_{i-1}}$}
\put (24.1,26.6) {$\Scale[0.95]{\tilde{A}_{i \! - \! 1}^n}$}
\put (72.8,26.6) {$\Scale[0.95]{\tilde{A}_{i}^n}$}
\put (35,12) {$\Scale[0.95]{\tilde{T}_{i - 1}^n}$}
\put (81,12) {$\Scale[0.95]{\tilde{T}_{i}^n}$}
\put (35,3) {$\Scale[0.95]{\tilde{Z}_{i - 1}^n}$}
\put (81.5,3) {$\Scale[0.95]{\tilde{Z}_{i}^n}$}
\put (42.5,25) {$\Scale[0.95]{S_{i-1}}$}
\put (89,25) {$\Scale[0.95]{S_{i}}$}
\end{overpic}
\end{small}
\caption{Graphical representation of the dependencies between random variables involved in the polar coding scheme. Independent random variables are indicated by white nodes, whereas those that are dependent are indicated by gray nodes.
}\label{fig:dependencies} 
\end{figure}
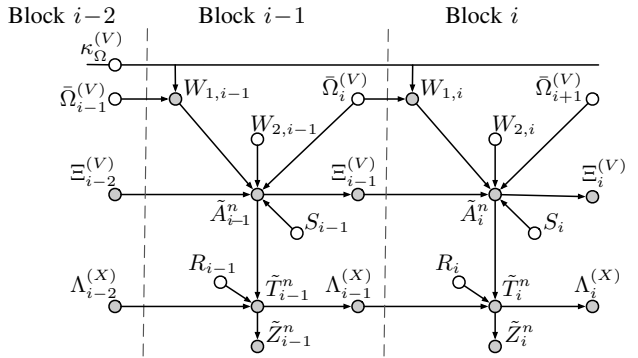

The following lemma shows that secrecy holds for any block $i \in [1,L]$ (block-wise strong secrecy).

\begin{lemma}
\label{lemma:secrecy1}
For $i \in [1,L]$ and sufficiently large $n$,
\begin{IEEEeqnarray}{rCl}
I \big( \tilde{A}_i [ \mathcal{H}_{V|Z}^{(n)} ] \tilde{T}_i [ \mathcal{H}_{V|Z}^{(n)} ]; \tilde{Z}_i^n \big) & \leq & \delta_n^{(\text{\emph{S}})} \nonumber 
\end{IEEEeqnarray}
where $\delta_n^{(\text{\emph{S}})} \triangleq 2 n \delta_n  + 2 \delta_n^{(*)} \big(4 n  -  \log \delta_n^{(*)} \big)$ and $\delta_n^{(*)}$ defined as in Lemma~\ref{lemma:distDMS}.
\end{lemma}
\begin{proof}
See Appendix~\ref{app:secrecy1}.
\end{proof}

Next, the following lemma shows that eavesdropper observations $\tilde{Z}_i^n$ are asymptotically statistically independent of observations $\tilde{Z}_{1:i-1}^n$ from previous blocks.
\begin{lemma}
\label{lemma:secrecy2}
For $i \in [1,L]$ and sufficiently large $n$,
\begin{IEEEeqnarray}{c}
I \big(S_{1:L} \tilde{Z}_{1:i-1}^n; \tilde{Z}_{i}^n \big) \leq \delta_n^{(\text{\emph{S}})},
\end{IEEEeqnarray}
where $\delta_n^{(\text{\emph{S}})}$ is defined as in Lemma~\ref{lemma:secrecy1}.
\end{lemma}
\begin{proof}
See Appendix~\ref{app:secrecy2}.
\end{proof}

Therefore, we obtain
\begin{IEEEeqnarray}{rCl}
\IEEEeqnarraymulticol{3}{l}{
I \big(S_{1:L}; \tilde{Z}_{1:L}^n \big)
} \nonumber \\
& \stackrel{(a)}{=} & I \big(S_{1:L}; \tilde{Z}_{1}^n \big) + \sum_{i=2}^L I \big(S_{1:L}; \tilde{Z}_{i}^n \big| \tilde{Z}_{1:i-1}^n \big) \nonumber \\
& \leq & I \big(S_{1:L}; \tilde{Z}_{1}^n \big) + \sum_{i=2}^L I \big(S_{1:L} \tilde{Z}_{1:i-1}^n; \tilde{Z}_{i}^n \big) \nonumber \\
& \stackrel{(b)}{\leq} & (L-1) \delta_n^{(\text{S})} +  I \big(S_{1:L}; \tilde{Z}_{1}^n \big) \nonumber \\
& = & (L-1) \delta_n^{(\text{S})} +  I \big(S_{1}; \tilde{Z}_{1}^n \big) + I \big(S_{2:L}; \tilde{Z}_{1}^n \big| S_1 \big)  \nonumber \\
& \stackrel{(c)}{=} & (L-1) \delta_n^{(\text{S})} + I \big(S_{1}; \tilde{Z}_{1}^n \big)   \nonumber \\
& \stackrel{(d)}{\leq} & L \delta_n^{(\text{S})}  \nonumber 
\end{IEEEeqnarray}
where $(a)$ follows from applying the chain rule; $(b)$ holds by Lemma~\ref{lemma:secrecy2}; $(c)$ holds by independence between $S_{2:L}$ and any random variable from block~$1$; and $(d)$ holds by Lemma~\ref{lemma:secrecy1} because $S_1 \subseteq \tilde{A}_1[\mathcal{H}^{(n)}_{V|Z}]$.


\begin{remark}
We conjecture that the use $\kappa_{\Omega}^{(V)}$ is not needed  for the polar coding scheme to satisfy the strong secrecy condition. However, the key is required in order to prove this condition by using a causal Bayesian graph similar to the one in Figure~\ref{fig:dependencies}. 
\end{remark}

\begin{remark}
Although non-causual (backward) dependencies between random variables of different blocks appear in \cite{8438560}, a secret seed as $\kappa_{\Omega}^{(V)}$ is not necessary for the polar coding scheme to provide strong secrecy. This is because random sequences that are repeated in adjacent blocks are stored only into those corresponding entries whose indices belong to the ``high entropy set given eavesdropper observations'', i.e., the equivalent sets of $\mathcal{H}_{V|Z}^{(n)}$ and $\mathcal{H}_{X|VZ}^{(n)}$ in our polar coding scheme. By contrast, notice that our polar coding scheme stores $[\Theta_i^{(V)},\Gamma_i^{(V)}]$ into some part of $\smash{\tilde{A}_i[(\mathcal{H}_{V|Z}^{(n)})^{\text{\emph{C}}}]}$.
\end{remark}

\begin{remark}
Another possibility for the polar coding scheme is to repeat at block~$i+1$ the modulo-2 addition between $[\Psi_{i}^{(V)},\Gamma_{i}^{(V)}]$ and a secret-key, instead of repeating an \emph{encrypted} version of $[\Theta_{i}^{(V)},\Gamma_{i}^{(V)}]$ at block $i-1$. Then, it is not difficult to prove that $I \big(S_{1:L} \tilde{Z}_{i+1:L}^n; \tilde{Z}_{i}^n \big) \leq \delta_n^{(\text{\emph{S}})}$ (similar to Lemma~\ref{lemma:secrecy2}). Thus, we can minimize the length of the secret seed depending on whether $| \mathcal{C}_1^{(n)} | < | \mathcal{C}_2^{(n)} | $ or vice versa.
\end{remark}


\section{Conclusion}\label{sec:conclusion}
A strongly secure polar coding scheme has been proposed for the WBC with two legitimate receivers and one eavesdropper. This polar code achieves the best known inner-bound on the achievable region of the CI-WBC model, where transmitter wants to send private and confidential messages to both receivers. Due to the non-degradedness assumption of the channel, the encoder builds a chaining construction that induces bidirectional dependencies between adjacent blocks. These dependencies need to be taken carefully into account in the secrecy analysis and make the use of secret seeds crucial. 

\begin{appendices}

\section{Proof of Lemma~\ref{lemma:distDMS}}\label{app:distributionDMS}
We generalize the results obtained in \cite{7447169} for any DMS $( \mathcal{V}_{1} \times \cdots \times \mathcal{V}_{M} \times \mathcal{O}_1 \times  \cdots \times \mathcal{O}_K,  p_{V_{1:M} O_{1:K}} )$ such that $\mathcal{V}_{\ell} \triangleq \{0,1\}$ for any $\ell \in [1,M]$, and $p_{V_{1:M} O_{1:K}}$ satisfies the Markov chain condition $V_1 - \cdots - V_M - O_{1:K}$. This DMS characterizes an encoding procedure for the broadcast channel with $K$ receivers (legitimate ones or eavesdroppers), where $O_k$ denotes the channel output~$k \in [1,K]$, and $V_{\ell}$ denotes the binary encoding input random variable~$\ell \in [1,M]$. Consider an i.i.d. $n$-sequence $( V_{1:M}^n, O_{1:K}^n )$ of this DMS, $n$ being any power of two. We define the polar transforms $U_{1:M}^n \sim p_{U_{1:M}^n}$, where $U_{\ell}^n \triangleq V_{\ell}^n G_n$ for any $\ell \in [1,M]$, and the sets
\begin{IEEEeqnarray}{rCl}
\mathcal{H}_{V_{\ell}|V_{\ell-1}}^{(n)} \! & \! \triangleq & \big\{ j \in [1,n]  \! : H \big( U_{\ell}(j) \big| U_{\ell}^{1:j-1}, V_{\ell-1}^n  \big) \nonumber\\*
\IEEEeqnarraymulticol{3}{r}{
 \geq 1 - \delta_n \big\},
}  \nonumber \\
\mathcal{L}_{V_{\ell}|V_{\ell-1}}^{(n)} \! & \! \triangleq & \big\{ j \in [1,n]  \! : H \big( U_{\ell}(j) \big| U_{\ell}^{1:j-1}, V_{\ell-1}^n  \big)  \leq \delta_n \big\}. \nonumber
\end{IEEEeqnarray}
Note that the model considered in this paper can be represented by the previous DMS if we take $K \triangleq 3$ (two legitimate receivers and one eavesdropper), $M \triangleq 2$ defining $V \triangleq V_1$, $X \triangleq V_2$ and the associated polar transforms $A^n \triangleq U_1^n$ and $T^n \triangleq U_2^n$.

Consider an encoder that forms $\tilde{U}_{1:M}^n$ with joint distribution $\tilde{q}_{U_{1:M}^n} \triangleq \prod_{\ell=1}^M \prod_{j=1}^n \tilde{q}_{{U}_{\ell}(j)|U_{\ell}^{1:j-1}V_{\ell-1}^n}$, where 
\begin{IEEEeqnarray}{rCl}
\IEEEeqnarraymulticol{3}{l}{
\tilde{q}_{U_{\ell}(j) | U_{\ell}^{1:j-1} V_{{\ell}-1}^n} \! \big( \tilde{u}_{\ell}(j) \big| \tilde{u}_{\ell}^{1:j-1}, \tilde{v}_{{\ell}-1}^n \big) 
}  \nonumber \\
& = & \left\{
\begin{array}{l}
\! \! \! \frac{1}{2} \qquad \qquad \qquad \qquad \qquad \qquad \quad \! \! \text{if } j \in \mathcal{H}_{V_{\ell}|V_{{\ell}-1}}^{(n)}, \\
\! \! \! p_{U_{\ell}(j) | U_{\ell}^{1:j-1} V_{{\ell}-1}^n} \big( \tilde{u}_{\ell}(j) \big| \tilde{u}_{\ell}^{1:j-1}, \tilde{v}_{{\ell}-1}^n \big) \\
\qquad  \qquad \qquad  \quad    \text{if } j \in \big( \mathcal{H}_{V_{\ell}|V_{\ell-1}}^{(n)} \big)^{\text{C}} \setminus \mathcal{L}_{V_{\ell}|V_{\ell-1}}^{(n)},  \\
\! \! \! \mathds{1} \big\{ \tilde{u}_{\ell}(j) = \xi^{(U_{\ell})}_{(j)} \big( \tilde{u}_{\ell}^{1:j-1},\tilde{v}_{\ell-1}^{n} \big) \big\} \, \, \text{if } j \in \mathcal{L}_{V_{\ell}|V_{{\ell}-1}}^{(n)} , 
\end{array} 
\right. \nonumber \\ \label{eq:dist_tilde}
\end{IEEEeqnarray}
for any $\ell \in [1,M]$, $V_{0}^n \triangleq \varnothing$ and $\xi^{(j)}$ being the deterministic $\argmax$ function defined as
\begin{IEEEeqnarray}{rCl}
\IEEEeqnarraymulticol{3}{l}{%
\xi_{(j)}^{(U_{\ell})} \big( {u}_{\ell}^{1:j-1} , {v}_{\ell-1}^n \big)
} \nonumber \\* 
& \triangleq & \argmax_{u \in \mathcal{X}} p_{U_{\ell}(j)|U_{\ell}^{1:j-1}V_{\ell-1}^n} \big( {u}_{\ell} \left| {u}_{\ell}^{1:j-1}, v_{\ell-1}^n \right. \big), \nonumber
\end{IEEEeqnarray}
Notice that the encoder in Section~\ref{sec:encoding} constructs, for any block~$i \in [1,L]$, the sequences $(\tilde{A}_i^n, \tilde{X}_i^n)$ with joint distribution $\tilde{q}_{A_i^n T_i^n}$ defined as in \eqref{eq:dist_tilde}. 

Consider another encoder that omits the use of the $\argmax$ function, but draws the corresponding elements randomly. Let $\check{U}_{1:M}^n$ denote the sequences constructed by this encoder. Then, their joint distribution is given by $\check{q}_{U_{1:M}^n} \triangleq \prod_{\ell=1}^M \prod_{j=1}^n \check{q}_{{U}_{\ell}(j)|U_{\ell}^{1:j-1}V_{\ell-1}^n}$, where
\begin{IEEEeqnarray}{rCl}
\IEEEeqnarraymulticol{3}{l}{
\check{q}_{U_{\ell}(j) | U_{\ell}^{1:j-1} V_{{\ell}-1}^n} \! \big( \check{u}_{\ell}(j) \big| \check{u}_{\ell}^{1:j-1}, \check{v}_{{\ell}-1}^n \big)
}  \nonumber \\
& = & \left\{
\begin{array}{l}
\frac{1}{2} \qquad \qquad \qquad \qquad \quad \qquad \, \, \, \, \,  \text{if } j \in \mathcal{H}_{V_{\ell}|V_{{\ell}-1}}^{(n)}, \\
p_{U_{\ell}(j) | U_{\ell}^{1:j-1} V_{{\ell}-1}^n} \big( \check{u}_{\ell}(j) \big| \check{u}_{\ell}^{1:j-1}, \check{v}_{{\ell}-1}^n \big) \\
\qquad  \qquad \qquad \qquad \qquad \quad   \text{if } j \in \big( \mathcal{H}_{V_{\ell}|V_{\ell-1}}^{(n)} \big)^{\text{C}}, 
\end{array} 
\right. \nonumber \\ \label{eq:dist_check}
\end{IEEEeqnarray}
for any $\ell \in [1,M]$, and $V_{0}^n \triangleq \varnothing$.

The following lemma shows that the joint distributions $p_{U_{1:M}^n}$ and $\check{q}_{U_{1:M}^n}$ are nearly statistically indistinguishable for sufficiently large $n$.

\begin{lemma}\label{lemma:distUc1Uc2}
Let $\delta_n = 2^{-n^{\beta}}$ for some $\beta \in (0, \frac{1}{2})$. Then
\begin{IEEEeqnarray}{c}
\mathbb{V} (\check{q}_{U_{1:M}^n}, p_{U_{1:M}^n}) \leq \sqrt{M}{\delta}^{(1)}_n, \nonumber
\end{IEEEeqnarray}
where ${\delta}^{(1)}_n \triangleq \sqrt{2 n \delta_n \ln 2}$.
\end{lemma}
\begin{proof}
The Kullback-Leibler distance between the distributions $p_{U_{1:M}^n}$ and $\check{q}_{U_{1:M}^n}$ is
\begin{IEEEeqnarray}{rCl}
\IEEEeqnarraymulticol{3}{l}{
\mathbb{D} \big( p_{U_{1:M}^n} \big\| \check{q}_{U_{1:M}^n} \big)
}  \nonumber \\
& \stackrel{(a)}{=} & \sum_{\ell = 1}^{M} \mathbb{E}_{{V_{\ell-1}^n }} \mathbb{D} \big( p_{U_{\ell}^n | V_{\ell-1}^n} \big\| \check{q}_{U_{\ell}^n | V_{\ell-1}^n} \big) \nonumber \\
& \stackrel{(b)}{=} & \sum_{\ell = 1}^M \sum_{j=1}^n \mathbb{E} \Big[ \mathbb{D} \Big( p_{U_{\ell}(j)|U_{\ell}^{1:j-1} V_{\ell-1}^n } \Big\| \check{q}_{U_{\ell}(j)|U_{\ell}^{1:j-1} V_{\ell-1}^n } \Big) \Big], \nonumber 
\end{IEEEeqnarray}
where $(a)$ holds by the chain rule, the invertibility of $G_n$ and the Markov chain condition satisfied by $U_{1:M}^n$ and $\check{U}_{1:M}^n$; and $(b)$ holds by the chain rule and by taking the expectation with respect to $(U_{\ell}^{1:j-1}, V_{\ell-1}^n)$. Thus, we obtain
\begin{IEEEeqnarray}{rCl}
\IEEEeqnarraymulticol{3}{l}{
\mathbb{D} \big( p_{U_{1:M}^n} \big\| \check{q}_{U_{1:M}^n} \big)
}  \nonumber \\
& \stackrel{(a)}{=} & \sum_{\ell = 1}^M \sum_{j \in \mathcal{H}_{V_{\ell}|V_{\ell-1}}^{(n)}} \Big( 1 - H \Big( U_{\ell}(j) \Big| U_{\ell}^{1:j-1}, V_{\ell-1}^{n} \Big) \Big) \nonumber \\
& \stackrel{(b)}{\leq} & M \delta_n \big| \mathcal{H}_{V_{\ell}|V_{\ell-1}}^{(n)} \big|, \nonumber
\end{IEEEeqnarray}
%
where $(a)$ holds by \eqref{eq:dist_check} and \cite[Lemma~10]{6975233}, i.e., $\mathbb{D}(p \| \check{q}) = 1 - H(p)$ if $\check{q}$ denotes the uniform distribution; and $(b)$ holds by the definition of $\mathcal{H}_{V_{\ell}|V_{\ell-1}}^{(n)}$. 

Finally, we get $\mathbb{V} (\check{q}_{U_{1:M}^n}, p_{U_{1:M}^n}) \leq \sqrt{2 M n \delta_n \ln 2}$ by Pinsker's inequality and because $|\mathcal{H}_{V_{\ell}|V_{\ell-1}}^{(n)}| \leq n$.
\end{proof}

Now, the following lemma proves that the joint distributions $\check{q}_{U_{1:M}^n}$ and $\tilde{q}_{U_{1:M}^n}$ are nearly statistically indistinguishable for $n$ large enough.

\begin{lemma}\label{lemma:distU1cU2}
Let $\delta_n = 2^{-n^{\beta}}$ for some $\beta \in (0, \frac{1}{2})$. Then
\begin{IEEEeqnarray}{c}
\mathbb{V} (\tilde{q}_{U_{1:M}^n}, \check{q}_{U_{1:M}^n}) \leq {\delta}^{(2)}_n, \nonumber
\end{IEEEeqnarray}
where ${\delta}^{(2)}_n \triangleq M n \sqrt{ 2 \sqrt{2} \delta_n^{(1)} \big( 2n - \log \sqrt{2} \delta_n^{(1)} \big) + \delta_n}$, and ${\delta}^{(1)}_n$ defined as in Lemma~\ref{lemma:distUc1Uc2}. 
\end{lemma}

\begin{proof}
The proof follows similar reasoning as the one for \cite[Lemma~2]{7447169}. Define a coupling \cite{levin2009markov} for $\check{U}_{1:M}^n$ and $\tilde{U}_{1:M}^n$ such that $\check{U}_{\ell} [ ( \mathcal{L}_{V_{\ell}|V_{\ell-1}}^{(n)} )^{\text{C}} ] = \tilde{U}_{\ell}[ ( \mathcal{L}_{V_{\ell}|V_{\ell-1}}^{(n)} )^{\text{C}} ]$ for any $\ell \in [1,M]$. Thus, we have
\begin{IEEEeqnarray}{rCl}
\IEEEeqnarraymulticol{3}{l}{
\mathbb{V} ( \tilde{q}_{U_{1:M}^n} , \check{q}_{U_{1:M}^n} )
}  \nonumber \\
& \stackrel{(a)}{\leq} & \mathbb{P} \big[  \tilde{U}_{1:M}^n \neq  \check{U}_{1:M}^n \big] \nonumber \\
& \stackrel{(b)}{\leq} & \sum_{\ell=1}^M \mathbb{P} \big[ \tilde{U}_{\ell}^n \neq \check{U}_{\ell}^n \big| \tilde{V}_{\ell-1}^n = \check{V}_{\ell-1}^n \big] \nonumber \\
& \stackrel{(c)}{\leq} & \sum_{\ell=1}^M \sum_{j=1}^n \mathbb{P} \Big[ \tilde{U}_{\ell}(j) \neq \check{U}_{\ell}(j) \Big| \mathcal{E}_{U^n_{\ell}V^n_{\ell-1}} \Big] \nonumber \\
 & \stackrel{(d)}{=} & \sum_{\ell=1}^M \sum_{j \in \mathcal{L}_{V_{\ell}|V_{\ell-1}}^{(n)}} \mathbb{E}_{(\check{U}_{\ell}^{1:j-1}, \check{V}_{\ell-1}^n)} \Big[ \nonumber \\
& & \quad  1 - p_{U_{\ell}(j) | U_{\ell}^{1:j-1}V_{\ell-1}^n} \big( u^{\ast}_{\ell}(j) |  \check{U}_{\ell}^{1:j-1}, \check{V}_{\ell-1}^n \big) \Big], \nonumber \\ \label{eq:comb1}
\end{IEEEeqnarray}
where $(a)$ holds by the coupling lemma \cite[Proposition~4.7]{levin2009markov}; $(b)$ by the union bound, the invertibility of $G_n$ and the Markov chain condition satisfied by $\tilde{U}_{1:M}^n$ and $\check{U}_{1:M}^n$; $(c)$ also holds by the union bound and defining $\mathcal{E}_{U^n_{\ell}V^n_{\ell-1}} \triangleq \{(\check{U}_{\ell}^{1:j-1}, \check{V}_{\ell-1}^n)=(\tilde{U}_{\ell}^{1:j-1}, \tilde{V}_{\ell-1}^n) \}$; and $(d)$ follows from \eqref{eq:dist_tilde} and \eqref{eq:dist_check} given that $\check{U}_{\ell} [ ( \mathcal{L}_{V_{\ell}|V_{\ell-1}}^{(n)} )^{\text{C}} ] = \tilde{U}_{\ell} [ ( \mathcal{L}_{V_{\ell}|V_{\ell-1}}^{(n)} )^{\text{C}} ]$ and from defining 
\begin{IEEEeqnarray}{c}
u_{\ell}^{\ast}(j) \triangleq \argmax_{u \in \{0,1\} } p_{U_{\ell}(j) | U_{\ell}^{1:j-1}V_{\ell-1}^n} ( u \big| \check{U}_{\ell}^{1:j-1}, \check{V}_{\ell-1}^n ). \nonumber
\end{IEEEeqnarray}

Next, for any $j \in [n]$ and sufficiently large $n$, we have
\begin{IEEEeqnarray}{rCl}
\IEEEeqnarraymulticol{3}{l}{
\big| H\big( U_{\ell} (j) \big| U_{\ell}^{1:j-1}, V_{\ell-1}^n \big) - H\big( U_{\ell} (j) \big| \check{U}_{\ell}^{1:j-1}, \check{V}_{\ell-1}^n \big) \big|
}  \nonumber \\
& \stackrel{(a)}{\leq} & \Big| H\big( U_{\ell}^{1:j-1}, V_{\ell-1}^n \big) - H\big( \check{U}_{\ell}^{1:j-1}, \check{V}_{\ell-1}^n \big) \Big| \nonumber \\
& & + \Big| H\big( U_{\ell}^{1:j}, V_{\ell-1}^n \big) - H\big( U_{\ell}(j), \check{U}_{\ell}^{1:j-1}, \check{V}_{\ell-1}^n \big) \Big| \nonumber \\
& \stackrel{(b)}{\leq} & 2 \mathbb{V} \big( \check{q}_{U_{\ell}^{1:j-1} U_{\ell-1}^n}, p_{U_{\ell}^{1:j-1} U_{\ell-1}^n} \big) \nonumber \\
& & \times \log \frac{2^{n+j-1}}{\mathbb{V} \big( \check{q}_{U_{\ell}^{1:j-1} U_{\ell-1}^n}, p_{U_{\ell}^{1:j-1} U_{\ell-1}^n} \big)} \nonumber \\
& \stackrel{(c)}{\leq} & 2 \sqrt{2} \delta_n^{(1)} \big( 2 n - \log \sqrt{2} \delta_n^{(1)} \big) \label{eq:absentr}
\end{IEEEeqnarray}
where $(a)$ holds by the chain rule of entropy and the triangle inequality; $(b)$ holds by \cite[Lemma~2.7]{csiszar2011information}, the invertibility of $G_n$, and because 
\begin{IEEEeqnarray}{rCl}
\IEEEeqnarraymulticol{3}{l}{
\mathbb{V} ( p_{U_{\ell}(j) | U_{\ell}^{1:j-1} U_{\ell-1}^n} \check{q}_{U_{\ell}^{1:j-1} U_{\ell-1}^n}, p_{U_{\ell}^{1:j} U_{\ell-1}^n} ) 
} \nonumber \\
& = & \mathbb{V} ( \check{q}_{U_{\ell}^{1:j-1} U_{\ell-1}^n}, p_{U_{\ell}^{1:j-1} U_{\ell-1}^n} ); \nonumber
\end{IEEEeqnarray}
and $(c)$ holds by Lemma~\ref{lemma:distUc1Uc2} (taking $M=2$) because 
\begin{IEEEeqnarray}{c}
\mathbb{V} ( \check{q}_{U_{\ell}^{1:j-1} U_{\ell-1}^n}, p_{U_{\ell}^{1:j-1} U_{\ell-1}^n} ) \leq \mathbb{V} ( \check{q}_{U_{\ell-1:\ell}^n}, p_{U_{\ell-1:\ell}} ), \nonumber
\end{IEEEeqnarray}
$x \mapsto x \log x$ is decreasing for $x > 0$ small enough, and $j \leq n$. Thus, for $\ell \in [1,M]$ and $j \in \mathcal{L}_{V_{\ell}|V_{\ell-1}}^{(n)}$, we have
\begin{IEEEeqnarray}{rCl}
\IEEEeqnarraymulticol{3}{l}{
2 \sqrt{2} \delta_n^{(1)} \big( 2n - \log \sqrt{2} \delta_n^{(1)} \big) + \delta_n
}  \nonumber \\
& \stackrel{(a)}{\geq} & 2 \sqrt{2} \delta_n^{(1)} \big( 2 n - \log \sqrt{2} \delta_n^{(1)} \big) + \nonumber \\
& & + H \big( U_{\ell}(j)| U_{\ell}^{1:j-1} , V_{\ell-1}^n \big) \nonumber \\
& \stackrel{(b)}{\geq} & H \big( U_{\ell}(j)| U_{\ell}^{1:j-1} , V_{\ell-1}^n \big) \nonumber \\
& \stackrel{(c)}{=} & \mathbb{E} \left[ h_2 \left( p_{U_{\ell}(j) | U_{\ell}^{1:j-1}V_{\ell-1}^n} \big( u_{\ell}^{\star}(j) \big| \check{U}_{\ell}^{1:j-1}, \check{V}_{\ell-1}^n \big) \right) 
 \right] \nonumber \\
& \geq & \mathbb{E} \left[ - \left(1 - p_{U_{\ell}(j) | U_{\ell}^{1:j-1}V_{\ell-1}^n} \big( u_{\ell}^{\star}(j) \big| \check{U}_{\ell}^{1:j-1}, \check{V}_{\ell-1}^n \big) \right) \right. \nonumber \\ 
& &  \! \times \left. \log \! \left( 1 - p_{U_{\ell}(j) | U_{\ell}^{1:j-1}V_{\ell-1}^n} \big( u_{\ell}^{\star}(j) \big| \check{U}_{\ell}^{1:j-1}, \check{V}_{\ell-1}^n \big) \right) \right] \nonumber \\
& \stackrel{(d)}{\geq} & \mathbb{E} \left[ \left(1 - p_{U_{\ell}(j) | U_{\ell}^{1:j-1}V_{\ell-1}^n} \big( u_{\ell}^{\star}(j) \big| \check{U}_{\ell}^{1:j-1}, \check{V}_{\ell-1}^n \big) \right)^2 \right] \nonumber \\
& \stackrel{(e)}{\geq} &  \mathbb{E}^2 \left[ \left(1 - p_{U_{\ell}(j) | U_{\ell}^{1:j-1}V_{\ell-1}^n} \big( u_{\ell}^{\star}(j) \big| \check{U}_{\ell}^{1:j-1}, \check{V}_{\ell-1}^n \big) \right) \right] \nonumber \\ \label{eq:comb2}
\end{IEEEeqnarray}
where $(a)$ holds by the definition of $\mathcal{L}_{V_{\ell}|V_{\ell-1}}^{(n)}$; $(b)$ holds by \eqref{eq:absentr}; in $(c)$ the expectation is taken with respect to $(\check{U}_{\ell}^{1:j-1}, \check{U}_{\ell-1}^n)$ and $h_2(p)$ denotes the binary entropy function, i.e., $h_2(p) = -p \log p - (1-p) \log (1-p)$; $(d)$ holds because $p_{U_{\ell}(j) | U_{\ell}^{1:j-1}V_{\ell-1}^n}( u_{\ell}^{\star}(j) | \check{U}_{\ell}^{1:j-1}, \check{V}_{\ell-1}^n ) \geq 1/2$ and $\log (x) < -x$ if $x \in [0, 1/2 )$; and $(e)$ follows from applying Jensen's inequality. 

Finally, by combining Equations~\eqref{eq:comb1} and~\eqref{eq:comb2}, and because $| \mathcal{L}_{V_{\ell}|V_{\ell-1}}^{(n)} | \leq n$, we have 
\begin{IEEEeqnarray}{rCl}
\IEEEeqnarraymulticol{3}{l}{
\mathbb{V} (\tilde{q}_{U_{1:M}^n}, \check{q}_{U_{1:M}^n})
}  \nonumber \\
& \leq & M n \sqrt{ 2 \sqrt{2} \delta_n^{(1)} ( 2n - \log \sqrt{2} \delta_n^{(1)} ) + \delta_n}, \nonumber
\end{IEEEeqnarray}
and the proof is complete.
\end{proof}

Hence, by Lemma~\ref{lemma:distUc1Uc2}, Lemma~\ref{lemma:distU1cU2} and by applying the triangle inequality, we obtain
\begin{IEEEeqnarray}{rCl}
\IEEEeqnarraymulticol{3}{l}{
\mathbb{V} (\tilde{q}_{U_{1:m}^n}, p_{U_{1:m}^n})
}  \nonumber \\
& \leq & \mathbb{V} (\tilde{q}_{U_{1:M}^n}, \check{q}_{U_{1:M}^n}) + \mathbb{V} (\check{q}_{U_{1:M}^n}, p_{U_{1:M}^n}) \nonumber \\
& \leq & M n \sqrt{ 2 \sqrt{2} \delta_n^{(1)} \big( 2n - \log \sqrt{2} \delta_n^{(1)} \big) + \delta_n} + \sqrt{M} {\delta}^{(1)}_n. \nonumber 
\end{IEEEeqnarray}
Moreover, since
\begin{IEEEeqnarray}{rCl}
\IEEEeqnarraymulticol{3}{l}{
\mathbb{V} (\tilde{q}_{U_{1:M}^n O_{1:K}^n}, p_{U_{1:M}^n O_{1:K}^n})
}  \nonumber \\
& = & \mathbb{V} (\tilde{q}_{U_{1:M}^n} p_{O_{1:K}^n | U_{1:M}^n}, p_{U_{1:M}^n} p_{O_{1:K}^n | U_{1:M}^n})  \nonumber \\
& = & \mathbb{V} (\tilde{q}_{U_{1:M}^n} , p_{U_{1:M}^n}), \nonumber 
\end{IEEEeqnarray}
the joint distributions $p_{U_{1:M}^n O_{1:K}^n}$ and $\check{q}_{U_{1:M}^n O_{1:K}^n}$ are also nearly statistically indistinguishable for $n$ large enough.

Additionally, we provide the following lemma, which relates the total variation distance between $\tilde{q}_{U_{1:M}^n O_{1:K}^n}$ and $p_{U_{1:M}^n O_{1:K}^n}$ with the corresponding entropies.

\begin{lemma}\label{lemma:appsec1_1}
Define $\mathcal{J}_{\ell}$ as any subset of $[1,n]$, where $\ell \in [1,M]$. Let $(U_{1:M}^n,O_{1:K}^n) \sim p_{U_{1:M}^n O_{1:K}^n}$ and $(\tilde{U}_{1:M}^n, \tilde{O}_{1:K}^n) \sim \tilde{q}_{U_{1:M}^n O_{1:K}^n}$ such that the total variation distance $\mathbb{V}\big( p_{U_{1:M}^n O_{1:K}^n}, \tilde{q}_{U_{1:M} O_{1:K}^n} \big) \leq \delta$, where $\delta \rightarrow 0$. Then, for sufficiently large $n$, we have
\begin{IEEEeqnarray}{rCl}
\IEEEeqnarraymulticol{3}{l}{
\big| H \big(\tilde{U}_1 [ \mathcal{J}_1 ] \dots \tilde{U}_{M} [ \mathcal{J}_M ]  \tilde{O}_{1:K}^n \big) 
} \nonumber \\
& & - H \big(U_1[ \mathcal{J}_1 ] \dots U_M[ \mathcal{J}_M ] O_{1:K}^n \big)  \big|  \nonumber \\
& \leq & (M+K+1) n \delta  -  \delta \log \delta.   \nonumber 
\end{IEEEeqnarray}
\end{lemma} 
\begin{proof}
From applying \cite[Lemma~2.7]{csiszar2011information}, we obtain
\begin{IEEEeqnarray}{rCl}
\IEEEeqnarraymulticol{3}{l}{
\big| H \big(\tilde{U}_1 [ \mathcal{J}_1 ] \dots \tilde{U}_{M} [ \mathcal{J}_M ]  \tilde{O}_{1:K}^n \big) 
} \nonumber \\
& & - H \big(U_1[ \mathcal{J}_1 ] \dots U_M[ \mathcal{J}_M ] O_{1:K}^n \big)  \big|  \nonumber \\
& \leq & \mathbb{V}\big( p_{U_1[\mathcal{J}_1 ]\dots T[\mathcal{J}_M ] O_{1:K}^n} ,\tilde{q}_{U_1[ \mathcal{J}_1 ] \dots U_M[ \mathcal{J}_M ] O_{1:K}^n}\big) \nonumber \\
& & \times \log \frac{2^{n + \sum_{\ell =1}^M | \mathcal{J}_{\ell} | + Kn}}{\mathbb{V}\big( p_{U_1[\mathcal{J}_1 ]\dots T[\mathcal{J}_M ] O_{1:K}^n} ,\tilde{q}_{U_1[ \mathcal{J}_1 ] \dots U_M[ \mathcal{J}_M ] O_{1:K}^n}\big)} \nonumber \\
& \leq & (M+ K + 1) n \delta - \delta \log \delta,   \nonumber 
\end{IEEEeqnarray}
where the last inequality holds by assumption, because the function $x \mapsto x \log x$ is decreasing for $x > 0$ small enough, and because $|\mathcal{J}_{\ell}| \leq n$ for any $\ell \in [1,M]$.
\end{proof}

\section{Proof of Lemma~\ref{lemma:secrecy1}}\label{app:secrecy1}
We prove that $\smash{\tilde{A}_i [ \mathcal{H}_{V|Z}^{(n)} ]}$ and $\smash{\tilde{A}_i [ \mathcal{H}_{X|VZ}^{(n)} ]}$ are asymptotically jointly independent of eavesdropper channel observations $\tilde{Z}_i^n$. To do so, we use the following lemma.

First, for any $i \in [1,L]$ and $n$ large enough, 
\begin{IEEEeqnarray}{rCl}
\IEEEeqnarraymulticol{3}{l}{
\Big| H \big(\tilde{A}_i [ \mathcal{H}_{V|Z}^{(n)} ]  \tilde{T}_i [ \mathcal{H}_{X|VZ}^{(n)} ] \big| \tilde{Z}_i^n \big) 
} \nonumber \\
& & - H \big(A [ \mathcal{H}_{V|Z}^{(n)} ]  T [ \mathcal{H}_{X|VZ}^{(n)} ] \big| Z^n \big) \Big| \nonumber \\
& \stackrel{(a)}{\leq} & \! \Big| H \big(\tilde{A}_i [ \mathcal{H}_{V|Z}^{(n)} ]  \tilde{T}_i [ \mathcal{H}_{X|VZ}^{(n)} ] \tilde{Z}_i^n \big) \nonumber \\
& & \! \! - H \big(A [ \mathcal{H}_{V|Z}^{(n)} ]  T [ \mathcal{H}_{X|VZ}^{(n)} ] Z^n \big) \Big| + \Big| H \big(\tilde{Z}_i^n \big) - H \big(Z^n \big) \Big| \nonumber \\
& \stackrel{(b)}{\leq} & 8 n \delta_n^{(*)} - 2 \delta_n^{(*)} \log \delta_n^{(*)},  \label{eq:aux_s} 
\end{IEEEeqnarray}
where $(a)$ follows from applying the chain rule of entropy and the triangle inequality; and $(b)$ holds by Lemma~\ref{lemma:appsec1_1} and Lemma~\ref{lemma:distDMS} because
\begin{IEEEeqnarray}{rCl}
\IEEEeqnarraymulticol{3}{l}{
\mathbb{V}\big( p_{Z^n} ,\tilde{q}_{Z_i^n}\big)
} \nonumber \\
& \leq & \mathbb{V}\big( p_{A [ \mathcal{H}_{V|Z}^{(n)} ]  T [ \mathcal{H}_{X|VZ}^{(n)} ] Z^n} ,\tilde{q}_{A_i [ \mathcal{H}_{V|Z}^{(n)} ]  T_i [ \mathcal{H}_{X|VZ}^{(n)} ] Z_i^n}\big) \nonumber \\
& \leq & \mathbb{V}(\tilde{q}_{V_i^n X_i^n Y_{(1),i}^n Y_{(2),i}^n Z_i^n },p_{V^n X^nY_{(1)}^n Y_{(2)}^n Z^n }) \nonumber \\
& \leq & \delta^{(*)}_n. \nonumber 
\end{IEEEeqnarray}
Therefore, we have
\begin{IEEEeqnarray}{rCl}
\IEEEeqnarraymulticol{3}{l}{
I \big(\tilde{A}_i [ \mathcal{H}_{V|Z}^{(n)} ] \tilde{T}_i [ \mathcal{H}_{V|Z}^{(n)} ]; \tilde{Z}_i^n \big)  \nonumber
} \nonumber \\
& \stackrel{(a)}{=} & \big| \mathcal{H}_{V|Z}^{(n)} \big|  + \big| \mathcal{H}_{X|VZ}^{(n)} \big|  - H \big( \tilde{A}_i [ \mathcal{H}_{V|Z}^{(n)} ] \tilde{T}_i [ \mathcal{H}_{V|Z}^{(n)} ] \big| \tilde{Z}_i^n \big)  \nonumber \\
& \stackrel{(b)}{\leq} & \big| \mathcal{H}_{V|Z}^{(n)} \big| + \big| \mathcal{H}_{X|VZ}^{(n)} \big|  - H \big( A [ \mathcal{H}_{V|Z}^{(n)} ]  T [ \mathcal{H}_{X|VZ}^{(n)} ] \big| Z^n \big) \nonumber \\
& &  + 2 \delta_n^{(*)} (4 n  - \log \delta_n^{(*)}) \nonumber \\
& \stackrel{(c)}{\leq} &  2 n \delta_n  + 2 \delta_n^{(*)} (4 n  - \log \delta_n^{(*)})  , \nonumber
\end{IEEEeqnarray}
where $(a)$ holds by the definition of mutual information in terms of entropies and the uniformity of $\tilde{A}_i [ \mathcal{H}_{V|Z}^{(n)} ]$ and $\tilde{T}_i [ \mathcal{H}_{V|Z}^{(n)} ]$; $(b)$ holds by \eqref{eq:aux_s}; and $(c)$ holds because 
\begin{IEEEeqnarray}{rCl}
\IEEEeqnarraymulticol{3}{l}{
H \big(A [ \mathcal{H}_{V|Z}^{(n)} ]  T [ \mathcal{H}_{X|VZ}^{(n)} ] \big| Z^n \big) 
} \nonumber \\
& \geq & H \big(A [ \mathcal{H}_{V|Z}^{(n)} ]  \big| Z^n \big) + H \big( T [ \mathcal{H}_{X|VZ}^{(n)} ] \big| A^n Z^n \big) \nonumber \\
& \geq  & \sum_{j \in \mathcal{H}_{V|Z}^{(n)}} H \big(A (j) \big| A^{1:j-1} Z^n \big) \nonumber \\
& & +  \sum_{j \in \mathcal{H}_{X|VZ}^{(n)}} H \big(T (j) \big| T^{1:j-1} V^n Z^n \big) \nonumber \\
& \geq & \big| \mathcal{H}_{V|Z}^{(n)} \big| (1 - \delta_n) + \big| \mathcal{H}_{X|VZ}^{(n)} \big| (1 - \delta_n) \nonumber
\end{IEEEeqnarray}
where we have used the fact that conditioning does not increase entropy, the invertibility of $G_n$, and the definition of $\mathcal{H}_{V|Z}^{(n)}$ and $\mathcal{H}_{X|VZ}^{(n)}$ in \eqref{eq:HUZ} and \eqref{eq:HV-UTZ} respectively.

\section{Proof of Lemma~\ref{lemma:secrecy2}}\label{app:secrecy2}
We prove that all confidential messages and eavesdropper observations from blocks $1$ to $i-1$, that is, $(S_{1:L}, \tilde{Z}_{1:i-1}^n)$, are asymptotically statistically independent of eavesdropper observations $\tilde{Z}_{i}^n$ at block~$i$. 

For any $i \in [2,L]$ and sufficiently large $n$, we have
\begin{IEEEeqnarray}{rCl}
\IEEEeqnarraymulticol{3}{l}{
I \big(S_{1:L}\tilde{Z}_{1:i-1}^n ; \tilde{Z}_{1:i}^n \big)
} \nonumber \\
& = & I \big(S_{1:i} \tilde{Z}_{1:i-1}^n ;\tilde{Z}_{i}^n \big) + I \big(S_{i+1:L} ; \tilde{Z}_{i}^n \big|  S_{1:i} \tilde{Z}_{1:i-1}^n \big) \nonumber \\
& \stackrel{(a)}{=} & I \big(S_{1:i} \tilde{Z}_{1:i-1}^n ;\tilde{Z}_{i}^n \big) \nonumber \\
& \leq & I \big(S_{1:i} \tilde{Z}_{1:i-1}^n \Xi_{i-1}^{(V)}  \Lambda_{i-1}^{(X)}; \tilde{Z}_{i}^n  \big) \nonumber \\
& = & I \big(S_{i} \Xi_{i-1}^{(V)}  \Lambda_{i-1}^{(X)}; \tilde{Z}_{i}^n  \big) \nonumber \\
& & +  I \big(S_{1:i-1} \tilde{Z}_{1:i-1}^n ; \tilde{Z}_{i}^n \big| S_i \Xi_{i-1}^{(V)}  \Lambda_{i-1}^{(X)} \big) \nonumber \\
& \stackrel{(b)}{\leq} & \delta_n^{(\text{S})}  +  I \big(S_{1:i-1} \tilde{Z}_{1:i-1}^n ; \tilde{Z}_{i}^n \big| S_i \Xi_{i-1}^{(V)}  \Lambda_{i-1}^{(X)} \big) \nonumber \\
& \leq & \delta_n^{(\text{S})}  +  I \big(S_{1:i-1} \tilde{Z}_{1:i-1}^n ; \tilde{Z}_{i}^n W_{1,i} \big| S_i \Xi_{i-1}^{(V)}  \Lambda_{i-1}^{(X)} \big) \nonumber \\
& = & \delta_n^{(\text{S})}  +  I \big(S_{1:i-1} \tilde{Z}_{1:i-1}^n ; W_{1,i} \big| S_i \Xi_{i-1}^{(V)}  \Lambda_{i-1}^{(X)} \big) \nonumber \\
& & +  I \big(S_{1:i-1} \tilde{Z}_{1:i-1}^n ; \tilde{Z}_{i}^n  \big|  S_i \Xi_{i-1}^{(V)}  \Lambda_{i-1}^{(X)} W_{1,i} \big)  \nonumber \\
& \stackrel{(c)}{=} & \delta_n^{(\text{S})}  +  I \big(S_{1:i-1} \tilde{Z}_{1:i-1}^n ; W_{1,i} \big| S_i \Xi_{i-1}^{(V)}  \Lambda_{i-1}^{(X)} \big) \nonumber \\
& \leq & \delta_n^{(\text{S})}  +  I \big( \tilde{A}_{1:i-1}^n \tilde{Z}_{1:i-1}^n  ; W_{1,i} \big| S_i \Xi_{i-1}^{(V)}  \Lambda_{i-1}^{(X)} \big) \nonumber \\
& = &  \delta_n^{(\text{S})}  +  I \big( \tilde{A}_{1:i-1}^n ; W_{1,i}\big| S_i \Xi_{i-1}^{(V)}  \Lambda_{i-1}^{(X)} \big) \nonumber \\
&  & +  I \big( \tilde{Z}_{1:i-1}^n  ; W_{1,i} \big| \tilde{A}_{1:i-1}^n S_i \Xi_{i-1}^{(V)}  \Lambda_{i-1}^{(X)}  \big) \nonumber \\
& \stackrel{(d)}{=} & \delta_n^{(\text{S})}  +  I \big( \tilde{A}_{1:i-1}^n ; W_{1,i} \big| S_i \Xi_{i-1}^{(V)}  \Lambda_{i-1}^{(X)} \big) \nonumber \\
& \stackrel{(e)}{=} & \delta_n^{(\text{S})}  + I \big( \tilde{A}_{1:i-1}^n ; \bar{\Omega}_i^{(V)} \oplus  \kappa_{\Omega}^{(V)} \big| S_i \Xi_{i-1}^{(V)}  \Lambda_{i-1}^{(X)} \big)  \nonumber \\
& \stackrel{(f)}{=} & \delta_n^{(\text{S})} \nonumber 
\end{IEEEeqnarray}
where $(a)$ holds by independence between $S_{i+1:L}$ and any random variable from blocks $1$ to $i$; $(b)$ holds by Lemma~\ref{lemma:secrecy1} because $(S_i,\Xi_{i-1}^{(V)})$ is stored into $\tilde{A}_i[\mathcal{H}^{(n)}_{V|Z}]$ and $\Lambda_{i-1}^{(X)} = \tilde{T}_i[\mathcal{H}^{(n)}_{X|VZ}]$; $(c)$ follows from applying \emph{d-separation} \cite{pearl2009causality} over the Bayesian graph in Figure~\ref{fig:dependencies} to obtain that $\tilde{Z}_{i}^n$ and $(S_{1:i-1}, \tilde{Z}_{1:i-1}^n)$ are conditionally independent given $(S_i, \Xi_{i-1}^{(V)} , \Lambda_{i-1}^{(X)},W_{1,i})$; $(d)$ also follows from applying \emph{d-separation} to obtain that $W_{1,i}$ and $\tilde{Z}_{1:i-1}^n$ are conditionally independent given $(\tilde{A}_{1:i-1}^n, S_i, \Xi_{i-1}^{(V)} , \Lambda_{i-1}^{(X)})$; $(e)$ holds by definition; and $(f)$ holds because $\bar{\Omega}_i^{(V)}$ is independent of $(S_i, \Xi_{i-1}^{(V)} , \Lambda_{i-1}^{(X)})$ and any random variable from blocks $1$ to $i-2$ and, moreover, because by \emph{crypto-lemma} \cite{crypto1241} we have $\bar{\Omega}_i^{(V)} \oplus \kappa_{\Omega}^{(V)}$ independent of $\tilde{A}_{i-1}^n$.
\end{appendices}

\bibliographystyle{ieeetr}
\bibliography{bibliography.bib}
\end{document}